\documentclass[11pt]{article}
\usepackage{amssymb}
\usepackage{amsmath}
\usepackage{amsmath}
\usepackage{amsthm}
\usepackage{amsfonts}
\usepackage{graphicx}
\usepackage{enumerate}
\usepackage{comment}
\usepackage{bm}
\usepackage{bbm} 
\usepackage{dsfont}
\usepackage[authoryear,longnamesfirst]{natbib}
\usepackage{xcolor,colortbl}
\usepackage{xr}
\usepackage{multirow}
\usepackage{multicol}
\usepackage{tikz}
\numberwithin{equation}{section}
\newtheorem{theorem}{Theorem}

\newtheorem{lemma}{Lemma}
\theoremstyle{definition}
\newtheorem{assumption}{Assumption}

\newcommand{\one}{\cellcolor{gray!50}}
\newcommand{\two}{\cellcolor{gray!44}}
\newcommand{\thr}{\cellcolor{gray!38}}
\newcommand{\fou}{\cellcolor{gray!32}}
\newcommand{\fiv}{\cellcolor{gray!26}}
\newcommand{\six}{\cellcolor{gray!20}}
\newcommand{\sev}{\cellcolor{gray!15}}
\newcommand{\eig}{\cellcolor{gray!10}}
\newcommand{\nin}{\cellcolor{gray!5}}
\newcommand{\ten}{\cellcolor{gray!0}}

\renewcommand{\hat}{\widehat}
\renewcommand{\tilde}{\widetilde}

\newcommand{\calF}[0]{\mathcal{F}}

\renewcommand{\hat}{\widehat}
\renewcommand{\bar}{\overline}
\allowdisplaybreaks
\begin{document}

\title{Standard Errors for Two-Way Clustering \\with Serially Correlated Time Effects\thanks{\setlength{\baselineskip}{5.0mm}We benefited from useful comments by A. Colin Cameron and seminar participants at Essex, Kobe, LSE, Michigan State, North Carolina State, Singapore Management University, and UC Davis, and participants in AMES in East and South-East Asia 2022, Cemmap/SNU Workshop on Advances in Econometrics 2022,  CIREQ Montr\'eal Econometrics Conference 2022, and NAWM 2023. All the remaining errors are ours. Hansen thanks the National Science Foundation and Phipps Chair for research support. The Stata command is available to install by \texttt{ssc install xtregtwo}.}}
\author{
	Harold D. Chiang\thanks{\setlength{\baselineskip}{5.0mm}Harold D. Chiang: hdchiang@wisc.edu. Department of Economics, University of Wisconsin-Madison, William H. Sewell Social Science Building, 1180 Observatory Drive,	Madison, WI 53706-1393, USA\smallskip} 
	\qquad 
		Bruce E. Hansen\thanks{\setlength{\baselineskip}{5.0mm}Bruce E. Hansen: bruce.hansen@wisc.edu. Department of Economics, University of Wisconsin-Madison, William H. Sewell Social Science Building, 1180 Observatory Drive,	Madison, WI 53706-1393, USA\smallskip} 
	\qquad 
	Yuya Sasaki\thanks{\setlength{\baselineskip}{5.0mm}Yuya Sasaki: yuya.sasaki@vanderbilt.edu. Department of Economics, Vanderbilt University, VU Station B \#351819, 2301 Vanderbilt Place, Nashville, TN 37235-1819, USA\smallskip}
}
\date{}
\maketitle

\begin{abstract}\setlength{\baselineskip}{6.6mm}
We propose improved standard errors and an asymptotic distribution theory for two-way clustered panels.
Our proposed estimator and theory allow for arbitrary serial dependence in the common time effects, which is excluded by existing two-way methods, including the popular two-way cluster standard errors of \citet*{CGM2011} and the cluster bootstrap of \cite*{menzel2021bootstrap}. 
Our asymptotic distribution theory is the first which allows for this level of inter-dependence among the observations. 
Under weak regularity conditions, we demonstrate that the least squares estimator is asymptotically normal, our proposed variance estimator is consistent, and t-ratios are asymptotically standard normal, permitting conventional inference.
The main results extend to two-way fixed-effect models.
We present simulation evidence that confidence intervals constructed with our proposed standard errors obtain superior coverage performance relative to existing methods. 
We illustrate the relevance of the proposed method in an empirical application to a standard Fama-French three-factor regression.
\end{abstract}

\newpage

\section{Introduction}

A standard panel data set has observations double-indexed over firms\footnote{The index $i$ can refer to any entity, such as firms, individuals, or households. For simplicity we will refer to these entities as ``firms''.} $i$ and time $t$. A panel is said to have a two-way dependence structure if there is dependence across individuals at any given time, and across time for any given individual. A common model for two-way dependence is the components structure $U_{it} = f(\alpha_i,\gamma_t,\varepsilon_{it})$, where $\alpha_i$ is a firm effect, $\gamma_t$ is a time effect, and $\varepsilon_{it}$ is an idiosyncratic effect. It is typical to view the time effects $\gamma_t$ as omitted macroeconomic variables, such as the state of the business cycle. Therefore, they are unlikely to be serially independent. Consequently, it is reasonable to treat $\gamma_t$ as a serially correlated time-series process.

Serial correlation in the common time-effects, however, creates an extra layer of serial dependence beyond two-way dependence. It induces dependence among observations which do not share a common firm or time index. This fundamentally complicates the dependence structure, rendering existing theory and methods inappropriate. Most importantly for practice, existing two-way clustered inference methods do not allow time effects with arbitrary serial correlation.  Moreover, no formal asymptotic theory has previously been developed under the two-way clustering setting with serially correlated time effects. This paper sheds new light on literature of two-way clustering by formally investigating asymptotic theory for serially correlated time effects in this context and proposing novel and theoretically supported method of  inference under such settings.  

The most popular inference method for two-way dependent panels is the two-way clustered standard errors of \citet*{CGM2011}, which we shall henceforth refer to as CGM. A related recently method is the bootstraps of \cite*{menzel2021bootstrap}; see also \citet[Sec. 3.3]{DDG2019} for a generalization to empirical processes. Both of these approaches allow for the components structure $U_{it} = f(\alpha_i,\gamma_t,\varepsilon_{it})$, but only under the additional strong condition that the time effects $\gamma_t$ are serially independent. The CGM standard errors explicitly calculate the variance allowing for traditional two-way dependence, not allowing for dependence induced by serially correlated time effects. Consequently, these methods exclude, by construction, the possibility that the common time component $\gamma_t$ is an unmodelled macroeconomic effect. 

 An important intuitive extension due to \citet*{Thompson2011} allows $\gamma_t$ to be serially correlated up to a known fixed number of lags and suggests to estimate the asymptotic variance by including unweighted, lagged autocovariance estimates to the CGM estimator. This relaxes the CGM assumptions by allowing serial correlation structures that are of $m$-dependence. In practice, however, it is difficult to implement since the serial dependence structure is not known a priori. Also, even under this $m$-dependence setting, no asymptotic distribution theory was provided. This is particularly troubling since serial correlated time effects induces a complicated dependence structure and thus was unclear whether this inference procedure is theoretically justified and under which conditions it is so. Indeed, our simulations unveil that this unweighted, fixed number of lags approach shows unsatisfactory finite sample performances under various DGPs; see Section \ref{sec:simulations}.  In addition, based on his own simulations, \citet*{Thompson2011} recommends omitting the correction for serial correlation unless the time dimension is large. Thus in practice, the Thompson estimator actually implemented by most (if not all) empirical researchers reduces to the CGM two-way estimator.

Furthermore, an asymptotic distribution theory for regression with two-way clustering with general serial correlated time effects is missing. \citet*{CGM2011} assert an asymptotic theory for estimation, but do not examine the impact of two-way dependence, nor examine standard error estimation. As previously mentioned, \citet*{Thompson2011} does not provide a distribution theory even under the $m$-dependence setting. \cite*{DDG2019}, \cite{mackinnon2021wild}, and \cite*{menzel2021bootstrap} do provide rigorous theory, yet only for settings without serial dependence. 

Clustered inference can alternatively be based on unstructured one-way dependence (over either $i$ or $t$, but not both simultaneously) using the popular clustered variance estimator of \citet*{liang1986longitudinal} and \citet*{arellano1987computing}. These methods, however, cannot account for two-way dependence.
 An alternative framework is one-way-cluster dependence across $i$ with weak serial dependence across $t$ \citep{driscoll_kraay_1998}.
A yet alternative framework has been provided by  \citet{vogelsang2012} and \cite{hidalgo2021inference}, which study panels with cross-sectional and temporal dependence under a different set of conditions.
They allow dependence across firms and time, but the dependence between observations within a time period, as well as the dependence of a cross-sectional unit observed over time, both decay as observations get further apart in time and space. 
Consequently, these alternative frameworks do not allow arbitrary two-way clustering. 

Clustered standard errors have become ubiquitous in applied economic research, as evidenced by a perusal of current applied journals, and by the enormous citations to several of the above-mentioned papers. \citet*{petersen2009estimating} provides an excellent review of these popular methods and their use in empirical research through 2009. Our perusal of current applied journals reveals that nearly all applications use either Liang-Zeger-Arelleno one-way clustering or CGM two-way clustering. While \citet*{Thompson2011} is also highly cited, our review indicates that empirical applications do not employ his correction for correlated time effects, but rather use the simpler CGM two-way clustering.

In this article, we modify the CGM and Thompson two-way clustered standard error to accommodates time effects with arbitrary stationary serial dependence. Our approach allows for cluster dependence within individuals $i$, within time periods $t$, and allows the common time component $\gamma_t$ to be serially dependent of arbitrary order. Ours is the first approach which allows this complexity of two-way dependence. This is accomplished by a correction involving kernel smoothing over the autocorrelations, with the number of autocorrelation lags increasing with sample size. To select the lag truncation parameter, we propose a simple rule based on \citet{Andrews1991}.

We provide an asymptotic theory of inference under weak regularity conditions, including the assumption that the time effects $\gamma_t$ are strictly stationary and mixing. We show that the least squares estimator is asymptotically normal, our proposed variance estimator is consistent, and t-ratios are asymptotically standard normal, permitting conventional inference. The proofs of these results are far from trivial. For example, our consistency proof for our proposed cluster-robust variance estimator is nonstandard. We show that the problem can be re-written into a claim of bounding a fourth-order sum of cross-moments of dependent time series, for which a mixing bound due to \cite{yoshihara1976limiting} can be applied. Furthermore, the same proof uses novel projection arguments which simplify the derivations.  

We explore the performance of our proposed method in a simple simulation experiment which compares the coverage probability of confidence intervals constructed with six different standard error methods. We find that our proposed method has the best performance relative to the competitors in each simulation design considered, and in some settings the difference is substantial. 

We also illustrate the relevance of the method with an empirical application to estimation of the slope coefficients in a standard Fama-French three-factor regression using two panels of stock returns. We find that our proposed standard errors are different -- and larger -- than conventional standard errors, for four of six regression estimates examined. 

The Stata command is available to install by \texttt{ssc install xtregtwo}.

The rest of this paper is organized as follows.
Section \ref{sec:background} discusses two-way dependence with correlated time effects and provides an informal overview of the method with a practical guide.
Section \ref{sec:theory} presents the formal theoretical results.
Section \ref{sec:simulations} provides simulation evidence on the practical performance of our proposed method.
Section \ref{sec:application} presents an empirical application to a standard Fama-French regression.
The appendix collects a mathematical proof of the main result, auxiliary lemmas and their proofs, and additional details omitted from the main text. 

\section{Two-Way Dependence with Correlated Time Effects}\label{sec:background}

\subsection{Least Squares Estimation}\label{sec:simple_setting}

Let $\left(Y_{it},X_{it}'\right)$ be a panel of observations over $i=1,...,N$ and $t=1,...,T$, where $Y_{it}$ is real-valued and $X_{it}$ is a $k \times 1$ vector. The model is the linear regression equation
\begin{align}\label{eq:linear_regression}
Y_{it} =X_{it}'\beta + U_{it}
\end{align}
with
\begin{align}\label{eq:linear_error}
E[X_{it}U_{it}]=0.
\end{align}

The standard estimator for $\beta$ is least squares
\begin{align}\label{eq:beta_hat}
\hat \beta=\left(\sum_{i=1}^N\sum_{t=1}^T X_{it} X_{it}'\right)^{-1} \left(\sum_{i=1}^N\sum_{t=1}^T X_{it} Y_{it}\right).
\end{align}
The least squares residuals are $\hat U_{it}=Y_{it}-X_{it}'\hat \beta $.

We are interested in the variance of $\hat \beta$. It can be calculated explicitly under the auxiliary assumption that the regressors are fixed and the error is strictly exogenous.\footnote{The strict exogeneity condition is not required for our main theory and is used only for an illustration purpose in the current section for the exact variance calculations.} We only use this assumption to motivate our covariance matrix estimator, however, and will not be needed for our asymptotic distribution theory.

The variance of $\hat \beta$ can be written as follows. Define the firm sums $R_i=\sum_{t=1}^T X_{it} U_{it}$, the time sums $S_t=\sum_{i=1}^N X_{it} U_{it}$, and the cross-sums $G_m=\sum_{t=1}^{T-m} S_{t} S_{t+m}'$ and $H_m =\sum_{i=1}^N \sum_{t=1}^{T-m} X_{it} U_{it}X_{i,t+m}' U_{i,t+m}$. With a little algebra we obtain the following decomposition.
\begin{align}\label{eq:v_hat}
V_{NT}=var(\hat \beta)=\hat Q^{-1} \Omega_{NT} \hat Q^{-1}
\end{align}
where
\begin{align}\label{eq:Q_hat}
\hat Q =\frac{1}{NT}\sum_{i=1}^N\sum_{t=1}^T X_{it} X_{it}' 
\end{align}
and
\begin{align}
\Omega_{NT} & = \frac{1}{(NT)^2}\sum_{i=1}^N E\left[R_i R_i' \right] \label{eq:Variance1}\\
 & + \frac{1}{(NT)^2}\sum_{t=1}^T E\left[S_t S_t' \right] \label{eq:Variance2}\\
 & - \frac{1}{(NT)^2}\sum_{i=1}^N\sum_{t=1}^T E[X_{it}X_{it}'U_{it}^2] \label{eq:Variance3}\\
 & + \frac{1}{(NT)^2}\sum_{m=1}^{T-1} E\left[G_m +G_m' -H_m -H_m'\right].  \label{eq:Variance4}
\end{align}

The expression \eqref{eq:Variance1}-\eqref{eq:Variance4} decomposes the variance of the least squares estimator into four components: \eqref{eq:Variance1} is the variance of the firm sums; \eqref{eq:Variance2} is the variance of the time sums; \eqref{eq:Variance3} is a correction for double-counting of the common variance in \eqref{eq:Variance1} and \eqref{eq:Variance2}; and \eqref{eq:Variance4} is the autocovariances of the time sums, corrected for double-counting. 

\subsection{Variance Estimation}

Estimators of the variance matrix $V_{NT}$ take the general form
\begin{align}
\hat V_{NT} =\hat Q^{-1} \hat \Omega_{NT} \hat Q^{-1}\label{eq:V_hat}
\end{align}
where $\hat \Omega_{NT}$ is some estimator of $\Omega_{NT}$. Different estimators make distinct assumptions on the covariances in \eqref{eq:Variance1}-\eqref{eq:Variance4} which lead to distinct estimators for $\Omega_{NT}$ in \eqref{eq:V_hat}. The Liang-Zeger-Arellano one-way cluster estimator assumes that observations are independent across $i$, implying that \eqref{eq:Variance2}+\eqref{eq:Variance3}+\eqref{eq:Variance4} equals zero. The ``cluster within $t$'' estimator assumes that observations are independent across $t$, implying that \eqref{eq:Variance1}+\eqref{eq:Variance3}+\eqref{eq:Variance4} equals zero. The CGM two-way estimator assumes that observations $it$ and $js$ are independent if $i\ne j$ or $t\ne s$, implying that \eqref{eq:Variance4} equals zero. The respective estimators take the same form as the assumed non-zero expressions in \eqref{eq:Variance1}-\eqref{eq:Variance4}. For example, the CGM variance estimator of $\Omega_{NT}$ is
\begin{align}\label{eq:CGM_estimator}
\frac{1}{(NT)^2}\left(\sum_{i=1}^N \hat R_i \hat R_i'  + \sum_{t=1}^T \hat S_t \hat S_t'  - \sum_{i=1}^N\sum_{t=1}^T X_{it}X_{it}'\hat U_{it}^2 \right)  
\end{align}
where $\hat R_i=\sum_{t=1}^T X_{it}\hat U_{it}$ and $\hat S_t=\sum_{i=1}^N X_{it}\hat U_{it}$.

\citet*{Thompson2011} assumes that the across-firm autocovariances are non-zero for small lags $m$, but zero for lags beyond a known constant $M$. This implies that the sum over $m$ in \eqref{eq:Variance4} can be truncated above $m=M$. This motivates his estimator of $\Omega_{NT}$, which is \eqref{eq:CGM_estimator} plus
\begin{align*}
\frac{1}{(NT)^2}\sum_{m=1}^M \left( \hat G_m +\hat G_m' -\hat H_m -\hat H_m' \right) 
\end{align*}
where $\hat G_m=\sum_{t=1}^{T-m} \hat S_{t} \hat S_{t+m}'$ and $\hat H_m =\sum_{t=1}^{T-m}\sum_{i=1}^N  X_{it}\hat U_{it}X_{i,t+m}'\hat U_{i,t+m}$.
\citet*{Thompson2011} does not discuss selection of $M$, other than to indicate that it is known a priori. For his simulations and empirical applications he sets $M=2$, which we take to be his default choice.

We illustrate the dependence patterns assumed by the different estimators in Figure \ref{fig:illustration}. Each panel shows an array with each entry depicting firm/time pairs $(i,t)$, with the star $\star$ marking the reference point $(i,t)=(1,1)$, and dependence structures indicated by the grey shading. Panel (A) illustrates the case of independent observations (which corresponds to the unclustered Eicker-Huber-White estimator) where the observation $(i,t)=(1,1)$ is uncorrelated with all other observations. Panel (B) illustrates the case of independence across firms (which corresponds to the Liang-Zeger-Arellano one-way cluster estimator) where the observation $(1,1)$ is correlated with $(1,t)$ for $t>1$, but is uncorrelated with all other observations. Panel (C) similarly illustrates the case of independence across time (the ``cluster within $t$'' estimator). Panel (D) illustrates the case where the observation $(1,1)$ is correlated with $(1,t)$ for $t>1$ and with $(i,1)$ for $i>1$ (corresponding to the CGM two-way clustered estimator). Panel (E) illustrates the case where two-way clustering is augmented to allow dependence between $(1,1)$ and $(i,t)$ for all $t\le 3$. This corresponds to Thompson's estimator. Finally, panel (F) illustrates the case where observation $(1,1)$ is correlated with all other observations. The dark-to-light shading is meant to imply that the correlation between $(1,1)$ and $(i,t)$ is expected to diminish for $t>1$.

\begin{figure}[tb]
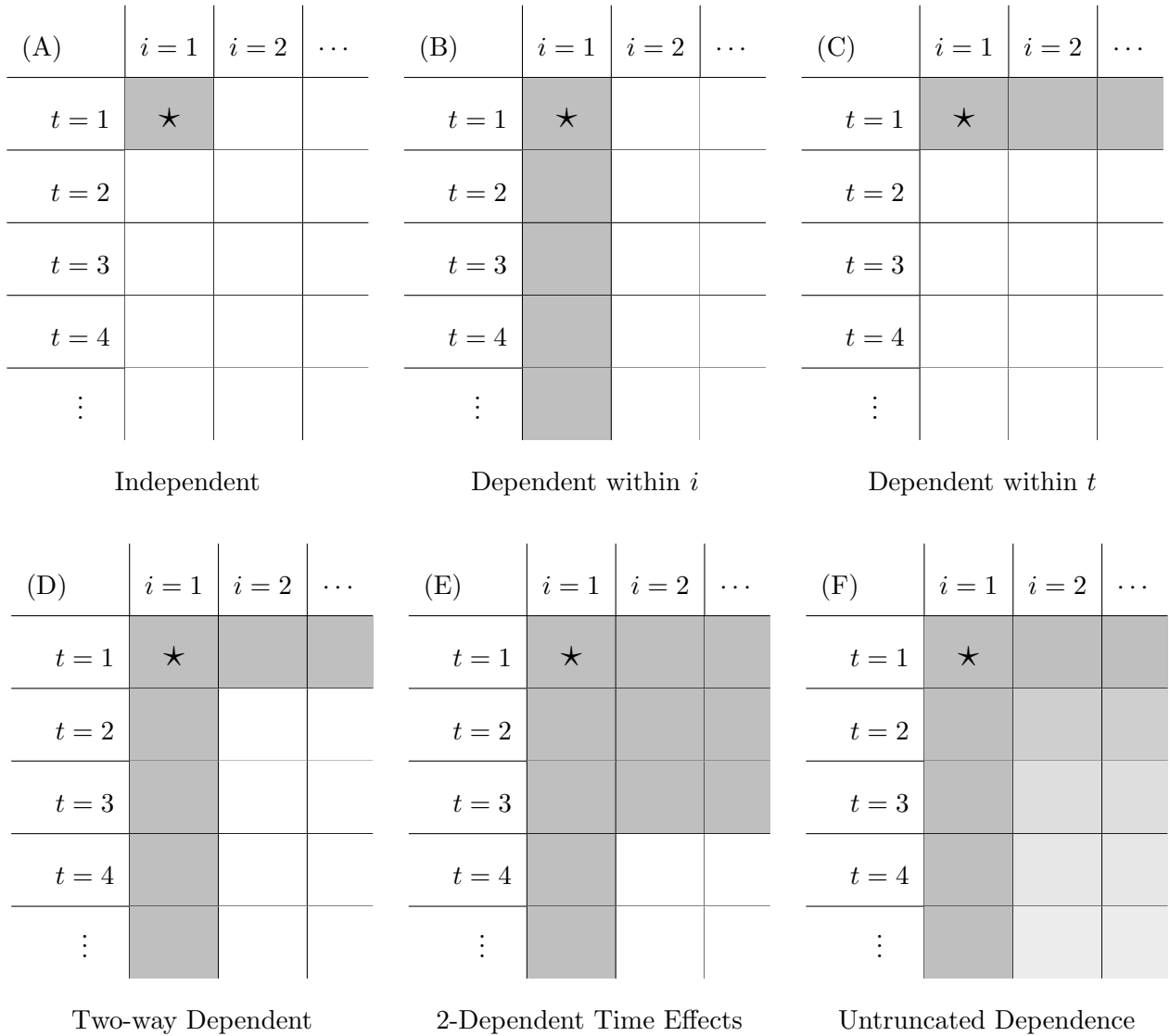

\begin{center}
\renewcommand{\arraystretch}{1.23}
\begin{tabular}{lc|c|c|c}
\multicolumn{2}{l|}{(A)}& $i=1$ & $i=2$ & $\cdots$ \\
\hline
&$t=1$    & \one \LARGE$\star$ & \ten & \ten \\
\hline
&$t=2$    & \ten  & \ten & \ten \\
\hline
&$t=3$    & \ten  & \ten & \ten \\
\hline
&$t=4$    & \ten  & \ten & \ten \\
\hline
&$\vdots$ & \ten  & \ten & \ten \\
\multicolumn{5}{c}{Independent}
\end{tabular}
\quad
\begin{tabular}{lc|c|c|c}
\multicolumn{2}{l|}{(B)}& $i=1$ & $i=2$ & $\cdots$ \\
\hline
&$t=1$    & \one \LARGE$\star$  & \ten & \ten \\
\hline
&$t=2$    & \one  & \ten & \ten \\
\hline
&$t=3$    & \one  & \ten & \ten \\
\hline
&$t=4$    & \one  & \ten & \ten \\
\hline
&$\vdots$ & \one  & \ten & \ten \\
\multicolumn{5}{c}{Dependent within $i$}
\end{tabular}
\quad
\begin{tabular}{lc|c|c|c}
\multicolumn{2}{l|}{(C)}& $i=1$ & $i=2$ & $\cdots$ \\
\hline
&$t=1$    & \one \LARGE$\star$  & \one & \one \\
\hline
&$t=2$    & \ten  & \ten & \ten \\
\hline
&$t=3$    & \ten  & \ten & \ten \\
\hline
&$t=4$    & \ten  & \ten & \ten \\
\hline
&$\vdots$ & \ten  & \ten & \ten \\
\multicolumn{5}{c}{Dependent within $t$}
\end{tabular}
\bigskip\\
\begin{tabular}{lc|c|c|c}
\multicolumn{2}{l|}{(D)}& $i=1$ & $i=2$ & $\cdots$ \\
\hline
&$t=1$    & \one \LARGE$\star$  & \one & \one \\
\hline
&$t=2$    & \one  & \ten & \ten \\
\hline
&$t=3$    & \one  & \ten & \ten \\
\hline
&$t=4$    & \one  & \ten & \ten \\
\hline
&$\vdots$ & \one  & \ten & \ten \\
\multicolumn{5}{c}{Two-way Dependent}
\end{tabular}
\quad
\begin{tabular}{lc|c|c|c}
\multicolumn{2}{l|}{(E)}& $i=1$ & $i=2$ & $\cdots$ \\
\hline
&$t=1$    & \one \LARGE$\star$  & \one & \one \\
\hline
&$t=2$    & \one  &  \one &  \one\\
\hline
&$t=3$    & \one  &  \one &  \one \\
\hline
&$t=4$    & \one  & \ten & \ten \\
\hline
&$\vdots$ & \one  & \ten & \ten \\
\multicolumn{5}{c}{2-Dependent Time Effects}
\end{tabular}
\quad
\begin{tabular}{lc|c|c|c}
\multicolumn{2}{l|}{(F)}& $i=1$ & $i=2$ & $\cdots$ \\
\hline
&$t=1$    & \one \LARGE$\star$  & \one & \one \\
\hline
&$t=2$    & \one  & \thr & \thr \\
\hline
&$t=3$    & \one  & \fiv & \fiv \\
\hline
&$t=4$    & \one  & \six & \six \\
\hline
&$\vdots$ & \one  & \sev & \sev \\
\multicolumn{5}{c}{Untruncated Dependence}
\end{tabular}
\end{center}
\caption{Clustered Dependence Structures.}
\label{fig:illustration}
\end{figure}

\subsection{Correlated Time Effects}\label{sec:time_effects}

To understand the source of cross-firm and cross-time dependence it is illuminating to consider the linear components model $Y_{it} = \alpha_i + \gamma_t +\varepsilon_{it}$ under the assumption of i.i.d. firm effects $\alpha_i$ and  idiosyncratic effects $\varepsilon_{it}$. If the time effect $\gamma_t$ is also i.i.d. then observations $(i,t)$ and $(j,s)$ are independent if $i\ne j$ and $t\ne s$. However, if $\gamma_t$ is serially dependent, then observations $(i,t)$ and $(j,s)$ can be dependent for arbitrary indices.

In most applications the time effect $\gamma_t$ is a proxy for omitted macroeconomic factors, and is therefore unlikely to be i.i.d. or have truncated serial dependence. Most macroeconomic variables have untruncated autocorrelation functions.

We empirically illustrate the importance of this feature. Consider two variables involved in a standard market value equation: log Tobin's average Q (market value divided by the stock of non-R\&D assets), and log of the R\&D stock (relative to the stock of non-R\&D assets). Panel regressions of the former on the latter have been the focus of \citet{griliches1981market} and many subsequent papers \citep*[e.g.,][]{hall2005market,bloom2013identifying,arora2021knowledge}. Using a panel of 727 firms for the years 1981--2001 from \citet*{bloom2013identifying} (see Appendix \ref{sec:estimation_gammax_gammay} for details) we estimated time effects for each series. Estimated time effects are plotted in the two panels of Figure \ref{fig:gamma}, with log R\&D stock on the left and log Tobin's Q on the right. The graphs reveal considerable serial correlation. Their estimated first-order autocorrelations are 0.425 (with a standard error of 0.003), and 0.467 (with a standard error of 0.007), respectively, which are quite large. Furthermore, the autocorrelations are strong at multiple lags as illustrated in their autocorrelograms, which are displayed in Figure \ref{fig:correlogram}. Together, this means that the time effects $\gamma_t$ for these series have substantial serial dependence, which is not well described by finite $M$-dependence. This implies that the dependence structures assumed by \citet*{CGM2011}, \cite*{menzel2021bootstrap}, and \citet*{Thompson2011} are incorrect, but rather need to be modified to allow for serial correlation of arbitrary order, as we propose in the next section.

\begin{figure}[tb]
	\centering
	\scalebox{0.22}{
	\begin{tabular}{cc}
	\includegraphics{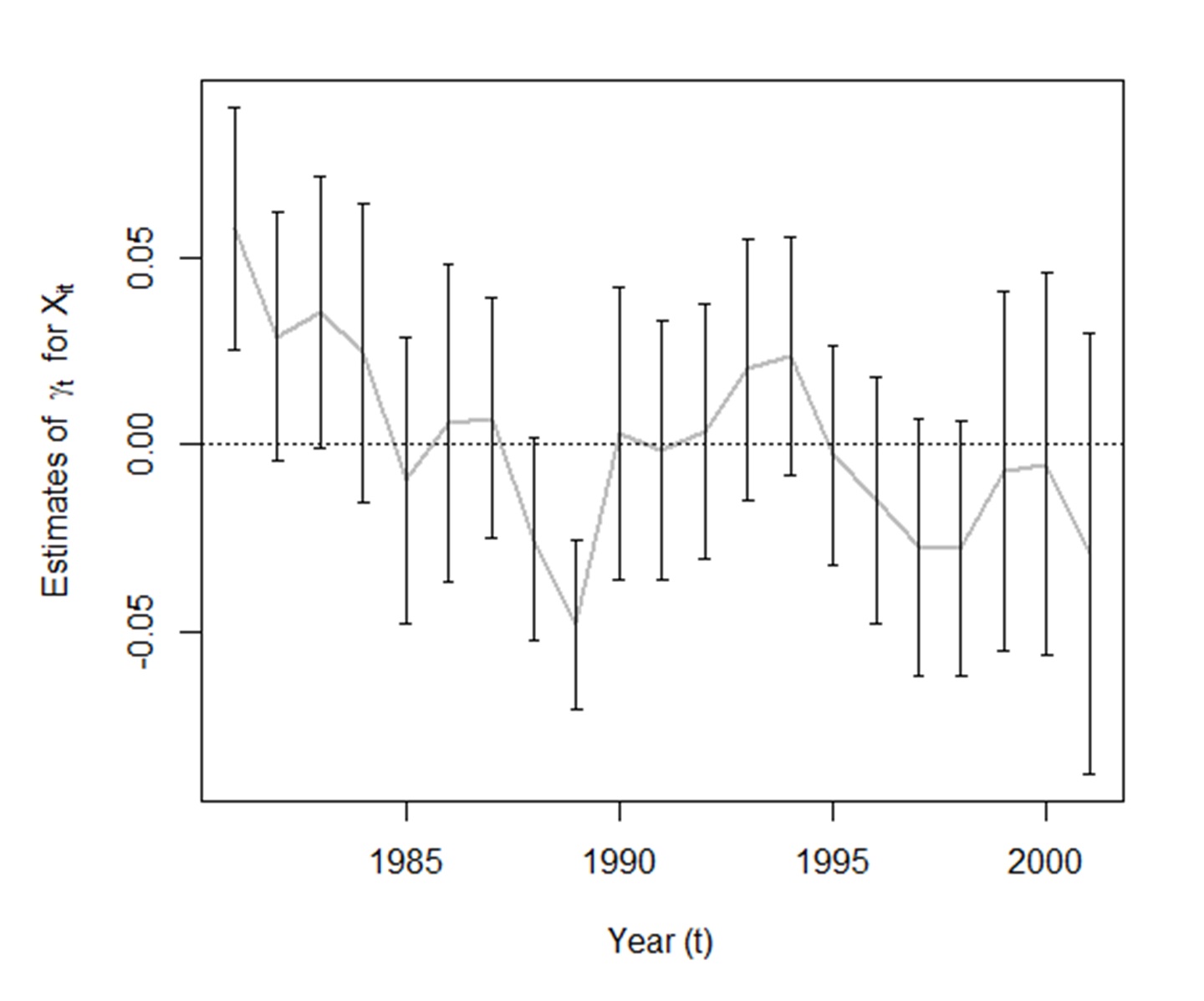}
	&
	\includegraphics{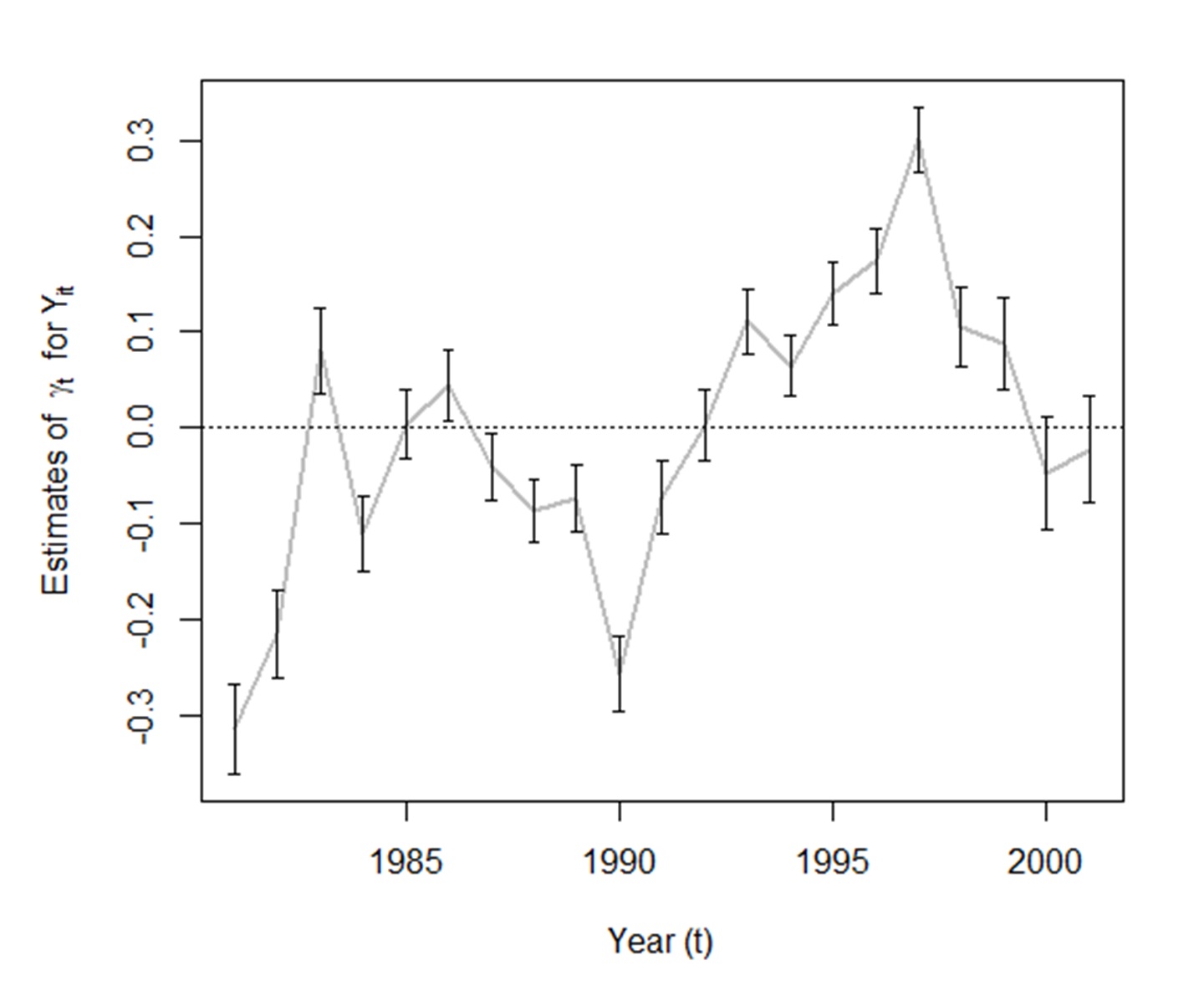}
	\end{tabular}
	}
	\caption{Estimates of the common time effects for log R\&D stock (left) and log Tobin's average Q (right). The vertical lines indicate pointwise 95\% confidence intervals.}${}$\\${}$
	\label{fig:gamma}
\end{figure}

\begin{figure}[tb]
	\centering
	\scalebox{0.22}{
	\begin{tabular}{cc}
	\includegraphics{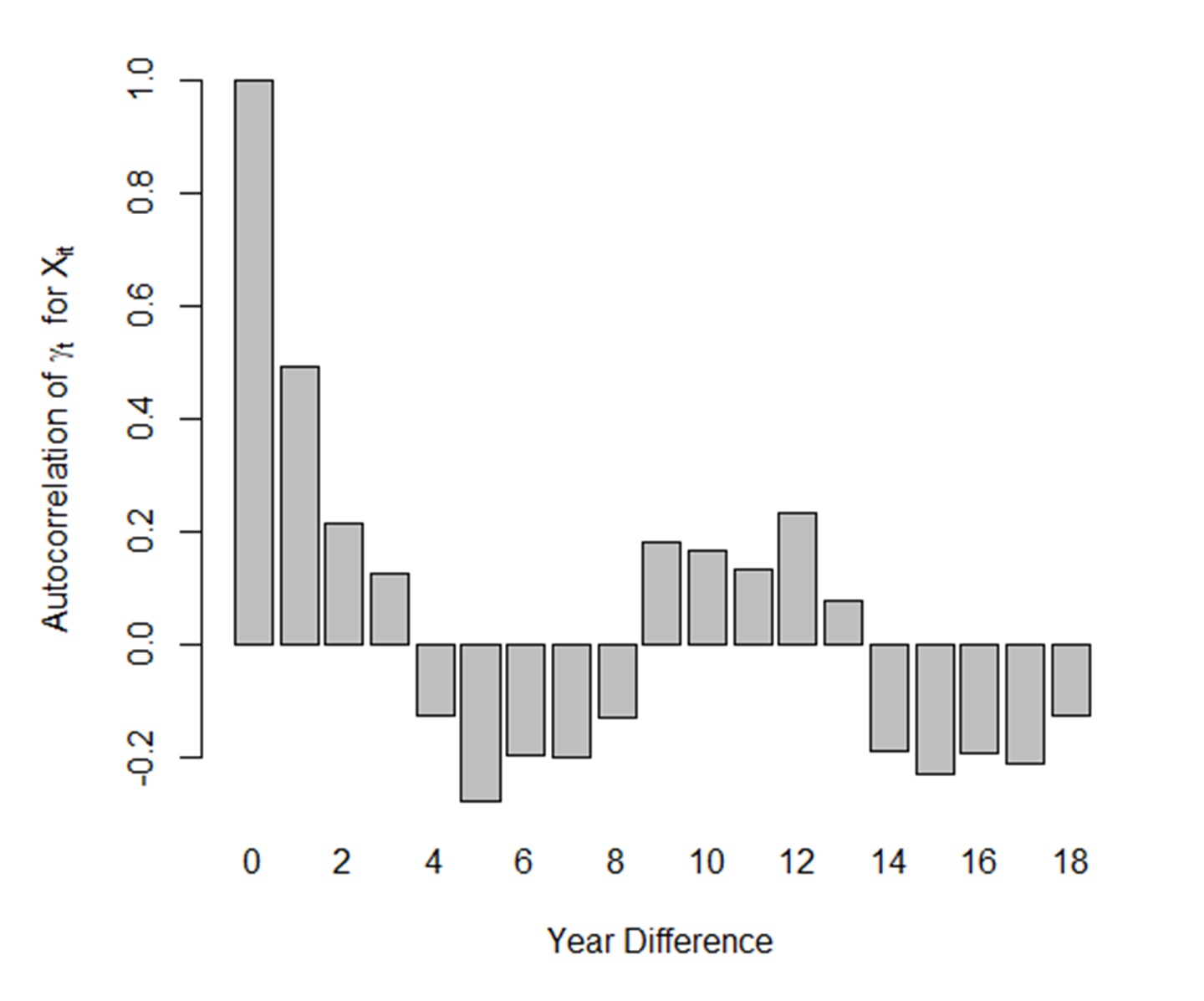}
	&
	\includegraphics{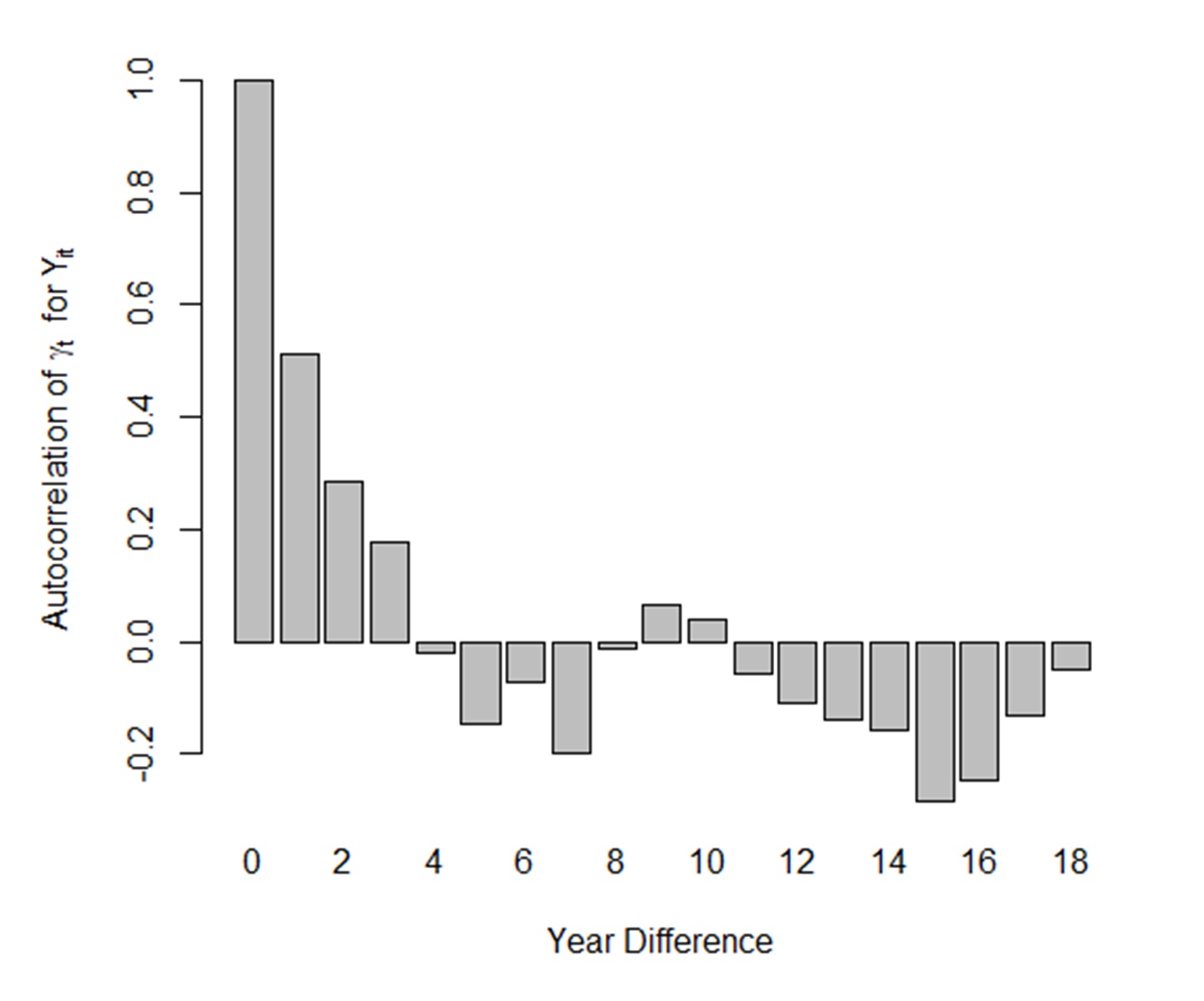}
	\end{tabular}
	}
	\caption{Autocorrelograms of the common time effects for log R\&D stock (left) and log Tobin's average Q (right).}${}$\\${}$
	\label{fig:correlogram}
\end{figure}

\subsection{Variance Estimation with Serially Correlated Time Effects}\label{sec:var_estimator}

As described in the previous section, serially correlated time effects $\gamma_t$ imply that the cross-firm autocorrelations $G_m$ in the variance decomposition \eqref{eq:Variance4} are non-zero at potentially any lag $m$. However, we cannot estimate these correlations well at all lags $m$ for fixed $T$, for the same reasons as arise in time-series variance estimation. Under the assumption that the time effects $\gamma_t$ are strictly stationary and weakly dependent (meaning that the autocorrelation function decays to zero) then it is sufficient to focus on the small lags $m$, using a weighted average of the terms in \eqref{eq:Variance4}, with the number of terms increasing with sample size. This motivates the following estimator of $\Omega_{NT}$
\begin{align}
\hat \Omega_{NT} = 
EVC&\left( 
\frac{1}{(NT)^2}\left(\sum_{i=1}^N \hat R_i \hat R_i'  + \sum_{t=1}^T \hat S_t \hat S_t'  - \sum_{i=1}^N\sum_{t=1}^T X_{it}X_{it}' \hat U_{it}^2\right) \right. \notag\\
 & + \left. \frac{1}{(NT)^2}\sum_{m=1}^M w(m,M)\left( \hat G_m +\hat G_m' -\hat H_m -\hat H_m' \right) \right)
\label{eq:M_estimator}
\end{align}
where $w(m,M)$ is a weight function and $EVC(\cdot)$ is an eigenvalue correction.\footnote{The $EVC(\cdot)$ replaces any negative eigenvalue by zero to ensure positive semidefiniteness.} Inserted into \eqref{eq:V_hat} we obtain our covariance estimator for $\hat \beta$. This variance estimator is a generalization of the \citet*{CGM2011}, \citet*{Thompson2011}, and \citet*{newey1987simple} estimators. The covariance matrix estimator $\hat V_{NT}$ can be multiplied by degree-of-freedom adjustments if desired, but we are unaware of any finite sample justification for a particular choice.
 The estimator simplifies to that of \citet*{CGM2011} when $M=0$, and to that of \citet*{Thompson2011} when $M=2$ and $w(m,M)=1$. Taking the square root of the diagonal elements of $\hat V_{NT}$ yields standard errors for the elements of $\hat\beta$.

The estimator \eqref{eq:M_estimator} depends on the choice of weights $w(m,M)$. Standard choices include the uniform (truncated) weights $w(m,M)=1$, and the triangular (Newey-West) weights $w(m,M)=1-m/(M+1)$, the latter popularized by \citet*{newey1987simple} for time-series data. We recommend the triangular weights. One advantage of this choice for time series applications as emphasized by \citet*{newey1987simple} is that this ensures a non-negative variance estimator. Unfortunately, the clustered estimator \eqref{eq:M_estimator} is not necessarily non-negative, even for $M=0$, as observed by \citet*{CGM2011}. However, the estimator \eqref{eq:M_estimator} is considerably less likely to be negative when the weights are triangular than uniform, which is an important practical advantage. 

The variance estimator \eqref{eq:M_estimator} critically depends on the number $M$, which is often called the lag truncation parameter. It is useful to note that $M$ does not need to be integer-valued. The choice of $M$ leads to a bias/precision trade-off, with larger values of $M$ leading to less bias in the estimator $\hat \Omega_{NT}$ of $\Omega_{NT}$, but less precision. In principle it is desirable to use a larger value of $M$ when the errors $U_{it}$ are highly serially correlated, and a smaller value of $M$ otherwise, but the extent of serial correlation is generally unknown, leading to the need for an empirical-based choice of $M$.

In the context of time-series variance estimation \citet{Andrews1991} proposed a data-driven choice of $M$ based on minimizing the asymptotic mean square error of the variance estimator, which is equivalent to the expression in \eqref{eq:M_estimator}. We can therefore apply his method for selection of $M$, treating the time-sums of the regression scores as time-series observations. Andrews' formula depends on the specific choice of weight function; we assume triangular weights. 

For $j=1,...,k$, let $X_{jit}$ be the $j$th element of $X_{it}$. Define the time-sums $S_{jt}=\sum_{i=1}^n X_{jit}\hat U_{it}$ of the regression scores. Fit by least squares the AR(1) equations $S_{jt} = \hat \rho _j S_{j,t-1}+\hat e_{jt}$. The Andrews rule\footnote{This is calculated from Andrews' equation (6.4), setting his weights $w_a$ to equal the inverse squared variances of the estimated AR(1) processes, which is appropriate for least squares estimation.} for the lag truncation $M$ is
\begin{align}\label{eq:m_hat}
\hat M  = 1.8171 \cdot \left(\frac{\sum_{j=1}^k \frac{\hat \rho _j^2}{\left(1-\hat \rho _j \right)^4}}{\sum_{j=1}^k \frac{\left(1-\hat \rho _j ^2\right)^2}{\left(1-\hat \rho _j \right)^4}} \right)^{1/3} T^{1/3}.
\end{align}
For the case of a scalar regressor this simplifies to
\begin{align*}
\hat M  = 1.8171 \cdot \left(\frac{\hat \rho^2}{\left(1-\hat \rho^2\right)^2}\right)^{1/3} T^{1/3}.
\end{align*}

For an even simpler choice, \citet{stock_watson} suggested the following rule-of-thumb. Setting $\rho = 0.25$ in the above formula (which occurs in a regression when both the regressor and regression error are AR(1) processes with AR(1) coefficients 0.5) the Andrews rule simplifies to $M =  0.75 \cdot T^{1/3}$. For example, for $T=50$, 100, and 200, respectively, the Stock-Watson rule is $M=2.7$, 3.5, and 4.4, respectively. The Stock-Watson can be used in place of the Andrews rule \eqref{eq:m_hat} if desired. 

\section{Econometric Theory}\label{sec:theory}

Consider a panel $\{D_{it} : 1 \le i \le N ; 1 \le t \le T \}$ of random vectors consisting of observed and/or unobserved variables that are relevant to the data generating process of the researcher's interest.
For instance, in the linear regression model presented in Section \ref{sec:background}, set $D_{it} = (Y_{it},X_{it}',U_{it})'$.
With a Borel-measurable function $f$ (generally unknown to the researcher), we consider the framework of panel dependence in $D_{it}$ generated through the stationary process
\begin{equation}\label{eq:generalized_ahk}
D_{it} = f(\alpha_i,\gamma_t,\varepsilon_{it}),
\end{equation}
where $\alpha_i$, $\gamma_t$, and $\varepsilon_{it}$ are random vectors of arbitrary dimension, with the sequences
$\{\alpha_i\}$, $\{\gamma_t\}$ and $\{\varepsilon_{it}\}$ mutually independent,
$\alpha_i$ is i.i.d. across $i$, and $\varepsilon_{it}$ is i.i.d. across $(i,t)$.\footnote{\label{foot:iid}These i.i.d. conditions may be relaxed in some ways, but we leave it for future research. The existing literature on two-way clustering explicitly or implicitly assumes these i.i.d. conditions, and we continue to impose them in this paper. Our focus, therefore, is to relax the i.i.d. assumption on the $\gamma_t$ component. One way to relax the i.i.d. conditions on the other components is to impose an MDS-type condition, in which case our main results remain to hold. Another is to allow for spatial mixing provided a researcher has spatial information associated with panel data. Also, see Section \ref{sec:summary}.} The sequence $\gamma_t$ is a strictly stationary serially correlated process.

The representation \eqref{eq:generalized_ahk} generalizes the Aldous-Hoover-Kallenberg (AHK) representation \citep{Kallenberg2006} which has been widely used for two-way clustering theory. See \citep[e.g.,][]{DDG2019,mackinnon2021wild,menzel2021bootstrap}. 
Indeed, it has been argued that the AHK representation is a natural modelling framework for two-way clustered data \citep*{mackinnon2021wild}.\footnote{\citet*{mackinnon2021wild} state ``[a] natural stochastic framework for the regression model with multiway clustered data is that of separately exchangeable random variables.'' Since separately exchangeable random variables may be represented by the AHK \citep[cf.][]{Kallenberg2006}, we make this assertion.}
A limitation of the AHK representation is that the time effects $\gamma_t$ are mutually independent. We relax this assumption, by directly assuming that \eqref{eq:generalized_ahk} holds, allowing $\gamma_t$ to be serially dependent. This is a strict generalization of the AHK representation.

In this section we provide an asymptotic distribution theory for the least squares estimator, our proposed covariance matrix estimator, and associated test statistics.
We start in Section \ref{sec:main} by examining a multivariate mean, followed in Section \ref{sec:ols} with linear regression. 

\subsection{Estimation of the Mean}\label{sec:main}
In this subsection, we focus on estimation of a multivariate mean. Let $X_{it}$ be an $m\times 1$ random vector satisfying equation \eqref{eq:generalized_ahk} for $D_{it}=X_{it}$. The standard estimator of the population mean $\theta =E[X_{it}]$ is the sample mean $\hat\theta = (NT)^{-1} \sum_{i=1}^N \sum_{t=1}^T X_{it}$. 

To obtain an asymptotic representation for $\hat\theta$, we use a Hoeffding-type decomposition. For simplicity and without loss of generality assume that $E[X_{it}]=0$. Define the random vectors $a_i=E[ X_{it}\mid \alpha_i]$, $b_t=E[ X_{it}\mid \gamma_t]$, and $e_{it}= X_{it}-a_i-b_t$. This gives rise to the following construct:
\begin{align}
\ X_{it} = a_i+b_t+e_{it}.
\label{eq:Xabe}
\end{align}
This expresses $X_{it}$ as a linear function of a firm effect $a_i$, time effect $b_t$, and error $e_{it}$. However, as \eqref{eq:Xabe} is a derived relationship, the error $e_{it}$ is not (in general) i.i.d. The decomposition has the following properties. The random vectors $a_i$ and $b_t$ are independent since they are each functions of the independent sequences $\{\alpha_i\}$ and $\{\gamma_t\}$. The sequence $\{a_i\}$ is i.i.d., and the sequence $\{b_t\}$ is strictly stationary. By iterated expectations we deduce the following: (1) $a_i$, $b_t$, and $e_{it}$ are each mean zero; (2) $E[e_{jt} \mid \alpha_i]=0$ and $E[e_{is} \mid \gamma_t]=0$ for any $i$, $j$, $t$, and $s$; (3) $E[ a_i e_{jt}']=0$ and $E[ b_t e_{is}']=0$ for any $i$, $j$, $t$, and $s$; (4) the sequences $\{a_i\}$, $\{b_t\}$, and $\{e_{it}\}$ are mutually uncorrelated; (5) conditional on $(\gamma_t,\gamma_s)$, $e_{it}$ and $e_{js}$ are independent for $j \ne i$.

Taking averages we find that
\begin{align}
\hat\theta = \frac{1}{N} \sum_{i=1}^N a_i+ \frac{1}{T} \sum_{t=1}^T b_t+ \frac{1}{NT} \sum_{i=1}^N \sum_{t=1}^T e_{it}.
\label{eq:thetahat}
\end{align}
The uncorrelatedness of the sequences implies that the three sums are uncorrelated. Hence the variance of $\hat\theta$ equals the sum of the variance matrices of the three components. The variance of the first component equals $N^{-1}$ times the variance matrix of $a_i$, the variance of the second component approximately equals $T^{-1}$ times the long-run variance matrix of $b_t$ (since the latter is serially correlated), and the variance of the third component approximately equals $(NT)^{-1}$ times the long-run variance matrix of $e_t$. The technical details are deferred to Appendix \ref{sec:thm:variance_mean}. 

We now present sufficient conditions for this decomposition to be valid. 

\begin{assumption}\label{a:NW}
For some $r>1$ and $\delta>0$, (i) $X_{it} = f(\alpha_i,\gamma_t,\varepsilon_{it})$ where $\{\alpha_i\}$, $\{\gamma_t\}$, and $\{\varepsilon_{it}\}$ are mutually independent sequences,
$\alpha_i$ is i.i.d across $i$, $\varepsilon_{it}$ is i.i.d across $(i,t)$, and $\gamma_t$ is strictly stationary.
(ii) $E[||X_{it}||^{4(r+\delta)}]<\infty$.
(iii) 
$\gamma_t$ is an $\alpha$-mixing sequence with size $2r/(r-1)$, that is, $\alpha(\ell)=O(\ell^{-\lambda})$ for a $\lambda>2r/(r-1)$.

\end{assumption}

Assumption \ref{a:NW} (i) assumes that the observed random vectors are generated following the nonlinear, nonseparable, factor structure \eqref{eq:generalized_ahk}. Assumption \ref{a:NW} (ii) imposes moment conditions. Assumption \ref{a:NW} (iii) imposes weak dependence on the time effects.
The moment and mixing conditions here are standard in time-series theory, including that of \cite{newey1987simple} and \citet*{hansen1992consistent}. 

Define the variance matrices:
\begin{align}
\Sigma_a &= E[a_ia_i'] \label{eq:sigma_a}\\
\Sigma_b &= \sum_{\ell=-\infty}^{\infty} E[b_tb_{t+\ell}'] \label{eq:sigma_b}\\
\Sigma_e &= \sum_{\ell=-\infty}^{\infty} E[e_{it}e_{i,t+\ell}'] \label{eq:sigma_e}
\end{align}
which are independent of $i$ and $t$. 
Given the decomposition \eqref{eq:Xabe}, we can write the variance of the sample mean as a weighted sum of the variance components \eqref{eq:sigma_a}-\eqref{eq:sigma_e}.

\begin{theorem}\label{thm:variance_mean}
	Suppose that Assumption \ref{a:NW} holds. Then $||\Sigma_a||<\infty$, $||\Sigma_b||<\infty$, and $||\Sigma_e||<\infty$, and as $(N,T)\to\infty$,
\begin{align*}
	var( \hat \theta )
	&=\frac{1}{N}\Sigma_a+\frac{1}{T}\Sigma_b\left(1+o(1)\right)+\frac{1}{NT}\Sigma_e\left(1+o(1)\right).
\end{align*}
Furthermore, $\hat \theta \stackrel{p}{\to} \theta$ as $N,T\to\infty$.
\end{theorem}

A proof is provided in Appendix \ref{sec:thm:variance_mean}. Theorem \ref{thm:variance_mean} shows that the asymptotic variance of the sample mean depends on three components, inversely proportional to the number of firms $N$, time dimension $T$, and their product $NT$. When either $\Sigma_a>0$ or $\Sigma_b>0$ (which occurs when there is a non-degenerate firm or time effect) then the third term in the asymptotic variance is of lower stochastic order. However, in the special case where $X_{it}$ is i.i.d., then $\Sigma_a=0$ and $\Sigma_b=0$ so the first two terms equal zero, the long-run variance of $e_{it}$ simplifies to $\Sigma_e=var(X_{it})$, and the variance expression simplifies to $var( \hat \theta )=var(X_{it})/NT$. Consequently, the rate of convergence of the sample mean depends on the cluster structure.

For our distribution theory we require the following additional condition.

\begin{assumption}\label{a:non-degenerate}
One of the the following two conditions holds.
\\
(i) Either $\Sigma_a>0$ or  $\Sigma_b>0$, and $N/T\to c\in (0,\infty)$ as $(N,T)\to\infty$. \\
{\large or}
\\
(ii)
$X_{it}$ are independent and identically distributed across $i$ and $t$, and $var(X_{it})>0$.
\end{assumption}

Assumption \ref{a:non-degenerate} (i) requires the presence of at least one-way clustering. 
This assumption is analogous to the positive definiteness condition in \citet[Propositions 4.1--4.2]{DDG2019} and equation (16) of \citet*{mackinnon2021wild}.
Our results will continue to hold even if $N$ and $T$ diverge at different rates, but we make the homogeneous rate assumption for ease of exposition.
On the other hand, our results will not hold under fixed $N$ or fixed $T$.
Assumption \ref{a:non-degenerate} (ii) is the contrary case of no clustering. As we show below, our results hold under either condition. While Assumption \ref{a:non-degenerate} is sufficient for our results, it is probably stronger than necessary, but is used for its simplicity and tractability. Assumption \ref{a:non-degenerate} does rule out possible scenarios, including cases which lead to non-Gaussian limit distributions \citep[e.g.,][Example 1.7]{menzel2021bootstrap} -- see our discussion in Section \ref{sec:summary}.
The non-Gaussian cases can be characterized by daga generating processes that consist of degenerate additive factors of $i$, degenerate additive factors of $t$, and a small number of interactive factors between $i$ and $t$ \citep{chiang2022standard}.
With this said, it is legitimate to rule out non-Gaussian cases as our focus is on standard errors (as in the title of this article), which would not make sense under non-Gaussian limit distributions.

\begin{theorem}\label{thm:main_asymptotics}
	Suppose that Assumptions \ref{a:NW} and \ref{a:non-degenerate} hold. Then,
\begin{align*}
	var( \hat \theta )^{-1/2}(\hat \theta - \theta)\stackrel{d}{\to} N(0,I_m).
\end{align*}
\end{theorem}

A proof is provided in Appendix \ref{sec:thm:main_asymptotics}.
Theorem \ref{thm:main_asymptotics} shows that the self-normalized sample mean is asymptotically normal. Self-normalization is used to allow for differing rates of convergence due to clustering structure. 

In this section, we focus on the setting where there is precisely one unit of observation per cluster intersection.
In applications, there may be heterogeneous per-cluster numbers of observations, e.g., unbalanced panels.
The above theory will straightforwardly extend to such cases -- see Appendix \ref{sec:multiple_observations}.

\subsection{Linear Regression}\label{sec:ols}

We now revisit the linear panel regression model \eqref{eq:linear_regression}--\eqref{eq:linear_error} and the OLS estimator \eqref{eq:beta_hat}.
In this subsection we set $D_{it} = (Y_{it},X_{it}',U_{it})'$ , so that the framework \eqref{eq:generalized_ahk} is
\begin{equation}\label{eq:DGP_OLS}
(Y_{it},X_{it}',U_{it})' = f(\alpha_i,\gamma_t,\varepsilon_{it}).
\end{equation}
Set $a_i=E[X_{it}U_{it}\mid \alpha_i]$ and $b_t=E[X_{it}U_{it}\mid \gamma_t]$. Let $\Sigma_a$ and $\Sigma_b$ be the variance/long-run variance matrices of $a_i$ and $b_t$, respectively.

\begin{assumption}\label{a:NW_OLS}
For some $\delta>0$ and $r>1$, (i) $\{(Y_{it},X_{it}',U_{it}): 1\le i\le N, 1\le t\le T\}$ are generated following \eqref{eq:DGP_OLS}, where $\{\alpha_i\}$, $\{\gamma_t\}$, and $\{\varepsilon_{it}\}$ are mutually independent sequences,
$\alpha_i$ is i.i.d across $i$, $\varepsilon_{it}$ is i.i.d across $(i,t)$, and $\gamma_t$ is strictly stationary.
(ii) $Q=E[X_{it}X_{it}']>0$, $E[\|X_{it}\|^{8(r+\delta)}]<\infty$, and $E[\|U_{it}\|^{8(r+\delta)}]<\infty$. 
(iii)  $\gamma_t$ is a $\beta$-mixing sequence with size $2r/(r-1)$, that is, $\beta(\ell)=O(\ell^{-\lambda})$ for a $\lambda>2r/(r-1)$.
(iv) One of the following two conditions hold: (1) Either $\Sigma_a>0$ or $\Sigma_b>0$, and  $N/T\to c\in (0,\infty)$ as $(N,T)\to\infty$; or (2) $(X_{it},U_{it})$ are independent and identically distributed across $i$ and $t$, and $var(X_{it}U_{it})>0$.
(v) For each $M\ge 1$ and $1\le m\le M$,  $w(m,M)=1-[m/(M+1)]$.
(vi) $M/\min\{N,T\}^{1/2}=o(1)$.
\end{assumption}

Assumptions \ref{a:NW_OLS} (i)--(v) above are the counterparts of Assumptions \ref{a:NW} and \ref{a:non-degenerate}, extended to the regression model. The moment and mixing conditions are standard in time series regression.
Assumptions \ref{a:NW_OLS} (v)--(vi) are needed for consistent variance estimation.
The $\alpha$-mixing condition of Assumption \ref{a:NW} has been strengthened to $\beta$-mixing in Assumption \ref{a:NW_OLS}. This is because our proof of consistent variance estimation relies on a deep fourth-order summability result due to \cite{yoshihara1976limiting} which relies on $\beta$-mixing. 

Under these assumptions, the asymptotic variance of $\hat \beta$ is
\begin{align*}
V_{NT}=Q^{-1}\Omega_{NT}Q^{-1}
\end{align*}
where $\Omega_{NT}$ is defined in \eqref{eq:Variance1}-\eqref{eq:Variance4}. Our proposed estimator of $V_{NT}$ is \eqref{eq:V_hat} with \eqref{eq:M_estimator}.

For our theory we focus on a vector-valued parameter $\theta=R'\beta$ for some $k \times m$ matrix $R$. This includes individual coefficients when $m=1$. The estimator of $\theta$ is $\hat \theta=R'\hat \beta$, its asymptotic variance is $\Sigma_{NT}=R'V_{NT}R$, with estimator $\hat\Sigma_{NT}=R'\hat V_{NT}R$. For the case $m=1$, a standard error for $\hat \theta$ is $\hat\sigma_{NT}=\sqrt{R'\hat V_{NT}R}$. We now establish consistency of our variance estimator

\begin{theorem}\label{thm:var_consistency}
	If Assumption \ref{a:NW_OLS} holds for model \eqref{eq:linear_regression}--\eqref{eq:linear_error}, then 
\begin{align*}
\Sigma_{NT}^{-1}\hat \Sigma_{NT} &\stackrel{p}{\to} I_k.
\end{align*}
\end{theorem}

A proof is provided in Appendix \ref{sec:var_consistency}. It relies on some technical lemmas in Appendix \ref{sec:lemmas} that are of potential independent interest.
This result shows that our proposed variance estimator is consistent. Notice that we state consistency as a self-normalized matrix ratio. This allow for the differing rates of convergence covered by Assumption \ref{a:NW_OLS}.

Theorem \ref{thm:var_consistency} is new. It is the first demonstration of consistent variance estimation under two-way clustering with serially dependent time effects. 

The proof of Theorem \ref{thm:var_consistency} includes some technical innovations. Of particular note is the use of the fourth-order summability condition of \cite{yoshihara1976limiting}, combined with projections on the individual-specific and time-specific factors. The summability condition is needed to calculate the variance  of the variance estimator, which is a fourth-order sum.

\begin{theorem}\label{thm:OLS}
	If Assumption \ref{a:NW_OLS} holds for the model \eqref{eq:linear_regression}--\eqref{eq:linear_error}, then 
\begin{align}
\Sigma_{NT}^{-1/2}(\hat \theta - \theta) &\stackrel{d}{\to} N(0,I_m)
\label{eq:thetadist1}
\end{align}
and
\begin{align}
\hat \Sigma_{NT}^{-1/2}(\hat \theta - \theta) &\stackrel{d}{\to} N(0,I_m).
\label{eq:thetadist2}
\end{align}
\end{theorem}

A proof is provided in Appendix \ref{sec:cor_OLS}.

Theorem \ref{thm:OLS} shows that the least squares estimator $\hat \theta$ is asymptotically normal. Normality holds when the estimator is standardized by its asymptotic covariance matrix $\Sigma_{NT}$ or by its estimator $\hat\Sigma_{NT}$. This latter result shows (for the case $m=1$) that t-ratios constructed with our standard errors are asymptotically $N(0,1)$. It also shows (for the case $m>1$) that Wald-type tests constructed with our covariance matrix estimator are asymptotically $\chi_m^2$. Hence conventional inference methods can be used with the least squares estimator $\hat \beta$, our variance estimator $\hat V_{NT}$, and our standard errors.

Theorem \ref{thm:OLS} is new. It is the first result which rigorously demonstrates asymptotic normality of least squares estimators and t-ratios under two-way clustering with serially correlated time effects.  The asymptotic normality presented here adapts to the unknown convergence rate. It is, however, worthy to note that this result is pointwise in DGP. In the absence of correlated time effects, \cite{menzel2021bootstrap} discusses the issues of uniform inference. In a linear panel data context, \cite{lu2022uniform} provides uniformly valid inference procedure. Although it remains unclear to us whether their approaches can be adapted to our framework, it is an interesting future research avenue to investigate the potential uniformity properties of the asymptotics under various sets of sequences of DGPs. 

In contrast, test statistics constructed with the popular CGM variance estimator will not have conventional asymptotic distributions when the time effects $\gamma_t$ are serially correlated. The CGM variance estimator is inconsistent in this situation, so test statistics will have distorted asymptotic distributions.

\section{Fixed-Effect Models}\label{sec:fixed_effect}

For panel data, researchers often include one-way or two-way fixed effects.
This section has two contents regarding fixed-effect models.
First, Section \ref{sec:fixed_effect_necessary} argues that two-way clustering is still necessary in general even if a researcher includes two-way fixed effects.
Second, Section \ref{sec:fixed_effect_theory} extends our theory to two-way fixed-effect regressions.

\subsection{Two-Way Clustering Is Still Necessary}\label{sec:fixed_effect_necessary}

Some empirical economists believe that it is unnecessary to cluster standard errors if fixed effects are included in estimation.
In this section, we argue that fixed effects will not generally solve the problem of two-way cluster dependence. 

Consider the two-way fixed-effect model:
\begin{align}
Y_{it} &= \beta_0 + \beta_1 X_{it} + U_{it}, \text{ where } \notag\\
X_{it} &= \alpha_{i1} \gamma_{t2} + \alpha_{i2} \gamma_{t1} + \varepsilon_{it0} \label{eq:FE_example}\\
U_{it} &= \alpha_{i0} + \gamma_{t0} + \alpha_{i1} \gamma_{t3} + \alpha_{i3} \gamma_{t1} + \varepsilon_{it1} \notag
\end{align}
Note that $X_{it}$ and $U_{it}$ are generated by the latent variables $\alpha_i=(\alpha_{i0},\alpha_{i1},\alpha_{i2},\alpha_{i3})$, $\gamma_t = (\gamma_{t0},\gamma_{t1},\gamma_{t2},\gamma_{t3})$ and $\varepsilon_{it} = (\varepsilon_{it0},\varepsilon_{it1})$.
Here, $\alpha_{i0}$ and $\gamma_{t0}$ are additive fixed effects.
Suppose that $\alpha_{i0}$, $\alpha_{i1}$, $\alpha_{i2}$, $\alpha_{i3}$, $\gamma_{t0}$, $\gamma_{t1}$, $\gamma_{t2}$, $\gamma_{t3}$, $\varepsilon_{it0}$, and $\varepsilon_{it1}$ are mutually independent with mean 0 and variance 1.

To abstract away from finite-sample issues, consider the \textit{population} double differences\footnote{For a formal account of the discrepancy between the \textit{population} and \textit{sample} double differences, see Section \ref{sec:fixed_effect_theory}.}
\begin{align*}
\tilde X_{it} &= X_{it} - \mu^X_{i} - \mu^X_{t} + \mu^X
&&\left(\text{where } \bar \mu^X_{i}=E[ X_{it} | \alpha_i], \mu^X_{t}=E[ X_{it} | \gamma_t], \& \ \mu^X = E[ X_{it} ]\right)
\\
\tilde U_{it} &= U_{it} - \mu^U_{i} - \mu^U_{t} + \mu^U
&&\left(\text{where } \bar \mu^U_{i}=E[ U_{it} | \alpha_i], \ \mu^U_{t}=E[ U_{it} | \gamma_t], \ \& \ \mu^U = E[ U_{it} ]\right)
\end{align*}
Note that they reduce to
\begin{align*}
\tilde X_{it} &= \alpha_{i1} \gamma_{t2} + \alpha_{i2} \gamma_{t1} + \varepsilon_{it0}\\
\tilde U_{it} &= \alpha_{i1} \gamma_{t3} + \alpha_{i3} \gamma_{t1} + \varepsilon_{it1}
\end{align*}
under the example \eqref{eq:FE_example}.
Certainly, the double differencing removes the FEs, $\alpha_{i0}$ and $\gamma_{t0}$, but still leaves the strong two-way dependence through $(\alpha_{i1},\alpha_{i2},\alpha_{i3})$ and $(\gamma_{t1},\gamma_{t2},\gamma_{t3})$.
Observe that there is no endogeneity, as
$
E[\tilde U_{it} | \tilde X_{it}]=0.
$
Furthermore, the score
\begin{align*}
\ddot X_{it} \ddot U_{it} 
=& 
\alpha_{i1}^2 \gamma_{t2} \gamma_{t3} + \alpha_{i1}\alpha_{i3}\gamma_{t1}\gamma_{t2} +\alpha_{i1}\alpha_{i2}\gamma_{t1}\gamma_{t3} + \alpha_{i2}\alpha_{i3}\gamma_{t1}^2 +
\\
& (\alpha_{i1}\gamma_{t2} + \alpha_{i2}\gamma_{t1})\varepsilon_{it1} +
(\alpha_{i1}\gamma_{t3} + \alpha_{i3}\gamma_{t1})\varepsilon_{it0}
\end{align*}
entails the non-degenerate projections
\begin{align*}
a_i = \alpha_{i2}\alpha_{i3} \qquad\text{ and }\qquad b_t = \gamma_{t2} \gamma_{t3}
\end{align*}
with respect to $\alpha_i$ and $\gamma_t$, respectively.

Hence, in this example, the components of $\alpha_i$ and $\gamma_t$ are not eliminated by the two-way fixed effects regression of $Y$ on $X$, so the score $\ddot X_{it} \ddot U_{it}$ is still two-way dependent, and two-way clustering is necessary for calculation of the covariance matrix. 
Furthermore, the score is not degenerate so our theory to be presented in Section \ref{sec:fixed_effect_theory} applies.

\subsection{Theory under Two-Way Fixed-Effect Models}\label{sec:fixed_effect_theory}

This section shows that our standard errors extend to two-way fixed-effect models.
In different settings, the existing literature has considered inference for two-way fixed-effect models. 
\cite{verdier2020estimation} studies linear regression with two-way fixed effects for fixed $T$ for sparsely matched data. 
\cite{juodis2021shock} considers bootstrap-based inference for linear models with two-way fixed effects under large $N$ and $T$ with a different set of assumptions. 

Consider the two-way fixed-effect model
\begin{align}
&Y_{it}=X_{it}'\beta + \xi_i+\eta_t +U_{it}, \label{eq:two_way_fe1}\\
&(X_{it}', \xi_i,\eta_t, U_{it})'=f(\alpha_i,\gamma_t,\varepsilon_{it}), \label{eq:two_way_fe2}
\end{align}
where $\xi_i$ and $\eta_t$ are fixed effects and $E[U_{it}]=0$.\footnote{Note that \eqref{eq:two_way_fe2} implicitly requires that $\xi_i = f_2(\alpha_i)$ and $\eta_t = f_3(\gamma_t)$.}
Define the within-transformed outcome and within-transformed regressors by
\begin{align*}
\ddot Y_{it} =& Y_{it} - \frac{1}{N}\sum_{i'=1}^NY_{it} -\frac{1}{T}\sum_{t'=1}^T Y_{it'}+\frac{1}{NT}\sum_{i'=1}^N\sum_{t'=1}^TY_{i't'},\qquad\text{and}\\
\ddot X_{it} =& X_{it} - \frac{1}{N}\sum_{i'=1}^NX_{it} -\frac{1}{T}\sum_{t'=1}^T X_{it'}+\frac{1}{NT}\sum_{i'=1}^N\sum_{t'=1}^TX_{i't'},
\end{align*}
respectively.
The within transformations induce complex dependence structure between transformed variables. 
Using idempotency of the within transformation matrices (see Ch. 17.8 of \citealt{hansen2021econometrics}), the two-way within estimator $\hat \beta$ for $\beta$ is defined as the OLS estimator of $\ddot Y_{it}$ on $\ddot X_{it}$ and thus satisfies
\begin{align*}
\hat \beta - \beta = \left(\frac{1}{NT}\sum_{i=1}^N\sum_{t=1}^T\ddot X_{it}\ddot X_{it}'\right)^{-1}\frac{1}{NT}\sum_{i=1}^N\sum_{t=1}^T \ddot{X}_{it}U_{it}.
\end{align*}
Also, define the variance estimator $\hat \Sigma_{NT}$ as in \eqref{eq:V_hat} with $(\ddot Y_{it},\ddot X_{it}')$ in place of $(Y_{it},X_{it}')$ and with the two-way within estimator $\hat \beta$ defined in this section.
Denote the population counterpart of the within-transformed regressor by
\begin{align*}
\tilde X_{it}=X_{it}-E[X_{it}|\gamma_t]- E[X_{it}|\alpha_i]+E[X_{it}].
\end{align*}
Note that $(\tilde X_{it}',U_{it})$ only depends on $\alpha_i,\gamma_t$, and $\varepsilon_{it}$ and satisfies $E[\tilde X_{it}]=0$.

\begin{theorem}[Regression models with two-way fixed effects]\label{thm:TWFE}
Suppose Assumption \ref{a:NW_OLS} (ii), (iii), (vi)(1), (v), (vi) holds for $(X_{it}',U_{it})$, and the outcome variable $Y_{it}$ is generated following \eqref{eq:two_way_fe1}--\eqref{eq:two_way_fe2} where $\alpha_i$, $\gamma_t$, and $\varepsilon_{it}$ are defined in the same way as in Assumption \ref{a:NW_OLS}. 
In addition, assume that
$E[\tilde X_{it} U_{it}]=0$
and
$\|X_{it}\|_\infty\le K$
for a constant $K$ that is independent of $N$ and $T$.
Then, the conclusions in Theorems \ref{thm:var_consistency} and \ref{thm:OLS} continue to hold, and thus
\begin{align*}
\hat \Sigma_{NT}^{-1/2}(\hat \theta - \theta) &\stackrel{d}{\to} N(0,I_m).
\end{align*}
\end{theorem}

A proof is provided in Appendix \ref{sec:TWFE}.
It is worthy noting that the proof is not a mere extension of the previous theorems, but requires nontrivial technicalities involving maximal inequalities in the time-series framework.

\section{Simulations: Comparisons of Alternative Standard Errors}\label{sec:simulations}

In this section, we use simulated data to examine the performance of our proposed robust variance estimator in comparison with six existing alternatives.

We generate data based on the linear model
\begin{align*}
Y_{it} = \beta_0 + \beta_1 X_{it} + U_{it},
\end{align*}
where the right-hand side variables $(X_{it},U_{it})'$ are generated through the panel dependence structure
\begin{align*}
X_{it} &= w_\alpha \alpha^x_i + w_\gamma \gamma^x_t + w_\varepsilon \varepsilon^x_{it}
\qquad\text{and}\\
U_{it} &= w_\alpha \alpha^u_i + w_\gamma \gamma^u_t + w_\varepsilon \varepsilon^u_{it}.
\end{align*}
We set $(\beta_0,\beta_1)= (1,1)$ throughout.
For the weight parameters, we use $(w_\alpha,w_\gamma,w_\varepsilon) = (0,0,1)$ to generate i.i.d. data and also use $(w_\alpha,w_\gamma,w_\varepsilon) = (0.25,0.50,0.25)$ to generate dependent data.
The latent components $(\alpha^x_i,\alpha^u_i,\varepsilon^x_{it},\varepsilon^u_{it})$ are all mutually independent $N(0,1)$.

The latent common time effects $(\gamma^x_t,\gamma^u_t)$ are dynamically generated according to the AR(1) design:
\begin{align*}
&
\gamma^x_t = \rho\gamma^x_{t-1} + \tilde\gamma^x_t
\text{ where $\tilde\gamma^x_t$ are independent draws from $N(0,1-\rho^2)$; and}
\\
&
\gamma^u_t = \rho\gamma^u_{t-1} + \tilde\gamma^u_t
\text{ where $\tilde\gamma^u_t$ are independent draws from $N(0,1-\rho^2)$.}
\end{align*}
The initial values are drawn from $N(0,1)$.
We vary the AR coefficient $\rho \in \{0.25,0.50,0.75\}$ across sets of simulations.

For each realization of observed data $\{(Y_{it},X_{it}) : 1 \le i \le N, 1 \le t \le T\}$ constructed according to the data generating process described above, we estimate $(\beta_0,\beta_1)$ by OLS.
Our objective is to evaluate the performance of our proposed robust variance estimator $\hat V_{NT} = \hat Q^{-1} \hat\Omega_{NT} \hat Q^{-1} $, where $\hat Q$ and $ \hat\Omega_{NT} $ are given in \eqref{eq:Q_hat} and \eqref{eq:M_estimator}, and the tuning parameter $\hat M$ is chosen according to the rule \eqref{eq:m_hat}.
Through simulation studies, we examine the performance of this robust variance estimator (hereafter referred to as CHS; Chiang-Hansen-Sasaki) in comparison with six existing alternative variance estimators which are in popular use for panel data analysis.
They include the heteroskedasticity robust estimator (EHW; Eicker-Huber-White; also known as HC0), the cluster robust estimator within $i$ (CR$i$; which corresponds to the Liang-Zeger-Arellano estimator), the cluster robust estimator within $t$ (CR$t$), the two-way cluster robust estimator \citep*[CGM;][]{CGM2011}, the wild bootstrap estimator \citep*[MNW;][]{mackinnon2021wild} for CGM, the bootstrap estimator \citep*[M;][]{menzel2021bootstrap},\footnote{Implementation of M requires some tuning parameters. We set the number of bootstrap iterations and the model selection tuning parameters, $\kappa_a$ and $\kappa_g$ following the simulation code for regressions by \citet{menzel2021bootstrap}. In addition to the default method of M, we also ran M without its model selection feature to find its results the same as those of the default method. Hence, we only report results by the default method of M.} and the two-way cluster robust estimator with 2-dependence \citep*[T;][]{Thompson2011}.

\begin{table}
\renewcommand{\arraystretch}{0.975}
\centering
\begin{tabular}{lcccccccccccc}
\multicolumn{12}{c}{I.I.D. Design: Nominal Probability = 95\%}\\
\hline\hline
& $N$ & $T$ & $\rho$ & EHW & CR$i$ & CR$t$ & CGM & MNW & M & T & CHS\\
\hline
(I)    & 50 & 100 & --- & \one 0.947 & \two 0.939 & \one 0.942 & \two 0.933 & \one 0.947 & \fiv 0.999 & \fou 0.912 & \one 0.949\\
(II)   & 75 & 75 & --- & \one 0.951 & \one 0.945 & \one 0.947 & \one 0.940 & \one 0.949 & \fiv 0.999 & \fou 0.913 & \one 0.953\\
(III)  & 100 & 50 & --- & \one 0.953 & \one 0.950 & \one 0.945 & \one 0.940 & \one 0.953 & \fiv 0.999 & \six 0.896 & \one 0.952\\
\hline\hline
&&&& \hspace{1.12cm} & \hspace{1.12cm} & \hspace{1.12cm} & \hspace{1.12cm} & \hspace{1.12cm} & \hspace{1.12cm} & \hspace{1.12cm} & \hspace{1.12cm} \ \\
\multicolumn{12}{c}{Dependence Design: Nominal Probability = 95\%}\\
\hline\hline
& $N$ & $T$ & $\rho$ & EHW & CR$i$ & CR$t$ & CGM & MNW & M & T & CHS\\
\hline
(IV)    & 50 & 100 & 0.25 & 0.293 & 0.484 & \fou 0.917 & \two 0.931 & \one 0.948 & \two 0.939 & \thr 0.921 & \one 0.953\\
(V)   & 75 & 75 & 0.25 & 0.251 & 0.386 & \thr 0.920 & \thr 0.928 & \one 0.945 & \two 0.934 & \fou 0.912 & \one 0.949\\
(VI)  & 100 & 50 & 0.25 & 0.218 & 0.291 & \fou 0.910 & \fou 0.915 & \two 0.936 & \fou 0.919 & \sev 0.882 & \two 0.936\\
\hline
(VII)    & 50 & 100 & 0.50 & 0.281 & 0.485 & \sev 0.884 & \fiv 0.904 & \thr 0.928 & \fou 0.916 & \fou 0.916 & \two 0.933\\
(VIII)   & 75 & 75 & 0.50 & 0.235 & 0.388 & \six 0.891 & \fiv 0.903 & \thr 0.924 & \fou 0.911 & \fiv 0.907 & \two 0.937\\
(IX)  & 100 & 50 & 0.50 & 0.206 & 0.290 & \sev 0.886 & \six 0.893 & \fou 0.915 & \fiv 0.900 & \eig 0.874 & \thr 0.924\\
\hline
(X)    & 50 & 100 & 0.75 & 0.257 & 0.511 & 0.813 & \nin 0.861 & \sev 0.887 & \eig 0.875 & \fou 0.913 & \fiv 0.909\\
(XI)   & 75 & 75 & 0.75 & 0.233 & 0.423 & 0.829 & 0.855 & \eig 0.879 & \nin 0.867 & \fiv 0.901 & \fiv 0.908\\
(XII)  & 100 & 50 & 0.75 & 0.196 & 0.325 & 0.825 & 0.840 & \eig 0.870 & 0.852 & \eig 0.870 & \six 0.890\\
\hline\hline
\end{tabular}
\caption{Coverage probabilities for the slope parameter $\beta_1$ for the OLS with the nominal probability of 95\% based on 10,000 Monte Carlo iterations. The top and bottom panels show results under the i.i.d. and dependence designs, respectively. The sample size is indicated by $(N,T)$. The parameter $\rho$ indicates the AR coefficient in the dependence design. EHW stands for Eicker–Huber–White, CR$i$ stands for cluster robust within $i$, CR$t$ stands for cluster robust within $t$, CGM stands for Cameron-Gelbach-Miller, MNW stands for MacKinnon-Nielsen-Webb, M stands for Menzel, T stands for Thompson, and CHS stands for Chiang-Hansen-Sasaki.\\${}$}
\label{tab:sim}
\end{table}

Table \ref{tab:sim} reports simulation results.
Reported values are the coverage frequencies for the slope parameter $\beta_1$ for the nominal probability of 95\% based on 10,000 Monte Carlo iterations. 
The top and bottom panels show coverage probability results under the i.i.d. design and the dependence design, respectively.
In each group of three consecutive rows, the panel sample sizes $(N,T)$ vary by rows.
Cells are shaded based on the proximity of the simulated coverage probability to the nominal probability of 0.95; the darker shades indicate more correct coverage.

We observe the following four points in these results.
First, under the i.i.d. design (rows (I)--(III)), EHW, CR$i$, CR$t$, CGM, MNW and CHS produce accurate coverage probabilities.
On the other hand, M yields over-coverage consistently across different sample sizes, and T yields under-coverage especially under small $T$.
Second, EHW and CR$i$ tend to behave poorly in general when there are two ways of cluster dependence as in rows (IV)--(XII).
Third, CR$t$ also tends to behave poorly under larger extents of serial dependence as in rows (X)--(XII). 
Likewise, CGM, MNW, and M perform less preferably as the serial dependence becomes even stronger, as in rows (X)--(XII).
These results are consistent with the fact that these methods do not account for serially correlated common time effects.
Fourth, in contrast, T and CHS behave more robustly under strong serial dependence especially when $T$ is large, as in rows (X) and (XI).
Whenever $T$ is small, as in rows (III), (VI), (IX) and (XII), however, T incurs severe under-coverage and hence CHS outperforms T in general.
The last observation that T performs poorly for small sample sizes is consistent with the similar observations made by \citet[][page 7]{Thompson2011} in his Monte Carlo simulation studies.
In summary, we demonstrate that confidence intervals constructed with our proposed standard errors lead to robustly superior coverage performance relative to the existing methods.

We ran additional simulations beyond those presented in this section.
Their results are found in Appendix \ref{sec:additional_simulations}.
Specifically, Appendix \ref{sec:simulations_power} illustrates power analyses, and Appendix \ref{sec:simulations_fixed_effect} presents simulations for the two-way fixed-effect estimator.

\section{An Empirical Application}\label{sec:application}

In this section, we highlight differences across the alternative standard error estimates for estimates of a simple asset pricing model.
Consider the Fama-French three-factor model
\begin{align*}
R_{it}-R_{ft} = \beta_1 (R_{Mt}-R_{ft}) + \beta_2 SMB_t + \beta_3 HML_t + e_{it},
\end{align*}
where 
$R_{it}$ is the total return of portfolio/stock $i$ in month $t$,
$R_{ft}$ is the risk-free rate of return in month $t$,
$R_{Mt}$ is the total market portfolio return in month $t$,
$SMB_t$ is the size premium (small$-$big),
$HML_t$ is the value premium (high$-$low), and
$\beta = (\beta_1,\beta_2,\beta_3)'$ are the factor coefficients.

We use two data sets of portfolio/stock returns.
They are 
(A) 44 industry portfolios excluding four financial sectors (banking, insurance, real estate, and trading) and (B) individual stocks.
For each of these data sets, (A) and (B), we use the monthly panel of length $120$ from January 2000 to December 2009.
For the individual stock data set (B), we use the balanced portion of the panel data, consisting of $N=779$ stocks.
The risk-free rate is based on the monthly 30-day T-bill beginning-of-month yield.
See Appendix \ref{sec:data_asset_pricing} for the source of data.

Let $\ddot Y_{it}$ and $\ddot X_{it}$ denote the within-transformations\footnote{
Technically, our theory does not allow for these demeaned variables as they are $(N,T)$-dependent.  But we expect that there should be no change to the theory if we allow for array data. Formalizing this would be a useful future research direction.
} of
$R_{it}-R_{ft}$ and $(R_{Mt}-R_{ft},SMB_t,HML_t)'$, respectively.
($\ddot X_{it}$ is homogeneous in the cross section.)
We estimate $\beta$ by the within-estimator
\begin{align*}
\hat \beta=\left(\frac{1}{NT}\sum_{i=1}^N\sum_{t=1}^T \ddot X_{it} \ddot X_{it}'\right)^{-1} \frac{1}{NT}\sum_{i=1}^N\sum_{t=1}^T \ddot X_{it} \ddot Y_{it}
\end{align*}
to remove any additive portfolio/stock fixed effects.
Thus, our proposed standard errors are computed based on $\hat V_{NT}=\hat Q^{-1}\hat \Omega_{NT}\hat Q^{-1}$ with $\hat Q$ and $\hat \Omega_{NT}$ from
\eqref{eq:Q_hat} and \eqref{eq:M_estimator} with $Y_{it}$, $X_{it}$ and $\hat U_{it}$ replaced by $\ddot Y_{it}$, $\ddot X_{it}$ and $\ddot Y_{it} - \ddot X_{it} \hat\beta$, respectively.
Table \ref{tab:estimation_results} summarizes estimates $\hat\beta$ of $\beta$ along with alternative standard error estimates of them for each of the two data sets described above.

\begin{table}[t]
	\centering
	\renewcommand{\arraystretch}{1}
		\begin{tabular}{ccccccccccc}
\multicolumn{9}{c}{(A) 44 Industry Portfolios. $(N,T)=(44,119)$.}\\
\hline\hline
&&\multicolumn{7}{c}{Standard Errors}\\
\cline{3-9}
&$\hat\beta$&EHW&CR$i$&CR$t$&CGM&M&T&CHS\\
\hline
MKT&0.959&0.022&0.055&0.030&0.059&0.021&0.059&0.059\\
SMB&0.076&0.029&0.035&0.041&0.045&0.029&0.056&0.052\\
HML&0.358&0.030&0.066&0.049&0.076&0.028&0.082&0.079\\
\hline\hline
${}$\hspace{1.2cm}${}$&${}$\hspace{1.2cm}${}$&${}$\hspace{1.2cm}${}$&${}$\hspace{1.2cm}${}$&${}$\hspace{1.2cm}${}$&${}$\hspace{1.2cm}${}$&${}$\hspace{1.2cm}${}$&${}$\hspace{1.2cm}${}$&${}$\hspace{1.2cm}${}$\\
\multicolumn{9}{c}{(B) Individual Stocks. $(N,T)=(779,119)$.}\\
\hline\hline
&&\multicolumn{7}{c}{Standard Errors}\\
\cline{3-9}
&$\hat\beta$&EHW&CR$i$&CR$t$&CGM&M&T&CHS\\
\hline
MKT&1.157&0.012&0.021&0.033&0.037&0.035&0.036&0.034\\
SMB&0.474&0.020&0.025&0.051&0.053&0.055&0.068&0.070\\
HML&0.173&0.016&0.029&0.053&0.058&0.054&0.055&0.053\\
\hline\hline
		\end{tabular}
	\caption{Estimates of the factor coefficients with six alternative standard error estimates of them. EHW stands for Eicker–Huber–White, CR$i$ stands for cluster robust within $i$, CR$t$ stands for cluster robust within $t$, CGM stands for Cameron-Gelbach-Miller, M stands for Menzel, T stands for Thompson, and CHS stands for Chiang-Hansen-Sasaki.}${}$
	\label{tab:estimation_results}
\end{table}

For each row of Table \ref{tab:estimation_results}, the standard error due to the Eicker-Huber-White (EHW) estimator is smaller than any other standard error.
On the other hand, the two-way cluster robust estimators of Cameron-Gelbach-Miller (CGM) and Thompson (T) and our proposed estimator (CHS) tend to yield the largest standard errors in each row.
The remaining three standard errors stay in the middle between these two groups.
The estimator by Menzel (M) behaves similarly to EHW in panel (A) while it behaves similarly to CGM in panel (B).
This puzzling outcome arises from the model selection feature of M.
In fact, M would also behave similarly to CGM in panel (A) as well if the tuning parameter of M were chosen to take a much smaller value.\footnote{As in the simulation section, we set the number of bootstrap iterations and the model selection tuning parameters for M following the simulation code for regressions by \citet{menzel2021bootstrap}.}
Because of these idiosyncratic behaviors of M that depend on discrete outcomes of model selection which in turn depend on tuning parameters, we will hereafter focus on the other estimators in comparing the results.

Observe in panel (A) that the statistical significance of the coefficient $\beta_2$ of SMB meaningfully diminishes as the standard error estimator becomes more robust (again, except for M).
Specifically, it is significant at the 5\% level with EHW, CR$i$ and M, but it becomes insignificant at this level with CR$t$, CGM, T, or CHS.
Furthermore, while it is significant at the 10\% level with EHW, CR$i$, CR$t$, and CGM, it becomes insignificant at this level with T and CHS.
This part of the estimation results shows a case where accounting for serial correlation in common time effects may even overturn conclusions from statistical inference based on the other standard errors. 
Accounting for arbitrary untruncated time correlation, however, CHS yields a slightly smaller standard error estimate than T for this case.

\section{Summary and Discussions}\label{sec:summary}

In this paper, we propose new robust standard error estimators for panel data.
The new estimators account for the cluster dependence within $i$, the cluster dependence within $t$, and serial dependence in the common time effects.
In particular, all the existing robust standard error estimators fail to accommodate untruncated serial dependence in the common time effects, while this feature is relevant to empirical data used in economics and finance.
Simulation studies show that the new standard errors produce robustly superior coverage performance than existing alternatives, including the heteroskedasticity robust estimator (Eicker-Huber-White; also known as HC0), the cluster robust estimator within $i$, the cluster robust estimator within $t$, the two-way cluster robust estimator \citep*{CGM2011}, the bootstrap estimator \citep*{menzel2021bootstrap}, and the two-way cluster robust estimator with 2-dependence \citep*{Thompson2011}.

In the rest of this section, we discuss limitations of our method and potentials of future research in relation to the existing literature.
Since the seminal work by \citet*{CGM2011} and \citet*{Thompson2011}, a few important papers have proposed methods of robust inference in two way cluster dependence.

\cite*{DDG2019} derive Donsker results under multiway cluster dependence.
In the current paper, we only derive limit distributions for finite-dimensional parameters that are relevant to many empirical applications in economics and finance.
In other words, \cite*{DDG2019} provide a generalization of existing results by allowing for empirical processes, we on the other hand provide a generalization in a different direction by allowing for serial dependence in common time effects.
Combining these two directions of generalization is left for future research.

\cite*{menzel2021bootstrap} proposes a method of (conservative) inference with uniform validity over a large class of distributions including the case of non-Gaussian degeneracy under two-way cluster dependence.
In the current paper, for the purpose of providing a simple method of inference via analytic standard error formulas, we focus on the case of Gaussian degeneracy as well as non-degenerate cases.
In other words, \cite*{menzel2021bootstrap} provides a generalization of existing results by allowing for uniformity over a large class, we on the other hand provide a generalization in a different direction by allowing for serial dependence in common time effects.
Combining these two directions of generalization is also left for future research.

\cite*{chiang2020inference} derive a high-dimensional central limit theorem under multiway cluster dependence.
In the current paper, we only consider finite-dimensional parameters that are relevant to many empirical applications in economics and finance.
In other words, \cite*{chiang2020inference} provide a generalization of existing results by allowing for high dimensionality, we on the other hand provide a generalization in a different direction by allowing for serial dependence in common time effects.
Again, combining these two directions of generalization is left for future research.

The independence conditions in Assumption \ref{a:NW} (i) may be relaxed in a couple of directions.
One way is relax the i.i.d. assumption on the $\varepsilon_{it}$ factor in \eqref{eq:generalized_ahk}.
With this said, the existing literature explicitly or implicitly makes this assumption, and we continue to focus on i.i.d. $\varepsilon_{it}$.
Another way is to relax the i.i.d. assumption on the $\alpha_{i}$ factor in \eqref{eq:generalized_ahk}.
For instance, if a researcher obtains spatial information associated with panel data, then it may be a possibility to allow for spatial $\alpha$-mixing \citep[e.g.,][]{jenish2009central}.
We leave these extensions for future research.

Our standard errors allowing for serial correlation in time effects effectively use the Newey-West-type long-run variance estimation, but it is well known that in some cases a relatively long time series may be required for such estimators to perform well \citep[e.g.,][]{lazarus2018har}.
Inference based on moving-block bootstrap \citep[e.g.,][]{gonccalves2011moving} may improve the finite-sample performance, and we suggest it as another direction for future research.
Another promising recent proposal by \citet{chen2023fixed} is to use fixed-b asymptotic theory to derive bias corrections and improve the distributional approximation.


\appendix

\section*{Appendix}

\section{Proof of Theorem \ref{thm:variance_mean}}\label{sec:thm:variance_mean}
\begin{proof}
Following the argument in the text, without loss of generality set $\theta = 0$, and make the decompositions \eqref{eq:Xabe} and \eqref{eq:thetahat}. Since the sums in \eqref{eq:thetahat} are uncorrelated, we find
\begin{align}
var(\hat\theta)=var\left( \frac{1}{N} \sum_{i=1}^N a_i \right)+var\left( \frac{1}{T} \sum_{t=1}^T b_t \right)+var\left( \frac{1}{NT} \sum_{i=1}^N \sum_{t=1}^T e_{it} \right). 
\label{eq:varsum}
\end{align}
We take each term separately. 

First, recall that $a_i$ is i.i.d., mean zero, and has variance matrix $\Sigma_{a}$. Since $a_i=E[X_{it}\mid \alpha_i]$, an application of the conditional Jensen inequality shows that 
\begin{align*}
||\Sigma_a||\le E[||a_i||^2] = E[||E[X_{it}\mid \alpha_i]||^2] \le E[E[||X_{it}||^2\mid \alpha_i]] = E[||X_{it}||^2],
\end{align*}
under Assumption \ref{a:NW}. Hence $||\Sigma_a||<\infty$ as claimed. We calculate that
\begin{align}
var\left( \frac{1}{N} \sum_{i=1}^N a_i \right)=\frac{1}{N}\Sigma_a.
\label{eq:vara}
\end{align}

Second, by a standard variance decomposition for stationary time series,
\begin{align}
var\left( \frac{1}{T} \sum_{t=1}^T b_t \right)&=\frac{1}{T}\sum_{\ell=-(T-1)}^{T-1} \left( 1-\frac{|\ell |}{T} \right)E[b_tb_{t+\ell}'] \nonumber\\
&=\frac{1}{T}\Sigma_b\left(1+o(1)\right).
\label{eq:varb}
\end{align}
The second equality holds if the sum \eqref{eq:sigma_b} converges, which we now demonstrate. Since $b_t=E[X_{it}\mid \gamma_t]$, an application of the conditional Jensen inequality shows that $E[||b_t||^s] = E[||E[X_{it}\mid \gamma_t]||^s] \le E[E[||X_{it}||^s\mid \gamma_t]] = E[||X_{it}||^s]$ for $s \ge 1$, and in particular $E[||b_t||^{4(r+\delta)}]<\infty$. Also, as $b_t$ is a function only of $\gamma_t$, it has the same mixing coefficients. By an application of Theorem 14.13 (ii) in \cite{hansen2021econometrics}, we find 
\begin{align*}
||\Sigma_{b}||=\left|\left|\sum_{\ell=-\infty}^{\infty}E[b_tb_{t+\ell}']\right|\right|\le 8 \left(E||b_t||^{4(r+\delta)}\right)^{1/2(r+\delta)} \sum_{\ell=-\infty}^{\infty}\alpha(\ell)^{1-1/2(r+\delta)}<\infty
\end{align*}
under Assumption \ref{a:NW}. Hence $||\Sigma_b||<\infty$ as claimed, and the bound \eqref{eq:varb} follows.

Third, the law of iterated expectations implies that for $j \ne i$
\begin{align*}
E[e_{it}e_{js}']=E[E[e_{it}e_{js}' \mid \gamma_t,\gamma_s]]=E[E[e_{it} \mid \gamma_t]E[e_{js}' \mid \gamma_s]]=0 
\end{align*}
the second equality since conditional on $(\gamma_t,\gamma_s)$, $e_{it}$ is independent of $e_{js}$ for $j \ne i$. Combined with the stationarity of $e_{it}$, this implies that
\begin{align}
var\left( \frac{1}{NT} \sum_{i=1}^N \sum_{t=1}^T e_{it} \right)
&=\frac{1}{(NT)^2}\sum_{i=1}^N\sum_{t=1}^T\sum_{s=1}^T E[e_{it}e_{is}'] \nonumber\\
&=\frac{1}{NT}\sum_{\ell=-(T-1)}^{T-1} \left( 1-\frac{|\ell |}{T} \right)E[e_{it}e_{i,t+\ell}'] \nonumber\\
&=\frac{1}{NT}\Sigma_e\left(1+o(1)\right).
\label{eq:vare}
\end{align}
The final equality holds if the sum \eqref{eq:sigma_e} converges, which we now demonstrate.
Conditional on $\alpha_i$, $e_{it}$ is an $\alpha$-mixing process with the same mixing coefficients as $\gamma_t$. Furthermore, the moments of $e_{it}$ are bounded by those of $X_{it}$. 
Using iterated expectations, Jensen's inequality, the fact that $E[e_{it} \mid \alpha_i]=0$, Theorem 14.13 (ii) in \cite{hansen2021econometrics}, and again Jensen's inequality,
\begin{align*}
||E[e_{it}e_{i,t+\ell}']|| \le E||E[e_{it}e_{i,t+\ell}'\mid \alpha_i]|| \le 8 \left(E||e_{it}||^{4(r+\delta)}\right)^{1/2(r+\delta)} \alpha(\ell)^{1-1/2(r+\delta)}.
\end{align*}
This (under Assumption \ref{a:NW}) shows that 
\begin{align*}
||\Sigma_e||=\left|\left|\sum_{\ell=-\infty}^{\infty}E[e_{it}e_{i,t+\ell}']\right|\right|\le 8 \left(E||e_{it}||^{4(r+\delta)}\right)^{1/2(r+\delta)} \sum_{\ell=-\infty}^{\infty}\alpha(\ell)^{1-1/2(r+\delta)}<\infty.
\end{align*}
Thus $||\Sigma_e||<\infty$ and the bound \eqref{eq:vare} follows.

Together, \eqref{eq:varsum}-\eqref{eq:vare} establish that
\begin{align*}
var(\hat\theta)=\frac{1}{N}\Sigma_a+\frac{1}{T}\Sigma_b\left(1+o(1)\right)+\frac{1}{NT}\Sigma_e\left(1+o(1)\right)
\end{align*}
as claimed. 

Since $||\Sigma_a||<\infty$, $||\Sigma_b||<\infty$, and $||\Sigma_e||<\infty$, it follows that $var(\hat\theta)\to 0$ as $N,T\to\infty$. By Chebyshev's inequality, we deduce that $\hat\theta\stackrel{p}{\to}\theta$, completing the proof.
\end{proof}

\section{Proof of Theorem \ref{thm:main_asymptotics}}\label{sec:thm:main_asymptotics}

\begin{proof}
Without loss of generality, set $\theta = 0$. The proof is different under Assumption \ref{a:non-degenerate} parts (i) and (ii). First, take case (ii). If $X_{it}$ is i.i.d. then $\sqrt{NT}\hat \theta\stackrel{d}{\to} N(0,var(X_{it}))$ by Lyapunov's central limit theorem. Also, $var(\hat \theta)=var(X_{it})/NT$. Combining, we find the stated result. Hence, for the remainder of the proof we focus on case (i).

Using \eqref{eq:Xabe},
\begin{align}
	\sqrt{N}\hat \theta=\frac{1}{\sqrt{N}} \sum_{i=1}^N a_i+\sqrt{\frac{N}{T}}\frac{1}{\sqrt{T}} \sum_{t=1}^T b_t+\frac{1}{\sqrt{N}T} \sum_{t=1}^T e_{it}.
 \label{eq:theta_ab}
\end{align}

The first term in \eqref{eq:theta_ab} consists of a sum of the i.i.d. zero-mean random vectors $a_i$ with finite variance $\Sigma_a$. By Lyapunov's central limit theorem, we deduce
\begin{align}
\frac{1}{\sqrt{N}}\sum_{i=1}^N a_{i}\stackrel{d}{\to} N(0,\Sigma_a). \label{eq:theta_a}
\end{align}

The second term in \eqref{eq:theta_ab} consists of the $\alpha$-mixing sequence $b_t$ (see Theorem 14.12 in \citet{hansen2021econometrics}). 
Applying the central limit theorem for $\alpha$-mixing sequences \citep[cf.][Theorem 14.15]{hansen2021econometrics} under Assumption \ref{a:NW}(iii), which implies $\sum_{\ell=1}^\infty \alpha(\ell)^{1-1/2(r+\delta)}<\infty$, we deduce 
\begin{align}
\frac{1}{\sqrt{T}}\sum_{t=1}^T b_{t}\stackrel{d}{\to}N(0,\Sigma_b) \label{eq:theta_b}.
\end{align}
The asymptotic variance $\Sigma_b$ was shown finite in Theorem \ref{thm:variance_mean}. The asymptotic distributions in \eqref{eq:theta_a} and \eqref{eq:theta_b} are independent since the sequences $a_i$ and $b_t$ are independent.

Equation \eqref{eq:vare} shows that the variance of the third term in \eqref{eq:theta_ab} is $O(T^{-1})$, and hence this term is $o_p(1)$. Together, \eqref{eq:theta_ab}-\eqref{eq:theta_b} plus $N/T\to c$ show that
\begin{align}
\sqrt{N}\hat \theta \stackrel{d}{\to}& N(0,\Sigma_a)+\sqrt{c}N(0,\Sigma_b)=N(0,\Sigma),
\label{eq:normal_limit}
\end{align}
where $\Sigma \equiv \Sigma_a+c\Sigma_b$. Assumption \ref{a:non-degenerate}(i) implies that $\Sigma>0$.

Theorem \ref{thm:variance_mean} and $N/T\to c$ establish that
\begin{align*}
N\cdot var(\hat\theta)=\Sigma_a+\frac{N}{T}\Sigma_b(1+o(1))+\frac{1}{T}\Sigma_e(1+o(1))\to\Sigma.
\end{align*}

Together, 
\begin{align*}
	\left(var( \hat \theta )\right)^{-1/2}\hat \theta=\left(N\cdot var( \hat \theta )\right)^{-1/2}\sqrt{N}\hat \theta
\stackrel{d}{\to} \Sigma^{-1/2}N(0,\Sigma)=N(0,I_m).
\end{align*}
This is the stated result.
\end{proof}


\section{Proof of Theorem \ref{thm:var_consistency}}\label{sec:var_consistency}
\begin{proof}

The proof branches into the two cases, (1) and (2), of Assumption \ref{a:NW_OLS} (iv). For readibility, we defer some of the lengthy technical calculations to Lemmas \ref{lem:GNW}-\ref{lem:EW_CRVE_degeneracy} in Appendix \ref{sec:lemmas}.

First, consider the case where Assumption \ref{a:NW_OLS} (iv) (1) holds.
From Theorem \ref{thm:variance_mean}, we have
\begin{align*}
	N\Omega_{NT}=\Sigma_a+\frac{N}{T}\Sigma_b+o(1)\to\Sigma>0,
\end{align*}
where $\Sigma=\Sigma_a+c\Sigma_b$. Thus 
\begin{align}
N\Sigma_{NT}=N R'Q^{-1}\Omega_{NT} Q^{-1} R\to R'Q^{-1}\Sigma Q^{-1}R>0. 
\label{eq:NSigma}
\end{align}

Next, by combining Lemmas \ref{lem:GNW} and \ref{lem:EW_CRVE} from Appendix \ref{sec:lemmas}, we have
	\begin{align*}
	N\hat \Omega_{NT} & =\underbrace{\frac{1}{NT^2}\sum_{i=1}^N \hat R_i \hat R_i'}_{=E[a_ia_i']+o_p(1)}  + \underbrace{\frac{1}{NT^2}\sum_{t=1}^T \hat S_t \hat S_t'}_{=cE[b_tb_t']+o_p(1)}  - \underbrace{\frac{1}{NT^2}\sum_{i=1}^N\sum_{t=1}^T X_{it}X_{it}' \hat U_{it}^2 }_{=O_p(T^{-1})}\notag\\
	& + \underbrace{\frac{1}{NT^2} \sum_{m=1}^M w(m,M)\left( \hat G_m +\hat G_m' -\hat H_m -\hat H_m' \right)}_{=c\sum_{m=-\infty}^\infty E[b_t b_{t+m}']+cE[b_tb_t'] +o_p(1)  } \\
&\stackrel{p}{\to} E[a_ia_i']+cE[b_tb_t']+c\sum_{m=-\infty}^\infty E[b_t b_{t+m}']+cE[b_tb_t']\\
&=\Sigma.
	\end{align*}
Now, consider $\hat Q$. 
Setting $\theta = \text{vec}(X_{it}X_{it}')$ and $\hat\theta = \text{vec}(\hat Q)$, which satisfy the conditions of Theorem \ref{thm:main_asymptotics} under Assumption \ref{a:NW_OLS}.
By the proof of Theorem \ref{thm:main_asymptotics} and the continuous mapping theorem, we obtain $||\hat Q^{-1}-Q^{-1}||=o_p(1)$ and $\hat Q^{-1}=O_p(1)$ by Assumption \ref{a:NW_OLS} (ii). 	
Together we have established that 
\begin{align}
N\hat \Sigma_{NT}=N R'\hat Q^{-1}\hat \Omega_{NT} \hat Q^{-1} R \stackrel{p}{\to} R'Q^{-1}\Sigma Q^{-1}R. 
\label{eq:NSigmaHat}
\end{align}

Equations \eqref{eq:NSigma} and \eqref{eq:NSigmaHat} together imply that
\begin{align*}
\Sigma_{NT}^{-1}\hat \Sigma_{NT} &\stackrel{p}{\to} \left(R'Q^{-1}\Sigma Q^{-1}R\right)^{-1}\left(R'Q^{-1}\Sigma Q^{-1}R\right)= I_k.
\end{align*}
This is the stated result.

Second, consider the case where Assumption \ref{a:NW_OLS} (iv) (2) holds. 
Observe that $\Sigma=var(X_{it}U_{it})>0$ under Assumption \ref{a:NW_OLS} (iv) (2).
Then $NT\Omega_{NT} = \Sigma>0$ and
\begin{align}
NT\Sigma_{NT} =NTR'Q^{-1}\Omega_{NT} Q^{-1} R= R'Q^{-1}\Sigma Q^{-1}R>0.
\label{eq:NTSigma}
\end{align}
Lemma \ref{lem:EW_CRVE_degeneracy} in Appendix \ref{sec:lemmas} implies 
$
NT\hat \Omega_{NT}\stackrel{p}{\to} NT\Omega_{NT}=\Sigma.
$
The law of large number for i.i.d. random variables, Assumption \ref{a:NW_OLS}(ii), and the continuous mapping theorem imply that $||\hat Q^{-1}-Q^{-1}||=o_p(1)$. 
Together, this implies that
\begin{align}
NT\hat \Sigma_{NT}=NT R'\hat Q^{-1}\hat \Omega_{NT} \hat Q^{-1} R \stackrel{p}{\to} R'Q^{-1}\Sigma Q^{-1}R. 
\label{eq:NTSigmaHat}
\end{align}
Equations \eqref{eq:NTSigma} and \eqref{eq:NTSigmaHat} together imply that
\begin{align*}
\Sigma_{NT}^{-1}\hat \Sigma_{NT} &\stackrel{p}{\to} \left(R'Q^{-1}\Sigma Q^{-1}R\right)^{-1}\left(R'Q^{-1}\Sigma Q^{-1}R\right)= I_k.
\end{align*}
This is the stated result and completes the proof.
\end{proof}


\section{Proof of Theorem \ref{thm:OLS}}\label{sec:cor_OLS}
\begin{proof}
Assumption \ref{a:NW_OLS} (iv) imposes either non-singular clustered dependence (1) or i.i.d. dependence (2). Under the latter the result is classical; hence we focus on condition (1). Some algebra reveals that
\begin{align}
\sqrt{N}(\hat \theta - \theta) &=\sqrt{N}R'\left(\frac{1}{NT}\sum_{i=1}^N\sum_{t=1}^T X_{it} X_{it}'\right)^{-1} \left(\frac{1}{NT}\sum_{i=1}^N\sum_{t=1}^T X_{it} U_{it}\right)\nonumber\\
&=R'Q^{-1}\frac{1}{\sqrt{N}T}\sum_{i=1}^N\sum_{t=1}^T X_{it}U_{it}-R'\hat Q^{-1}\left(\hat Q-Q\right)Q^{-1}\left(\frac{1}{\sqrt{N}T}\sum_{i=1}^N\sum_{t=1}^T X_{it} U_{it}\right).
\label{eq:thetaV}
\end{align}
The first term in \eqref{eq:thetaV} is a self-normalized sample mean in the random vectors $X_{it}U_{it}$, which satisfy the conditions of Theorem \ref{thm:main_asymptotics}. Consequently, for $\Sigma=\Sigma_a+c\Sigma_b$, the first term in \eqref{eq:thetaV} satisfies
\begin{align*}
R'Q^{-1}\frac{1}{\sqrt{N}T}\sum_{i=1}^N\sum_{t=1}^T X_{it} U_{it}\stackrel{d}{\to}R'Q^{-1}N(0,\Sigma)=N(0,R'Q^{-1}\Sigma Q^{-1}R),
\end{align*}
as shown in \eqref{eq:normal_limit}. 

Recall that $\hat Q$ is the sample average of the variables $X_{it}X_{it}'$, which satisfy the conditions of Theorem 2. It follows that $||\hat Q-Q||=o_p(1)$ and $\hat Q^{-1}=O_p(1)$. Consequently, the second term in \eqref{eq:thetaV} is $o_p(1)$. Together, we deduce that
\begin{align}
\sqrt{N}(\hat \theta - \theta) \stackrel{d}{\to}N(0,R'Q^{-1}\Sigma Q^{-1}R).
\label{eq:thetahatnormal}
\end{align}

From Theorem \ref{thm:variance_mean}, 
\begin{align*}
N\Omega_{NT}=\Sigma_a+\frac{N}{T}\Sigma_b+o(1)\to\Sigma.
\end{align*}
Thus we find that
\begin{align}
N\Sigma_{NT}\to R'Q^{-1}N\Omega_{NT}Q^{-1}R.
\label{eq:SigmaNT}
\end{align}

Together, \eqref{eq:thetahatnormal} and \eqref{eq:SigmaNT} imply that
\begin{align*}
\Sigma_{NT}^{-1/2}(\hat \theta - \theta) &=(N\Sigma_{NT})^{-1/2}\sqrt{N}(\hat \theta - \theta)\\
&\stackrel{d}{\to}(R'Q^{-1}N\Omega_{NT}Q^{-1}R)^{-1/2}N(0,R'Q^{-1}\Sigma Q^{-1}R)\\
&=N(0,I_m).
\end{align*}
This is \eqref{eq:thetadist1}.

Equation \eqref{eq:thetadist2} follows by combining \eqref{eq:thetadist1} with Theorem \ref{thm:var_consistency}.
\end{proof}

	\section{Proof of Theorem \ref{thm:TWFE}}
	\label{sec:TWFE}
\begin{proof}
	
	We first consider the case of non-degeneracy. 
	Since $X_{it}- E[X_{it}|\gamma_t]$ is independent across $i$ conditionally on $(\gamma_t)_{t=1}^T $ and $\|X_{it}\|_\infty \le K$, 
	Theorem 2.14.1 in \cite{vdVW1996} yields
	\begin{align*}
	&E	\left[\max_{t=1,...,T}\left|\frac{1}{N} \sum_{i'=1}^N X_{it} - E[X_{it}|\gamma_t]\right||(\gamma_t)_{t=1}^T\right]
	\lesssim \sqrt{\frac{K^2\log T}{N}}.
	\end{align*}
	Integrating out both sides using Fubini's theorem,
	we obtain
	\begin{align}\label{eq:tw_proof_gamma}
	&E	\left[\max_{t=1,...,T}\left|\frac{1}{N} \sum_{i'=1}^N X_{it} - E[X_{it}|\gamma_t]\right|\right]
	\lesssim \sqrt{\frac{K^2\log T}{N}}=o(1).
	\end{align}

	We are now going to show
	\begin{align}\label{eq:tw_proof_alpha}
	E\left[\max_{i=1,...,N}\left|\frac{1}{T} \sum_{t'=1}^T X_{it'} - E[X_{it}|\alpha_i]\right|\right]=o(1).
	\end{align}
	We use Bernstein's big-block-small-block argument (for example, see Step 1 in the Proof of Theorem E.1 in \citealt{chernozhukov2019inference}) to show \eqref{eq:tw_proof_alpha}.
	Specifically,
	let $q=q_T\sim T^{3/4}$, $s=s_T\sim T^{1/3}$, and $m=T/(q+s)$ be positive sequences of integers satisfying $q+s\le T/2$. 
	It immediately follows that $q,s\to\infty$, $q=o(T)$,  $s^2/T=o(1)$, $q^{-1}s^2\log N=o(1)$, and $m\sim T^{1/4}$.
	Define $I_1=\{1,...,q\}$, $J_1=\{q+1,...,q+s\},$ ...$I_m=\{(m-1)(q+s)+1,...,m(q+s)\}$, $J_m=\{(m-1)(q+s)+q+1,...,m(q+s)\}$, $J_{m+1}=\{m(q+s)+1,...,T\}$. 
	As $\lambda>2r/(r-1) >2$, we have
	\begin{align*}
	m \beta(s)= m O(s^{-\lambda}) =o(1),
	\end{align*}
	where $\beta(\cdot)$ is the $\beta$-mixing coefficient.
	The integers, $q$ and $s$, will serve as the lengths of big and small blocks, respectively, and $m$ is the number of big blocks. 
	Now, for each $i=1,...,N$, one has the decomposition
	\begin{align}
	\frac{1}{\sqrt{T}}\sum_{t=1}^T \{X_{it}-E[X_{it}|\alpha_i]\}=\frac{1}{\sqrt{T}}\sum_{l=1}^m  L_{il}+\frac{1}{\sqrt{T}}\sum_{l=1}^m  S_{il} + \frac{1}{\sqrt{T}}S_{i,(m+1)},\label{eq:tw_proof_decomposition}\\
	\text{where }	L_{il}=\sum_{t\in I_l} \{X_{it}-E[X_{it}|\alpha_i]\},\quad 	S_{il}=\sum_{t\in J_l} \{X_{it}-E[X_{it}|\alpha_i]\}.
	\notag
	\end{align}
	$L_{il}$ (respectively, $S_{il}$) equals the sum over a big block (respectively, small block).
	Define $(\hat L_{il})_{l=1}^m$ and $(\hat S_{il})_{l=1}^{m+1}$ to be the decoupled copies of $( L_{il})_{l=1}^m$ and $( S_{il})_{l=1}^{m+1}$, respectively.
	That is, they are two independent sequences of random vectors such that
	\begin{align*}
	\hat L_{il}\overset{d}{=} L_{il} \text{ for } l \in \{1,\ldots,m\} \quad\text{ and }\quad \hat S_{il}\overset{d}{=} S_{il} \text{ for } l \in \{1,\ldots,m+1\}
	\end{align*}
	conditionally on $(\alpha_i)_i$.
	We now claim that, for any $y\in \mathbb R$, it holds that
	\begin{align}
	&P\left( \max_{1\le i\le N}\frac{1}{\sqrt{T}}\sum_{l=1}^m \hat L_{il}\le y- o(1) \mid (\alpha_i)_i\right)- o(1) \nonumber\\
	\le& P\left(\max_{1\le i\le N}\frac{1}{\sqrt{T}}\sum_{t=1}^T \{X_{it}-E[X_{it}|\alpha_i]\}\le y \mid (\alpha_i)_i\right)\nonumber\\
	\le& P\left( \max_{1\le i\le N}\frac{1}{\sqrt{T}}\sum_{l=1}^m \hat L_{il}\le y+o(1) \mid (\alpha_i)_i\right)+ o(1) .\label{eq:decoupling}
	\end{align}
	We will prove only the second inequality as the first one follows from a mirrored argument.
	By \eqref{eq:tw_proof_decomposition}, we have
	\begin{align}\label{eq:tw_proof_decomposition2}
	&\max_{i} \frac{1}{\sqrt{T}} \sum_{t=1}^T \{X_{it}-E[X_{it}|\alpha_i]\}
	\\
	\le&
	\left|	\max_{1\le i \le N} \frac{1}{\sqrt{T}} \sum_{t=1}^T \{X_{it}-E[X_{it}|\alpha_i]\} - \max_{1\le i \le N} \frac{1}{\sqrt{T}}\sum_{l=1}^m L_{il}\right|+\max_{1\le i \le N} \frac{1}{\sqrt{T}}\sum_{l=1}^m L_{il}\nonumber\\
	\le&
	\left|	\max_{1\le i \le N} \frac{1}{\sqrt{T}} \sum_{l} S_{il}\right|+ \left|\max_{1\le i \le N} \frac{1}{\sqrt{T}} S_{i,(m+1)}\right|+\max_{1\le i \le N} \frac{1}{\sqrt{T}}\sum_{l=1}^m L_{il}.
	\end{align}
	Applying Corollary 2.7 in \cite{yu1994rates}, we have
	\begin{align}
	&\sup_{y\in \mathbb R}\left|P\left(\max_{1\le i\le N} \sum_{l=1}^m L_{il}\le y\mid (\alpha_i)_i\right)-P\left(\max_{1\le i\le N} \sum_{l=1}^m \hat L_{il}\le y\mid (\alpha_i)_i\right)\right|\le (m-1)\beta(s),
	\label{eq:tw_proof_yu1}
	\\
	&\sup_{y>0}\left|P\left(\max_{1\le i\le N} \left|\sum_{l=1}^m S_{il}\right|> y\mid (\alpha_i)_i\right)-P\left(\max_{1\le i\le N} \left|\sum_{l=1}^m \hat S_{il}\right|> y\mid (\alpha_i)_i\right)\right|\le (m-1)\beta(q).
	\label{eq:tw_proof_yu2}
	\end{align}
	Therefore, for every $\delta_1,\delta_2>0$, \eqref{eq:tw_proof_decomposition} yields
	\begin{align}
	&P\left(\max_{1\le i\le N}\frac{1}{\sqrt{T}}\sum_{t=1}^T \{X_{it}-E[X_{it}|\alpha_i]\}\le y\mid (\alpha_i)_i\right) \notag\\
	\le&
	P\left(\max_{1\le i\le N}\frac{1}{\sqrt{T}}\sum_{l=1}^m \hat L_{il}\le y +\delta_1+\delta_2\mid (\alpha_i)_i\right)
	+P\left(\max_{1\le i\le N}\left|\frac{1}{\sqrt{T}}\sum_{l=1}^m \hat S_{il}\right|> \delta_1\mid (\alpha_i)_i\right) \notag\\
	&+P\left(\max_{1\le i\le N}\left|\frac{1}{\sqrt{T}} S_{i,(m+1)}\right|> \delta_2\mid (\alpha_i)_i\right)+ 2 (m-1)\beta(s) \notag\\
	=:&(i)+(ii)+(iii)+(iv).
	\label{eq:tw_proof_decomposition3}
	\end{align}
	
	Recall that $m\beta(s)=o(1)$, and thus $(iv)=o(1)$.
	Second, we bound the term $(iii)$ in \eqref{eq:tw_proof_decomposition3} by noting that $S_{i,(m+1)}$ consists of a sum over $s$ terms, each of which is bounded in modulus by $2K$. 
	Thus
	\begin{align*}
	\max_{1\le i \le N}	\left|\frac{1}{\sqrt{T}} S_{i,(m+1)}\right|\lesssim \frac{2sK}{\sqrt{T}}=o(1),
	\end{align*}
	and hence $(iii)=o(1)$ for any $\delta_2>0$. 
	
	Next, we bound the term $(ii)$ in \eqref{eq:tw_proof_decomposition3}. Since $\|X_{it}\|_\infty\le K$ and $\hat S_{il}$ is the sum of $s$ terms, it follows that $\|\hat S_{il}\|_\infty \le 2sK$.
	Then, 
	applying Theorem 2.14.1. in \cite{vdVW1996} then yields
	\begin{align*}
	 E\left[\max_{1\le i\le N}\left|\frac{1}{\sqrt{T}}\sum_{l=1}^m \hat S_{il}\right|\mid (\alpha_i)_i\right]\lesssim&
	\sqrt{\frac{s^2\log N}{q}}=o(1).
	\end{align*}
	Markov's inequality then implies that the term $(ii)$ is $o(1)$.
	
	As $\delta_1, \delta_2$ are arbitrary, this verifies Equation (\ref{eq:decoupling}).
	
	Now, by the independence of $(\hat L_{il})_{l=1}^m$ and $\|\hat L_{il}\|_\infty\le 2q K$, Theorem 2.14.1 in \cite{vdVW1996} can be applied to yield that
	\begin{align*}
	E\left[\max_{1\le i \le N}\left|\frac{1}{\sqrt{T}}\sum_{l=1}^m \hat L_{il}\right|\right] \lesssim \sqrt{\frac{m}{T}}\cdot\sqrt{q^2\log N} \lesssim\sqrt{q\log N}.
	\end{align*}
	In the light of
	\begin{align*}
	\max_{1\le i \le N} \frac{1}{\sqrt{T}}\sum_{t=1}^T \{X_{it}-E[X_{it}|\alpha_i]\}\stackrel{d}{=}	\max_{1\le i \le N} \frac{1}{\sqrt{T}}\sum_{l=1}^m \hat L_{il},
	\end{align*}
	implied by  Equation (\ref{eq:decoupling}),
	we now conclude
	\begin{align*}
	E\left[\max_{i=1,...,N}\left|\frac{1}{T} \sum_{t'=1}^T X_{it'} - E[X_{it}|\alpha_i]\right|\right]\lesssim \sqrt{\frac{q\log N}{T}}=o(1).
	\end{align*}

	Note that  
	\begin{align}\label{eq:tw_proof_convrate}
	\frac{1}{NT}\sum_{i=1}^N\sum_{t=1}^T X_{it}=E[X_{it}]+O\left(\frac{1}{\sqrt{N\wedge T}}\right)
	\end{align}
	by our Theorem \ref{thm:main_asymptotics}.
	
	By combining the above uniform rates, under non-degeneracy, we have
	\begin{align*}
	\frac{1}{NT}\sum_{i=1}^N\sum_{t=1}^T \ddot{X}_{it} U_{it}
	=&\frac{1}{NT}\sum_{i=1}^N\sum_{t=1}^T \left(X_{it}-\frac{1}{N}\sum_{i'=1}^N X_{i't} - \frac{1}{T}\sum_{t'=1}^TX_{it'}+\frac{1}{NT}\sum_{i'=1}^N\sum_{t'=1}^TX_{i't'}\right) 
	U_{it}\\
	=&\frac{1}{NT}\sum_{i=1}^N\sum_{t=1}^T \left(X_{it}-E[X_{it}|\gamma_t]- E[X_{it}|\alpha_i]+E[X_{it}]\right)U_{it}  +o_p(1)\cdot\frac{1}{NT}\sum_{i=1}^N\sum_{t=1}^T U_{it}\\
	=&\frac{1}{NT}\sum_{i=1}^N\sum_{t=1}^T\tilde X_{it} U_{it}+ o_p(1)\cdot\frac{1}{NT}\sum_{i=1}^N\sum_{t=1}^T U_{it}\\
	=&\frac{1}{NT}\sum_{i=1}^N\sum_{t=1}^T\tilde X_{it} U_{it}+ o_p\left(\frac{1}{\sqrt{N\wedge T}}\right),
	\end{align*}
	where the first equality follows by the definition of $\ddot X_{it}$,
	the second equality follows by \eqref{eq:tw_proof_gamma}, \eqref{eq:tw_proof_alpha} and \eqref{eq:tw_proof_convrate},
	the third equality follows by the definition of $\tilde X_{it}$, and
	the fourth equality uses $(NT)^{-1}\sum_{i=1}^N\sum_{t=1}^T U_{it}=O_p((N\wedge T)^{-1/2})$ which follows from our Theorem \ref{thm:main_asymptotics}.
	
	Following a similar decomposition and a crude calculation, we have 
	\begin{align*}
	\frac{1}{NT}\sum_{i=1}^N\sum_{t=1}^T \ddot X_{it}\ddot X_{it}'=\frac{1}{NT}\sum_{i=1}^N\sum_{t=1}^T \tilde X_{it}\tilde X_{it}'+o_p(1).
	\end{align*}

	We have shown that
	\begin{align*}
	\sqrt{N\wedge T}(\hat \beta - \beta) =\left(\frac{1}{NT}\sum_{i=1}^N \sum_{t=1}^T\tilde X_{it}\tilde X_{it}'\right)\frac{1}{\sqrt{N\wedge T}}\sum_{i=1}^N \sum_{t=1}^T\tilde X_{it} U_{it} + o_p(1).
	\end{align*}
	Note that the sums have summands of the forms, $\tilde X_{it} \tilde X_{it}'$ and $\tilde X_{it} U_{it}$, which only depend on $(\alpha_i,\gamma_t,\varepsilon_{it})$.
	Thus, our Theorems \ref{thm:main_asymptotics} and \ref{thm:var_consistency} can be applied. 
	By replicating the Proof of Theorem \ref{thm:OLS}, we have the desired result for the non-degenerate case.

	Similarly, under i.i.d. sampling, by applying Theorem 2.14.1 in \cite{vdVW1996}, we have
	\begin{align*}
	E	\left[\max_{t=1,...,T}\left|\frac{1}{N} \sum_{i'=1}^N X_{it} - E[X_{it}|\gamma_t]\right|\right]
	\bigvee
	E\left[\max_{i=1,...,N}\left|\frac{1}{T} \sum_{t'=1}^T X_{it'} - E[X_{it}|\alpha_i]\right|\right]
	\\
	=O\left(\sqrt{\frac{\log N T}{N\wedge T}}\right)=o(1),
	\end{align*}
	and thus
	\begin{align*}
	\frac{1}{NT}\sum_{i=1}^N\sum_{t=1}^T \ddot{X}_{it} U_{it}=&\frac{1}{NT}\sum_{i=1}^N\sum_{t=1}^T \left(X_{it} -\frac{1}{N}\sum_{i'=1}^NX_{i't} -\frac{1}{T}\sum_{t'=1}^TX_{it'}+\frac{1}{NT}\sum_{i'=1}^N\sum_{t'=1}^TX_{i't'}\right) U_{it}\\
	=&
	\frac{1}{NT}\sum_{i=1}^N\sum_{t=1}^T \tilde X_{it}U_{it}  +o_p(1)\cdot \frac{1}{NT}\sum_{i=1}^N\sum_{t=1}^T U_{it}
	\\
	=&\frac{1}{NT}\sum_{i=1}^N\sum_{t=1}^T\tilde X_{it}U_{it} + o_p\left(\frac{1}{\sqrt{NT}}\right),
	\end{align*}
	where the first sum has mean zero and has its summands depending only on $\varepsilon_{it}$ and thus is i.i.d. over $i$ and $t$. 
	The rest follows from the same arguments in the non-degenerate case.
\end{proof}

\setcounter{section}{5}
\section{Technical Lemmas}\label{sec:lemmas}
This section contains key technical lemmas for consistency of variance estimation under the current asymptotic setting. 
Throughout this section, for any $a,b\in\mathbb R^+\cup\{0\}$, we use the short-hand notation $a\lesssim b$ to indicate $a\le C b$ for some $C<\infty$ independent of $(N,T)$. 
\subsection{Generalized Newey-West Estimator under Non-Degeneracy}

\begin{lemma}[Generalized Newey-West Estimator under Non-Degeneracy]\label{lem:GNW}
	If Assumption \ref{a:NW_OLS} holds with (iv)(1), then
	\begin{align*}
	\frac{1}{N^2 T}\sum_{m=1}^{M}w(m,M)(\hat G_m' -\hat H_m')
	&=\frac{1}{N^2 T}\sum_{m=1}^{M}w(m,M)\sum_{t=m+1}^{T} \sum_{i=1}^N\sum_{i'\ne i}^N  X_{it}\hat U_{it}\hat U_{i',t-m}X_{i',t-m}'\\ &\stackrel{p}{\to} \sum_{m=1}^{\infty}  E[X_{it}U_{it}U_{i',t-m}X_{i',t-m}'] \\
	&=\sum_{m=1}^\infty E[b_t b_{t-m}']
	\end{align*}
	for $i\ne i'$.
\end{lemma}

\begin{proof}
	A sequence of symmetric matrices $A_n$ converges to a symmetric matrix $A_0$ if and only if $b'A_nb\to b'A_0b$ for all comfortable $b$. Therefore, it suffices to assume without the loss of generality that $k=1$.
	
	First notice that by the law of total covariance as well as Assumption \ref{a:NW_OLS}(i), for $m=1,...,M$ and $T=m+1,...,T$, it holds that
	\begin{align*}
	&E[X_{it} U_{it} U_{i't-m} X_{i',t-m}]\\
	=&  cov (X_{it}U_{it},X_{i',t-m}U_{i',t-m})\\
	=&cov (E[X_{it}U_{it}\mid \gamma_t,\gamma_{t-m}],E[X_{i',t-m}U_{i',t-m}\mid \gamma_t,\gamma_{t-m}] )
	+E[cov( X_{i,t}U_{i,t},X_{i',t-m}U_{i',t-m}\mid \gamma_{t},\gamma_{t-m})]\\
	=& cov (E[X_{it}U_{it}\mid \gamma_t],E[X_{i,t-m}U_{i,t-m}\mid \gamma_{t-m}] )+0
	=E[b_tb_{t-m}],
	\end{align*}
	where the second to the last equality follows from the independence of $\alpha_i$.
	This verifies the last equality on the right hand side of the statement.
	To show the convergence in probability, consider the decomposition
	\begin{align}
	&\left|\sum_{m=1}^{M}\frac{w(m,M)}{N^2 T} \sum_{t=m+1}^{T} \sum_{i=1}^N\sum_{i'\ne i}^N  X_{it}\hat U_{it}\hat U_{i',t-m}X_{i',t-m} - \sum_{m=1}^{\infty}  E[X_{it}U_{it}U_{i',t-m}X_{i',t-m}]\right|\nonumber\\
	\le 
	&\left|\sum_{m=1}^{M}\frac{w(m,M)}{N^2 T} \sum_{t=m+1}^{T} \sum_{i=1}^N\sum_{i'\ne i}^N \left\{X_{it}\hat U_{it}\hat U_{i',t-m}X_{i',t-m} -   X_{it} U_{it} U_{i',t-m}X_{i',t-m} \right\}\right|\nonumber\\
	&+
	\left|\sum_{m=1}^{M}\frac{w(m,M)}{N^2 T}\sum_{t=m+1}^{T}  \sum_{i=1}^N\sum_{i'\ne i}^N \left\{  X_{it} U_{it} U_{i',t-m}X_{i',t-m}-E[ X_{it} U_{it} U_{i',t-m}X_{i',t-m}]\right\}\right|\nonumber\\
	&+\left|\sum_{m=1}^{M}\frac{1}{N^2 T}|w(m,M)-1| \sum_{t=m+1}^{T} \sum_{i=1}^N\sum_{i'\ne i}^N E[ X_{it} U_{it} U_{i',t-m}X_{i',t-m}]\right|\nonumber\\
	&+\left|\sum_{m= M+1}^{\infty}\frac{1}{N^2 T}\sum_{t=m+1}^{T}  \sum_{i=1}^N\sum_{i'\ne i}^N E[ X_{it} U_{it} U_{i',t-m}X_{i',t-m}]\right|+o_p(1)\nonumber\\
	=:&(1)+(2)+(3)+(4)+o_p(1).\label{eq:GNW_decomp}
	\end{align}
	Note that we have used the fact that $$\sum_{m=1}^{\infty}  E[X_{it}U_{it}U_{i',t-m}X_{i',t-m}]=\frac{1}{N^2 T}\sum_{m=1}^{\infty} \sum_{t=m+1}^{T} \sum_{i=1}^N\sum_{i'\ne i}^N E[X_{it} U_{it} U_{i',t-m}X_{i',t-m}]+o(1)$$
	under Assumption \ref{a:NW_OLS} (i).
	It suffices to show that each of the four terms, (1)--(4), is asymptotically negligible.
	
	First, consider term $(2)$.
	Define
	\begin{align*}
	Z_{tm}&=\frac{1}{N^2}\sum_{i=1}^N\sum_{i'\ne i}^N \left\{   X_{it} U_{it} U_{i',t-m}X_{i',t-m}-E[ X_{it} U_{it} U_{i',t-m}X_{i',t-m}\mid (\alpha_i)_{i=1}^N]\right\} \qquad\text{and}\\
	\tilde Z_{tm}&=\frac{1}{N^2}\sum_{i=1}^N\sum_{i'\ne i}^N \left\{  E[X_{it} U_{it} U_{i',t-m}X_{i',t-m}\mid (\alpha_i)_{i=1}^N]-E[ X_{it} U_{it} U_{i',t-m}X_{i',t-m}]\right\}
	\end{align*}
	for each $N$, $t$ and $m$.
	With this notation, we can bound term $(2)$ as
	\begin{align}\label{eq:further_decomposition2}
	(2)	\le \left|\sum_{m=1}^{M}\frac{ w(m,M)}{T}\sum_{t=m+1}^T Z_{tm}\right|+\left|\sum_{m=1}^{M}\frac{ w(m,M)}{T}\sum_{t=m+1}^T \tilde Z_{tm}\right|.
	\end{align}
	
	Consider the first term on the right-hand side of \eqref{eq:further_decomposition2}.
	Observe that $E[Z_{tm}\mid (\alpha_i)_{i=1}^N]=0$. By Theorem 14.2 in \cite{davidson1994stochastic} with $r=2(r+\delta)$ (where $r$ on the left-hand side is in terms of the notation by \cite{davidson1994stochastic}, and $r$ and $\delta$ on the right-hand side satisfy our Assumption \ref{a:NW_OLS}) and $p=2$,
	\begin{align*}
	\left\{	E\left[ \left|E\left[ Z_{tm} \mid (\alpha_i)_{i=1}^N, \calF_{-\infty}^{t-\ell}  \right] \right|^{2}\mid (\alpha_i)_{i=1}^N\right]\right\}^{1/2} \le  6 \alpha(\ell)^{1/2-1/2(r+\delta)} \left\{ E[|Z_{tm}|^{2(r+\delta)}\mid (\alpha_i)_{i=1}^N] \right\}^{1/2(r+\delta)} 
	\end{align*}
	almost surely.
	Then, by Lemma A in \cite{hansen1992consistent} with $\beta=2$,
	\begin{align*}
	\left\{ E\left[ \left|\frac{1}{T}\sum_{t=m+1}^{T}Z_{tm}\right|^2\mid (\alpha_i)_{i=1}^N\right] \right\}^{1/2}
	&\lesssim \frac{1}{T} \sum_{\ell=1}^\infty \alpha(\ell)^{1/2-1/2(r+\delta)} \left\{\sum_{t=m+1}^T \left( E[|Z_{tm}|^{2(r+\delta)}\mid (\alpha_i)_{i=1}^N] \right)^{1/(r+\delta)} \right\}^{1/2}\\
	&\lesssim T^{-1/2}  \left\{ E[|Z_{tm}|^{2(r+\delta)}\mid (\alpha_i)_{i=1}^N] \right\}^{1/2(r+\delta)}
	\end{align*}
	for each $m\ge 0$. 
	Here, we have used the boundedness $\sum_{\ell=1}^\infty \alpha(\ell)^{1/2-1/2(r+\delta)} <\infty$ implied by Assumption \ref{a:NW_OLS} (iii). 
	By Minkowski's inequality and the inequality obtained above, we have
	\begin{align*}
	&\frac{T^{1/2}}{M} \left\{E\left[\left| \sum_{m=1}^{M} \frac{w(m,M)}{T}\sum_{t=m+1}^T  Z_{tm}\right|^{2}\mid (\alpha_i)_{i=1}^N\right]\right\}^{1/2}\\
	&\le
	\frac{T^{1/2}}{M} \sum_{m=1}^{M} |w(m,M)|   \left\{ E\left[ \left|\frac{1}{T}\sum_{t=m+1}^{T}Z_{tm}\right|^2\mid (\alpha_i)_{i=1}^N\right] \right\}^{1/2}\\
	& \lesssim 
	\left\{ E[|Z_{tm}|^{2(r+\delta)}\mid (\alpha_i)_{i=1}^N] \right\}^{1/2(r+\delta)}
	\end{align*}
	uniformly in $T$.
	By Markov's inequality, for any $\varepsilon>0$,
	\begin{align*}
	P\left(\left|\sum_{m=1}^{M}\frac{ w(m,M)}{T}\sum_{t=m+1}^T Z_{tm}\right|>\varepsilon \mid (\alpha_i)_{i=1}^N\right)
	=O\left( E\left[\left| \sum_{m=1}^{M} \frac{w(m,M)}{T}\sum_{t=m+1}^T  Z_{tm}\right|^{2}\mid (\alpha_i)_{i=1}^N\right]\right)
	\end{align*}
	almost surely.
	Thus, by Fubini theorem, Jensen's inequality, and the bounded moment $\left\{ E[|Z_{tm}|^{2(r+\delta)}] \right\}^{1/2(r+\delta)}<\infty$
	following Assumption \ref{a:NW_OLS} (ii), we have
	\begin{align*}
	P\left(\left|\sum_{m=1}^{M} \frac{w(m,M)}{T}\sum_{t=m+1}^T Z_{tm}\right|>\varepsilon\right)
	=O\left( \frac{M^2}{T}\right)=o(1)
	\end{align*}
	as $M^2/T=o(1)$ under Assumption \ref{a:NW_OLS} (vi).
	
	Next, consider the second term on the right-hand side of \eqref{eq:further_decomposition2}.
	By Minkowski's inequality, 
	\begin{align*}
	\frac{N^{1/2}}{M}\left\{E\left[\left| \sum_{m=1}^{M} \frac{w(m,M)}{T}\sum_{t=m+1}^T  \tilde Z_{tm}\right|^{2}\right]\right\}^{1/2}
	\le
	\frac{N^{1/2}}{M}\sum_{m=1}^{M} |w(m,M)| \left\{ E\left[ \left|\frac{1}{T}\sum_{t=m+1}^{T}\tilde Z_{tm}\right|^2\right] \right\}^{1/2} <\infty
	\end{align*}
	uniformly in $T$.
	To see the last inequality, set $\theta=0$ without loss of generality.
	By the identical distribution of $\gamma_t$,
	\begin{align*}
	&E\left[ \left|\frac{1}{T}\sum_{t=m+1}^{T}\tilde Z_{tm}\right|^2\right]= O\left(E[|\tilde Z_{tm}|^2]\right).
	\end{align*}
	Now, fix any $m$ and
	denote $W_{ii'}=E[X_{i,t+m}U_{i,t+m}U_{i' t}X_{i' t}]$.
	Then
	\begin{align*} 
	E[|\tilde Z_{tm}|^2]&=\frac{1}{N^4}\sum_{i=1}^N \sum_{i'\ne i}^N \sum_{\iota=1}^N  \sum_{\iota'\ne \iota}^N  E\Bigg[ \Bigg( E[W_{ii'} \mid (\alpha_i)_{i=1}^N] -E[W_{ii'} ]\Bigg)\Bigg(
	E[W_{\iota\iota'}\mid (\alpha_i)_{i=1}^N]-E[W_{\iota\iota'} ]
	\Bigg)\Bigg]\\
	&= 
	\frac{1}{N^4}\sum_{i=1}^N \sum_{i'\ne i}^N \sum_{\iota=1}^N  \sum_{\iota'\ne \iota}^N  E\Bigg[ \Bigg( E[W_{ii'} \mid \alpha_{i},\alpha_{i'}] -E[W_{ii'} ]\Bigg)\Bigg(
	E[W_{\iota\iota'}\mid \alpha_{\iota},\alpha_{\iota'}]-E[W_{\iota\iota'} ]
	\Bigg)\Bigg]\\
	&\le 
	\frac{2}{N^4}\sum_{i=1}^N \sum_{i'\ne i}^N  \sum_{\iota'\ne i}^N  E\Bigg[ \Bigg( E[W_{ii'} \mid \alpha_{i},\alpha_{i'}] -E[W_{ii'} ]\Bigg)\Bigg(
	E[W_{i\iota'}\mid \alpha_{i},\alpha_{\iota'}]-E[W_{i\iota'} ]
	\Bigg)\Bigg]\\
	&+\frac{2}{N^4}\sum_{i=1}^N \sum_{i'\ne i}^N  \sum_{\iota=1}^N  E\Bigg[ \Bigg( E[W_{ii'} \mid \alpha_{i},\alpha_{i'}] -E[W_{ii'} ]\Bigg)\Bigg(
	E[W_{\iota i}\mid \alpha_{\iota},\alpha_{i}]-E[W_{\iota i} ]
	\Bigg)\Bigg]\\
	&+\frac{1}{N^4}\sum_{i=1}^N \sum_{i'\ne i}^N    E\Bigg[ \Bigg( E[W_{ii'} \mid \alpha_{i},\alpha_{i'}] -E[W_{ii'} ]\Bigg)^2\Bigg]\\
	&\le C\left(\frac{1}{N} E[|X_{i,t+m}U_{i,t+m}U_{\iota t}X_{\iota t}|^2]\right)\\
	&\le  C\left(\frac{1}{N} E[|X_{it}U_{it}|^4]\right)
	=O\left(\frac{1}{N}\right)
	\end{align*}
	for some constant $C>0$,
	following Assumption \ref{a:NW_OLS} (ii) and Jensen's inequality. Note that the second inequality holds since $ E[W_{ii'} \mid \alpha_{i},\alpha_{i'}]$ and $ E[W_{\iota\iota'} \mid \alpha_{\iota},\alpha_{\iota'}]$ are independent when $(i,i')\ne(\iota,\iota ')$. This bound holds uniformly over all $ m=1,...,M$.
	An application of Markov's inequality combined with the above calculations yields
	\begin{align*}
	P\left(\left|\sum_{m=1}^{M} \frac{w(m,M)}{T}\sum_{t=m+1}^T \tilde Z_{tm}\right|>\varepsilon\right)
	=O\left( \frac{M^2}{N}\right)=o(1).
	\end{align*}
	
	Now, take term $(1)$, which we bound as follows.
	\begin{align}\label{eq:further_decomposition1}
	(1)\lesssim
	\Bigg|\sum_{m=1}^{M}\frac{w(m,M)}{N^2 T} \sum_{t=m+1}^{T} \sum_{i=1}^N\sum_{i'\ne i}^N X_{it}\Bigg\{(\beta -\hat \beta)X_{it} U_{i',t-m}+U_{it} X_{i',t-m}(\beta -\hat \beta)\nonumber\\
	\qquad\qquad\qquad\qquad\qquad\qquad\qquad+X_{it} X_{i',t-m}(\beta -\hat \beta)^2 \Bigg\}X_{i',t-m}\Bigg|.
	\end{align}
	The first term on the right-hand side of \eqref{eq:further_decomposition1} is bounded by
	\begin{align*}
	|\beta -\hat \beta|\Bigg|\sum_{m=1}^{M}\frac{w(m,M)}{N^2 T} \sum_{t=m+1}^{T} \sum_{i=1}^N\sum_{i'\ne i}^N X_{it}^2 U_{i',t-m}X_{i',t-m}'\Bigg|=O_p\left(\frac{M}{\sqrt{\min\{N,T\}}}\right)=o_p(1)
	\end{align*}
	under Assumption \ref{a:NW_OLS} (ii) and (vi).
	Here, we have used $|\hat \beta-\beta|=O_p((\min\{N,T\})^{-1/2})$ implied by the central limit theorem under Assumption \ref{a:NW_OLS}. 
	A similar argument applies to the second term on the right-hand side of \eqref{eq:further_decomposition1}.
	
	Next, consider term $(3)$, which equals
	\begin{align*}
	(3)=\left|\sum_{m=1}^{M}|w(m,M)-1|\frac{1}{ T} \sum_{t=m+1}^{T}  E[ X_{it} U_{it} U_{i',t-m}X_{i',t-m}]\right|.
	\end{align*}
	Applying Theorem 14.13.2 in \cite{hansen2021econometrics} conditional on $(\alpha_i)_{i=1}^N$, we have
	\begin{align*}
	& \left|E\left[   X_{it} U_{it} U_{i',t-m}X_{i',t-m}\mid (\alpha_i)_{i=1}^N\right]\right|\\
	&			\lesssim
	\alpha(m)^{(1+2\delta)/(4+4\delta)}\left\{ E[| X_{it} U_{it}  |^{4(r+\delta)}\mid  \alpha_i]\right\}^{1/4(r+\delta)}\left\{ E[|  U_{i',t-m}X_{i',t-m}|^{2}\mid  \alpha_{i'}]\right\}^{1/2}
	\end{align*}
	with $\sum_{m=1}^\infty  \alpha(m)^{(1+2\delta)/(4+4\delta)}<\infty$ following Assumption \ref{a:NW_OLS} (iii). 
	Recall that $\alpha_i$ are i.i.d. 
	By integrating out $\alpha_i$ and $\alpha_{i'}$ and applying Jensen's inequality, we have
	\begin{align*}
	\left|E[ X_{it} U_{it} U_{i',t-m}X_{i',t-m}]\right|	
	&\le E\Big[\Big|E\left[   X_{it} U_{it} U_{i',t-m}X_{i',t-m}\mid (\alpha_i)_{i=1}^N\right]\Big|\Big]\\
	&\lesssim
	\alpha(m)^{(1+2\delta)/(4+4\delta)}\left\{ E[|X_{it}U_{it}|^{4(r+\delta)}]\right\}^{1/4(r+\delta)}\left\{ E[|X_{it}U_{it}|^{2}]\right\}^{1/2}
	\\
	&\lesssim
	\alpha(m)^{(1+2\delta)/(4+4\delta)}
	\end{align*}
	for all $m$.
	Since $w(m,M)\to 1$ as $T\to \infty$ for each $m$, the dominated convergence theorem, the above bound, and Assumption \ref{a:NW_OLS} (iii) together imply $(3)=o(1)$.
	
	Finally, consider term $(4)$.
	By the law of total covariance, one can rewrite $(4)$ as
	\begin{align*}
	(4)=O\left(\left|\sum_{m= M+1}^{\infty}  cov ( X_{it} U_{it}, X_{i',t-m}U_{i',t-m})\right|\right)=O\left(\left|\sum_{m= M+1}^{\infty}  cov (E[X_{it} U_{it}\mid \gamma_{t}],E[X_{i',t-m}U_{i',t-m}\mid \gamma_{t-m}] )\right|\right).
	\end{align*}
	Thus, $(4)=o(1)$ follows from an application of Lemma 6.17 in \cite{white1984asymptotic} as $m\to\infty$. 
\end{proof}

\subsection{Eicker-White CRVE under Non-Degeneracy}

\begin{lemma}[Eicker-White CRVE under non-degeneracy]\label{lem:EW_CRVE}
	If Assumption \ref{a:NW_OLS} holds with (iv)(1), then 
	\begin{align}
	\frac{1}{N^2T}\sum_{t=1}^T \hat S_t \hat S_t' =&\frac{1}{N^2T}\sum_{t=1}^T\sum_{i=1}^N \sum_{i'=1}^N X_{it}\hat U_{it}\hat U_{i't}X_{i't}'\nonumber\\
	& \stackrel{p}{\to}
	E\left[ (E[X_{it}U_{it}\mid \gamma_t]) (E[X_{it}U_{it}\mid \gamma_t])'\right] \nonumber\\
	& = E[b_t b_t'], \label{eq:lem:EW_CRVE1}
	\end{align}
	and
	\begin{align}
	\frac{1}{NT^2}\sum_{i=1}^N \hat R_i \hat R_i'=&\frac{1}{NT^2} \sum_{i=1}^N\sum_{t=1}^T \sum_{t'=1}^T X_{it}\hat U_{it}\hat U_{it'}X_{it'}'\nonumber\\
	&\stackrel{p}{\to}E\left[ (E[X_{it}U_{it}\mid \alpha_i]) (E[X_{it}U_{it}\mid \alpha_i])'\right] \nonumber\\
	&= E[a_ia_i'].\label{eq:lem:EW_CRVE2}
	\end{align}
\end{lemma}

\begin{proof}
	
	Throughout the proof, assume $k=1$  without loss of generality. 
	\bigskip\\
	\textbf{Proof of \eqref{eq:lem:EW_CRVE1}:}
	Notice that we have
	\begin{align*}
	E[ X_{it}U_{it} U_{i't}X_{i't}]&=cov(X_{it}U_{it},U_{i't}X_{i't})\\
	&=E[cov(X_{it}U_{it},U_{i't}X_{i't}|\gamma_t)]+ cov(E[X_{it}U_{it}|\gamma_{t}],E[X_{i't}U_{i't}|\gamma_{t}])\\
	&=0+E\left[ (E[X_{it}U_{it}\mid \gamma_t]) (E[X_{it}U_{it}\mid \gamma_t])\right]
	\end{align*}
	by the law of total covariance.
	Thus, it suffices to bound the right-hand side of
	\begin{align}
	&\left|\frac{1}{N^2T}\sum_{t=1}^T\sum_{i=1}^N \sum_{i'=1}^N X_{it}\hat U_{it}\hat U_{i't}X_{i't}- E[ X_{it}U_{it} U_{i't}X_{i't}]\right|\notag\\
	\le& 
	\left|\frac{1}{N^2T}\sum_{t=1}^T \sum_{i=1}^N \sum_{i'=1}^N\{ X_{it}\hat U_{it}\hat U_{i't}X_{i't}-X_{it} U_{it} U_{i't}X_{i't}\}\right|+\notag\\
	&
	\left|\frac{1}{N^2T}\sum_{t=1}^T \sum_{i=1}^N \sum_{i'=1}^N X_{it} U_{it} U_{i't}X_{i't}- E[ X_{it}U_{it} U_{i't}X_{i't}]\right|
	=:(1)+(2).
	\label{eq:proof_first_statement1}
	\end{align}
	
	First, consider term $(2)$. 
	Set
	\begin{align*}
	Z_{t}&=\frac{1}{N^2}\sum_{i=1}^N\sum_{i'=1}^N \left\{   X_{it} U_{it} U_{i't}X_{i't}-E[ X_{it} U_{it} U_{i't}X_{i't}\mid (\alpha_i)_{i=1}^N]\right\} \qquad\text{and}\\
	\tilde Z_{t}&=\frac{1}{N^2}\sum_{i=1}^N\sum_{i'= 1}^N \left\{  E[X_{it} U_{it} U_{i't}X_{i't}\mid (\alpha_i)_{i=1}^N]-E[ X_{it} U_{it} U_{i't}X_{i't}]\right\}
	\end{align*}
	for each $t$, $N$. 
	We decompose term $(2)$ as
	\begin{align}
	(2)\le& \left|\frac{1}{T}\sum_{t=1}^T Z_t\right|+ \left|\frac{1}{T}\sum_{t=1}^T \tilde Z_t\right|=:(3)+(4).
	\label{eq:proof_first_statement2}
	\end{align}
	
	Second, consider term $(3)$.
	Observe that $E[Z_t\mid (\alpha_i)_{i=1}^N ]=0$.
	By Theorem 14.2 in \cite{davidson1994stochastic} with $r=2(r+\delta)$ (where $r$ on the left-hand side is in terms of the notation by \cite{davidson1994stochastic}, and $r$ and $\delta$ on the right-hand side satisfy our Assumption \ref{a:NW_OLS}) and $p=2$, we have
	\begin{align*}
	\left\{	E\left[ \left|E\left[ Z_{t} \mid (\alpha_i)_{i=1}^N, \calF_{-\infty}^{t-\ell}  \right] \right|^{2}\mid (\alpha_i)_{i=1}^N\right]\right\}^{1/2} \le  6 \alpha(\ell)^{1/2-1/2(r+\delta)} \left\{ E[|Z_{t}|^{2(r+\delta)}\mid (\alpha_i)_{i=1}^N] \right\}^{1/2(r+\delta)} 
	\end{align*}
	almost surely.
	Then, by Lemma A in \cite{hansen1992consistent} with $\beta=2$, it holds that
	\begin{align*}
	\left\{ E\left[ \left|\frac{1}{T}\sum_{t=1}^{T}Z_{t}\right|^2\mid (\alpha_i)_{i=1}^N\right] \right\}^{1/2}
	&\lesssim \frac{1}{T} \sum_{\ell=1}^\infty \alpha(\ell)^{1/2-1/2(r+\delta)} \left\{\sum_{t=1}^T \left( E[|Z_{t}|^{2(r+\delta)}\mid (\alpha_i)_{i=1}^N] \right)^{1/(r+\delta)} \right\}^{1/2}\\
	&\lesssim T^{-1/2}  \left\{ E[|Z_{t}|^{2(r+\delta)}\mid (\alpha_i)_{i=1}^N] \right\}^{1/2(r+\delta)}.
	\end{align*}
	Thus, we have
	\begin{align*}
	T^{1/2} \left\{E\left[\left|  \frac{1}{T}\sum_{t=1}^T  Z_{t}\right|^{2}\mid (\alpha_i)_{i=1}^N\right]\right\}^{1/2}
	&\lesssim 
	\left\{ E[|Z_{tm}|^{2(r+\delta)}\mid (\alpha_i)_{i=1}^N] \right\}^{1/2(r+\delta)}
	\end{align*}
	uniformly in $T$.
	By Markov's inequality, we obtain for any $\varepsilon>0$
	\begin{align*}
	P\left(\left|\frac{1}{T}\sum_{t=1}^T Z_{t}\right|>\varepsilon \mid (\alpha_i)_{i=1}^N\right)
	&=O\left( E\left[\left|  \frac{1}{T}\sum_{t=1}^T  Z_{t}\right|^{2}\mid (\alpha_i)_{i=1}^N\right]\right)
	\end{align*}
	almost surely.
	Therefore, by Fubini theorem, Jensen's inequality, and the bounded moment $E[|Z_{t}|^{2(r+\delta)}]<\infty$ that holds under Assumption \ref{a:NW_OLS} (ii),  we have
	\begin{align*}
	P\left(\left| \frac{1}{T}\sum_{t=1}^T Z_{t}\right|>\varepsilon\right)
	&=O\left( \frac{1}{T}\right)=o(1),
	\end{align*}
	showing $(3) = o_p(1)$ in \eqref{eq:proof_first_statement2}.
	
	Third, consider term $(4)$ in \eqref{eq:proof_first_statement2}. 
	By the identical distribution of the $\gamma_t$, we have $E[|T^{-1}\sum_{t=1}^T \tilde Z_t|^2]=O(E[|\tilde Z_t|^2])=O(N^{-1})$, the final equality by a direct calculation. Markov's inequality implies that
	\begin{align*}
	P\left(\left| \frac{1}{T}\sum_{t=1}^T \tilde Z_{t}\right|>\varepsilon\right)
	&=O\left( \frac{1}{N}\right)=o(1),
	\end{align*}
	showing $(4) = o_p(1)$ in \eqref{eq:proof_first_statement2}.
	Thus, we obtain $(2)=o_p(1)$ in \eqref{eq:proof_first_statement1}.
	
	Finally, consider term $(1)$ in \eqref{eq:proof_first_statement1}. Note that 
	\begin{align*}
	(1)\lesssim
	\Bigg|\frac{1}{N^2 T} \sum_{t=m+1}^{T} \sum_{i=1}^N\sum_{i'= 1}^N X_{it}\Bigg\{(\beta -\hat \beta)X_{it} U_{i't}+U_{it} X_{i't}(\beta -\hat \beta)+X_{it} X_{i't}(\beta -\hat \beta)^2 \Bigg\}X_{i't}\Bigg|.
	\end{align*}
	Similar to (\ref{eq:further_decomposition1}), the first two terms are $O_p(|\hat \beta-\beta|)=O_p((\min\{N,T\})^{-1/2})$ while the third term is $O_p(|\hat \beta-\beta|^2)=O_p((\min\{N,T\})^{-1})$. 
	It therefore follows that $(1)=O_p((\min\{N,T\})^{-1/2})$ in \eqref{eq:proof_first_statement1}.
	\quad
	\bigskip\\
	\noindent
	\textbf{Proof of \eqref{eq:lem:EW_CRVE2}:}
	Observe that
	\begin{align}
	&\left|\frac{1}{NT^2} \sum_{i=1}^N\sum_{t=1}^T \sum_{t'=1}^T X_{it}\hat U_{it}\hat U_{it'}X_{it'}-E\left[ (E[X_{it}U_{it}\mid \alpha_i]) (E[X_{it}U_{it}\mid \alpha_i])\right]\right|\notag\\
	&\le
	\left|\frac{1}{NT^2} \sum_{i=1}^N\sum_{t=1}^T \sum_{t'=1}^T \left\{X_{it}\hat U_{it} \hat U_{it'}X_{it'}-X_{it} U_{it}U_{it'}X_{it'}\right\}\right|\notag\\
	&+
	\left|\frac{1}{NT^2} \sum_{i=1}^N\sum_{t=1}^T \sum_{t'=1}^T \{X_{it} U_{it} U_{it'}X_{it'}-E[X_{it} U_{it} U_{it'}X_{it'}]\}\right|\notag\\
	&+
	\left|\frac{1}{NT^2} \sum_{i=1}^N\sum_{t=1}^T \sum_{t'=1}^T E[X_{it} U_{it} U_{it'}X_{it'}]-E\left[ (E[X_{it}U_{it}\mid \alpha_i]) (E[X_{it}U_{it}\mid \alpha_i])\right]\right|
	= (5) +(6)+(7).
	\label{eq:proof_second_statement1}
	\end{align}
	Term $(5)$ can be shown to be $O_p(|\hat \beta- \beta|)=o_p(1)$ similarly to term $(1)$ in \eqref{eq:proof_first_statement1}.
	Consider term $(7)$. By the law of total covariances,
	\begin{align*}
	E[X_{it}U_{it} U_{it'}X_{it'}]=&cov(E[X_{it}U_{it}|\alpha_i],E[X_{it'}U_{it'}|\alpha_i])+E[cov(X_{it}U_{it},X_{it'}U_{it'}|\alpha_i)].
	\end{align*}
	From this equality follows
	\begin{align}
	&\frac{1}{NT^2} \sum_{i=1}^N\sum_{t=1}^T \sum_{t'=1}^T E\left[ X_{it} U_{it} U_{it'}X_{it'}\right] \notag\\&=
	\frac{1}{T^2}\sum_{t=1}^T\sum_{t'=1}^T \left\{   cov(E[X_{it}U_{it}\mid \alpha_i],E[X_{it'}U_{it'}\mid \alpha_i])+E[cov(X_{it}U_{it},X_{it'}U_{it'}\mid \alpha_i) ]\right\} \notag\\
	&=E\left[ (E[X_{it}U_{it}\mid \alpha_i]) (E[X_{it}U_{it}\mid \alpha_i])\right]+ o(1).
	\label{eq:proof_second_statement3}
	\end{align}
	To see the second equality, note that, by Assumption \ref{a:NW_OLS} (i)--(iii) and an application of Theorem 14.13 (ii) in \cite{hansen2021econometrics} along with Jensen's inequality, for any $t,t'\in\{1,...,T\}$, we have
	\begin{align*}
	E[cov(X_{it}U_{it},X_{it'}U_{it'}\mid \alpha_i)]\le 8 \left(E[|X_{it}U_{it}|^{4(r+\delta)}]\right)^{1/2(r+\delta)} \alpha(|t-t'|)^{1-1/2(r+\delta)}.
	\end{align*}
	Moreover, Assumption \ref{a:NW} (iii) implies $\sum_{\ell=1}^\infty \alpha(\ell)^{1-1/2(r+\delta)}<\infty$. 
	Therefore, $$T^{-2}\sum_{t=1}^T\sum_{t'=1}^T  E[cov(X_{it}U_{it},X_{it'}U_{it'}\mid \alpha_1) ]=o(1)$$ follows, which in turn implies the second equality in \eqref{eq:proof_second_statement3}. This shows that $(7)=o(1)$ in \eqref{eq:proof_second_statement1}.
	
	Finally, take term $(6)$ in \eqref{eq:proof_second_statement1}. Define 
	\begin{align*}
	Z_i &=\frac{1}{T^2} \sum_{t=1}^T \sum_{t'=1}^T [X_{it} U_{it} U_{it'}X_{it'}-E\left[X_{it} U_{it} U_{it'}X_{it'} \mid (\gamma_t)_{t=1}^T\right]]
	\end{align*}
	and
	\begin{align*}
	\tilde Z_i &=\frac{1}{T^2} \sum_{t=1}^T \sum_{t'=1}^T E\left[X_{it} U_{it} U_{it'}X_{it'} \mid (\gamma_t)_{t=1}^T\right]-E\left[X_{it} U_{it} U_{it'}X_{it'} \right]\\
	&=\frac{1}{T^2} \sum_{t=1}^T \sum_{t'=1}^T \left\{h_i(\gamma_t,\gamma_{t'})-E[h_i(\gamma_t,\gamma_{t'})]\right\}
	\end{align*}
	where $h_i(\gamma_t,\gamma_{t'})= E\left[X_{it} U_{it} U_{it'}X_{it'} \mid\gamma_t,\gamma_{t'}\right]$.
	We have the bound:
	\begin{align}
	(6) &\le \left|\frac{1}{N} \sum_{i=1}^N Z_i\right|+
	\left|\frac{1}{N} \sum_{i=1}^N\tilde Z_i\right|=:(8)+(9).
	\label{eq:proof_second_statement2}
	\end{align}
	Take $(8)$. Note that conditional on $(\gamma_t)_{t=1}^T$, $Z_i$ are mutually independent and mean zero. Jensen's and Markov's inequalities then imply that $(8)=o_p(1)$. Now, consider term $(9)$. Note that $E[\tilde Z_i]=0$ and
	\begin{align*}
	E\left[\left|\frac{1}{N}\sum_{i=1}^N \tilde Z_i\right|^2\right]\le 	\frac{1}{N^2}\sum_{i=1}^N\sum_{i'=1}^N E\left| \tilde Z_i \tilde Z_{i'}\right|\le E[\tilde Z_i^2],
	\end{align*}
	the final inequality by Cauchy-Schwarz and the fact that $\tilde Z_i$ are identically distributed since $\alpha_i$ are.
	
	By a direct calculation, 
	\begin{align*}
	E[\tilde Z_i^2]
	&=\left[\frac{1}{T^4}
	\sum_{t=1}^T \sum_{t'=1}^T \sum_{t''=1}^T \sum_{t'''=1}^T h_i(\gamma_t,\gamma_{t'})h_i(\gamma_{t''}\gamma_{t'''}) \right]-
	\left(\frac{1}{T^2}
	\sum_{t=1}^T \sum_{t'=1}^T E[h_i(\gamma_t,\gamma_{t'})]E[h_i(\gamma_{t''}\gamma_{t'''})]\right)^2
	\\
	&\le\frac{4}{T^4}
	\sum_{t=1}^T \sum_{t'=1}^T \sum_{t''=1}^T \sum_{t'''=1}^Th_i(\gamma_t,\gamma_{t'})h_i(\gamma_{t''}\gamma_{t'''}).
	\end{align*}
	To control the right hand side, we apply Lemma 2 in  \cite{yoshihara1976limiting} with $\delta=2(r+\delta)-2$, $\delta'=2r-2$, $r=2(r+\delta)$ (where $r$ and $\delta$ on the left-hand side are in terms of the notation by  \cite{yoshihara1976limiting}, and $r$ and $\delta$ on the right-hand side satisfy our Assumption \ref{a:NW_OLS}).  With this setting, note that
	\begin{align*}
	\frac{2+\delta'}{\delta'}=\frac{2r}{2(r-1)}<\frac{2r}{r-1}<\lambda,
	\end{align*}
	where $r$ here is from Assumption \ref{a:NW_OLS},
	and thus Assumption \ref{a:NW_OLS} (iii) implies that
	\begin{align*}
	\beta(\ell)=O(\ell^{-\lambda})=O(\ell^{-(2+\delta')/\delta'}),
	\end{align*}
	which is sufficient for the mixing requirement in the Lemma 2 of \cite{yoshihara1976limiting}.
	We now verify Condition (2.4) in \cite{yoshihara1976limiting}, which requires
	\begin{align*}
	E\left[\left|E[X_{it}U_{it} U_{it'} X_{it'}|\gamma_t ,\gamma_{t'}]\right|^{2(r+\delta)} \right]<\infty.
	\end{align*}
	This follows from our Assumption \ref{a:NW_OLS}(ii), Jensen's inequality, and Cauchy-Schwarz's inequality, as
	\begin{align*}
	E\left[\left|E[X_{it}U_{it} U_{it'} X_{it'}|\gamma_t ,\gamma_{t'}]\right|^{2(r+\delta)} \right] &\le E\left[|X_{it}U_{it} U_{it'} X_{it'}|^{2(r+\delta)} \right]\\
	& \le
	\left\{	E\left[|X_{it}U_{it}|^{4(r+\delta)} \right]\cdot E\left[|U_{it'} X_{it'}|^{4(r+\delta)} \right]\right\}^{1/2}<\infty.
	\end{align*}
	We next verify
	Condition (2.3) in \cite{yoshihara1976limiting}, which requires
	\begin{align*}
	\int \int \left|E[X_{it}U_{it} U_{it'} X_{it'}'|\gamma_t=u ,\gamma_{t'}=v]\right|^{2(r+\delta)} dF(u)dF(v)<\infty,
	\end{align*}
	where $F(\cdot)$ is the common CDF of $\gamma_i$.
	This holds since, by Cauchy-Schwarz and  Assumption \ref{a:NW_OLS}(i) 
	\begin{align*}
	\left|E[X_{it}U_{it} U_{it'} X_{it'}|\gamma_t=u ,\gamma_{t'}=v]\right|^2\le E[(X_{it}U_{it})^2|\gamma_t=u ]  E[ (U_{it'} X_{it'})^2|\gamma_{t'}=v],
	\end{align*}
	and thus by Jensen's inequality, 
	\begin{align*}
	&\int \int \left|E[X_{it}U_{it} U_{it'} X_{it'}|\gamma_t=u ,\gamma_{t'}=v]\right|^{2(r+\delta)} dF(u)dF(v)\\
	& \le\int \int \left|E[(X_{it}U_{it})^2|\gamma_t=u ]\cdot  E[ (U_{it'} X_{it'})^2|\gamma_{t'}=v]\right|^{(r+\delta)} dF(u)dF(v)\\
	& \le \int E[|X_{it}U_{it}|^{2(r+\delta)}|\gamma_t=u ]dF(u)\cdot  \int E[ |U_{it'} X_{it'}|^{2(r+\delta))}|\gamma_{t'}=v] dF(v)\\
	& = \left(E[E[|X_{it}U_{it}|^{2(r+\delta)}|\gamma_t ]]\right)^2=\left(E[|X_{it}U_{it}|^{2(r+\delta)}]\right)^2<
	\infty
	\end{align*}
	under Assumption \ref{a:NW_OLS}(ii).
	Applying Lemma 2 of \cite{yoshihara1976limiting} following Equation (10) in \cite{dehling2010central} now yields
	\begin{align*}
	E\left[\frac{4}{T^4}
	\sum_{t=1}^T \sum_{t'=1}^T \sum_{t''=1}^T \sum_{t'''=1}^T h_i(\gamma_t,\gamma_{t'})h_i(\gamma_{t''}\gamma_{t'''})\right]=O(T^{-1-\delta/(r-1)(r+\delta)})=o(1),
	\end{align*}
	where the last equality follows because $\delta/(r-1)(r+\delta)>0$.
	This result and Markov's inequality together imply $(9)=o_p(1)$ in \eqref{eq:proof_second_statement2}, and hence $(6) = o_p(1)$ in \eqref{eq:proof_second_statement2}. This completes the proof of \eqref{eq:lem:EW_CRVE2}.
\end{proof}

\subsection{Generalized Newey-West and Eicker-White CRVE under Degeneracy}

\begin{lemma}[Generalized Newey-West and Eicker-White CRVE under Degeneracy]\label{lem:EW_CRVE_degeneracy}
	If Assumption \ref{a:NW_OLS} holds with (iv)(2), then
	\begin{align}
	&\frac{1}{(NT)^2}\sum_{i=1}^N \sum_{t=1}^T X_{it}\hat U_{it} \hat U_{it} X_{it}'-\frac{var(X_{11}U_{11})}{NT}=o_p\left(\frac{1}{NT}\right),\label{eq:EW_CRVE_degeneracy1}\\
	&\frac{1}{(NT)^2}\sum_{i=1}^N \sum_{i'=1}^N\sum_{t=1}^T X_{it}\hat U_{it} \hat U_{i't} X_{i't}'\notag\\
	&\qquad+\sum_{m=1}^{M}\frac{w(m,M)}{(N T)^2}\sum_{t=m+1}^{T} \sum_{i=1}^N\sum_{i'\ne i}^N  \left(X_{it}\hat U_{it} \hat U_{i',t-m} X_{i',t-m}'+X_{i',t-m}\hat U_{i',t-m}\hat U_{it}X_{it}'  \right)\notag\\
	&\qquad-\frac{var(X_{11}U_{11})}{NT}=o_p\left(\frac{1}{T}\right),\label{eq:EW_CRVE_degeneracy2}\\
	\text{and}&\notag\\
	&\frac{1}{(NT)^2} \sum_{i=1}^N\sum_{t=1}^T \sum_{t'=1}^T X_{it}\hat U_{it} \hat U_{it'} X_{it'}'-\frac{var(X_{11}U_{11})}{NT}=o_p\left(\frac{1}{N}\right).\label{eq:EW_CRVE_degeneracy3}
	\end{align}
	
\end{lemma}

\begin{proof}
	The first statement \eqref{eq:EW_CRVE_degeneracy1} follows from the consistency of $\hat \beta$ implied by its asymptotic normality from the first part of Theorem \ref{thm:OLS} (which does not rely on the current lemma) and the law of large numbers for i.i.d. random variables.
	
	The second statement  \eqref{eq:EW_CRVE_degeneracy2} follows from
	\begin{align*}
	&\frac{1}{N^2T}\sum_{i=1}^N \sum_{i'=1}^N\sum_{t=1}^T X_{it}\hat U_{it}\hat U_{i't} X_{i't}'=
	\\
	&\underbrace{	 \left(\frac{N-1}{N}\right) \frac{1}{N(N-1)T}\sum_{i=1}^N \sum_{i'\ne i}^N\sum_{t=1}^T X_{it}\hat U_{it}\hat U_{i't} X_{i't}'}_{=O_p((N^2T)^{-1/2})}+\underbrace{\frac{1}{N^2T}\sum_{i=1}^N \sum_{t=1}^T X_{it}\hat U_{it}\hat U_{it}X_{it}'}_{=var(R_{11})/N+O_p(N^{-1}(NT)^{-1/2})},
	\end{align*}
	where the first term on the right-hand side is $o_p(1)$ by the law of large numbers for i.i.d. data, 
	and  
	the second term on the right-hand side is $o_p(1/N)$ by the first statement \eqref{eq:EW_CRVE_degeneracy1} .
	The consistency of $\hat \beta$ and
	the law of large numbers for i.i.d. data imply that
	\begin{align*}
	&\frac{1}{N^2 T}\sum_{t=m+1}^{T} \sum_{i=1}^N\sum_{i'\ne i}^N  X_{it}\hat U_{it}\hat U_{i',t-m}X_{i',t-m}'
	\\&=	\frac{1}{N^2 T}\sum_{t=m+1}^{T} \sum_{i=1}^N\sum_{i'\ne i}^N  X_{it} U_{it} U_{i',t-m}X_{i',t-m}' +O_p\left(\|\hat \beta- \beta\|\right)\\
	&=O_p\left(\frac{1}{\sqrt{N^2T}}\right) +O_p\left(\frac{1}{\sqrt{NT}}\right).
	\end{align*}
	Thus, by Assumption \ref{a:NW_OLS} (v)(vi), we have
	\begin{align*}
	\sum_{m=1}^{M}\frac{w(m,M)}{(N T)^2}\sum_{t=m+1}^{T} \sum_{i=1}^N\sum_{i'\ne i}^N  X_{it}\hat U_{it}\hat U_{i',t-m}X_{i',t-m}' =O_p\left(\frac{ M}{T\sqrt{NT}}\right)=o_p\left(\frac{1}{T}\right).
	\end{align*}
	Combining these probability limits together yields the second statement  \eqref{eq:EW_CRVE_degeneracy2} .
	
	The third statement \eqref{eq:EW_CRVE_degeneracy3} can be shown similarly to the first part of the second statement.
\end{proof}

\section{Heterogeneous Per-Cluster Numbers of Observations}\label{sec:multiple_observations}

The main results presented in Section \ref{sec:main} in the main text straightforwardly extends to a more general case with possibly zero or multiple observations per cluster intersection.
Suppose that the $(i,t)$-th cluster contains  $J_{it}$ units $\{X_{it1},\ldots,X_{itJ_{it}}\}$ of observations, where $J_{it}$ is considered a random variable.
As in the main text, let $\theta = E[X_{itj}]=0$ without loss of generality.
In this extended setting, the sample mean is defined by
\begin{align*}
\widetilde\vartheta 
=
\frac{1}{\sum_{i=1}^N \sum_{t=1}^T J_{it}} \sum_{i=1}^N \sum_{t=1}^T \sum_{j=1}^{J_{it}} X_{itj}.
\end{align*}

The standard approach to clustered data is to treat the cluster sum
$
\overline X_{it} = \sum_{j=1}^{J_{it}} X_{itj}
$
as the effective observation for the $(i,t)$-th unit.
Accordingly, define their projections
$\overline a_i = E[\overline X_{it}|\alpha_i]$ and
$\overline b_t = E[\overline X_{it}|\gamma_t]$.
Let their (long-run) variances be denoted by
$\overline\Sigma_a = E[\overline a_i \overline a_i']$ and
$\overline\Sigma_b = \sum_{\ell = -\infty}^\infty E[\overline b_t \overline b_{t+\ell}']$.
With
$
\widehat\mu_J = \frac{1}{NT}\sum_{i=1}^N \sum_{t=1}^T J_{it},
$
we now extend Assumptions \ref{a:NW} and \ref{a:non-degenerate} in the baseline model as follows.

\begin{assumption}\label{a:NW_multiple}
	(i) Assumption \ref{a:NW} holds with $X_{it}$ replaced by $\overline X_{it}$.
	(ii) Assumption \ref{a:non-degenerate} holds with $\Sigma_a$ and $\Sigma_b$ replaced by $\overline\Sigma_a$ and $\overline\Sigma_b$, respectively.
	(iii) $\widehat\mu_J \stackrel{p}{\to}\mu_J \in (0,\infty)$.
\end{assumption}
Note that more low level conditions on sampling formulated in terms of $J_{it}$ and $X_{itj}$ can be done following the approach taken in Assumption 1 in \cite{DDG2018}.
The following theorem states that the conclusions of Theorems \ref{thm:variance_mean} and \ref{thm:main_asymptotics} continue to hold under this generalized setting.

\begin{theorem}\label{thm:main_asymptotics_multiple}
	Suppose that Assumption \ref{a:NW_multiple} holds. Then,
	\begin{align*}
	\widehat \vartheta \stackrel{p}{\rightarrow} \theta
	\qquad\text{ and }\qquad
	var( \hat \vartheta )^{-1/2}(\hat \vartheta - \theta)\stackrel{d}{\to} N(0,I_m).
	\end{align*}
\end{theorem}

\begin{proof}
	Under Assumption \ref{a:NW_multiple} (i), it immediately follows that the conclusion of Theorem \ref{thm:variance_mean} holds with $\theta$, $\widehat\theta$, $a_i$, $b_t$, and $e_{it}$ replaced by 
	$\vartheta = E[\overline X_{it}]$,
	$\widehat\vartheta = \frac{1}{NT}\sum_{i=1}^N \sum_{t=1}^T \overline X_{it}$,
	$\overline a_i$,
	$\overline b_t$, and
	$\overline e_{it} = \overline X_{it} - \overline a_i - \overline b_t$, respectively.
	Since
	$
	\widetilde\vartheta 
	=
	\widehat\mu_J^{-1} \widehat\vartheta,
	$
	it follows that
	$\widehat \vartheta \stackrel{p}{\rightarrow} \theta$ as $N,T \rightarrow \infty$ under Assumption \ref{a:NW_multiple} (iii).
	
	Under Assumption \ref{a:NW_multiple} (i)--(ii), the conclusion of Theorem \ref{thm:main_asymptotics} holds with $\widehat\theta$ replaced by $\widehat\vartheta$.
	Since
	$
	\widetilde\vartheta 
	=
	\widehat\mu_J^{-1} \widehat\vartheta,
	$
	it follows that
	$
	var(\widetilde\vartheta)^{-1/2}(\widetilde\vartheta-\theta) \stackrel{d}{\rightarrow} N(0,I_m)
	$
	holds under Assumption \ref{a:NW_multiple} (iii).
\end{proof}

\section{Additional Information about Data Analyses}\label{sec:additional_details_data}

\subsection{Data used for Section \ref{sec:time_effects}}

For the analysis in Section \ref{sec:time_effects}, we use the data from \citet*{bloom2013identifying}.
This data set is publicly available as a supplementary material of \citet*{bloom2013identifying} from the Econometric Society.
We use two variables contained in spillovers.dta.
The log Tobin's average Q is available as the variable named lq.
The log R\&D stock divided by capital stock is available as the variable named grd\_k.

\subsection{Estimation of $\gamma$ for the Preliminary Analysis in Section \ref{sec:time_effects}}\label{sec:estimation_gammax_gammay}

In the preliminary analysis of the market value data to motivate our novel standard error formula, we estimate $\hat\gamma_t$ 
for each $t$ and its partial autocorrelation coefficient in Section \ref{sec:time_effects}.
The current appendix section presents a concrete estimation procedure used to obtain the estimates in Section \ref{sec:time_effects}.

To eliminate the firm effect $\alpha_i$, we first apply the within transformation of $Y_{it}$ and obtain $\ddot{Y}_{it} = \gamma_t - \bar\gamma + \varepsilon_i - \bar\varepsilon_i$ for each $(i,t)$, where $\bar\gamma = T^{-1}\sum_{t=1}^T \gamma_t$ and $\bar\varepsilon_i = T^{-1} \sum_{t=1}^T \varepsilon_{it}$.
We can then estimate $\gamma_t$ up to location $\bar\gamma$ by $\hat\gamma_t = N^{-1}\sum_{i=1}^N \ddot{Y}_{it}$ for each $t$.
The partial autocorrelation of $\{\gamma_t\}_t$ in can be also estimated by running the first-order autoregression of $\ddot{Y}_{it}$.
We report the HC0 standard error in Section \ref{sec:time_effects} because our robust standard error formula is yet to be introduced as of Section \ref{sec:time_effects}.
The autocorrelogram is produced based on the series $\{\hat\gamma_t\}_{t=1}^T$.


\subsection{Data used for Section \ref{sec:application}}\label{sec:data_asset_pricing}

For the analysis in Section \ref{sec:application}, we use the data from \citet*{gagliardini2016time}.
This data set is publicly available as a supplementary material of \citet*{gagliardini2016time} from the Econometric Society.
We combined data from multiple files located in the folder named GOS\_dataCodes\_paper.
The returns of the 44 industry portfolios are available in Wspace\_44Indu.mat,
the returns of the 9936 individual stocks are available in Wspace\_CRSPCMST\_ret.mat, 
the Fama-French factors are available in Wspace\_Fact.mat, and 
the risk-free rates (monthly 30-day T-bill yields) are available in RiskFree.mat. 

\section{Additional Simulations}\label{sec:additional_simulations}

\subsection{Power}\label{sec:simulations_power}

The baseline simulation studies presented in Section \ref{sec:simulations} focuses on the coverage probabilities.
In this section, we present additional simulation studies focusing on the power.

We continue to use the same simulation designs as in Section \ref{sec:simulations}.
Instead of computing the coverage probabilities, however, we now compute the rejection probabilities for the hypothesis $H_0: \beta_1=b$ for various values of $b \in [0.5,1.5]$ with the nominal size of 5\%. 
Recall that the true value is $\beta_1=1$.
We use the sample size of $N=T=50$ throughout, and run 10,000 Monte Carlo iterations for each set of simulations.

Figure \ref{fig:power} illustrates the power curves for EHW, CR$i$, CR$t$, CGM, MNW, M, T, and CHS.
Panel (A) illustrates power curves under the i.i.d. design. Panels (B), (C), and (D) illustrate power curves under the dependence designs with $\rho=0.25$, 0.50, and 0.75, respectively.

\begin{figure}
	\centering
	\begin{tabular}{cc}
		(A) I.I.D. Design
		&
		(B) Dependence Design with $\rho=0.25$\\
		\includegraphics[width=0.45\textwidth]{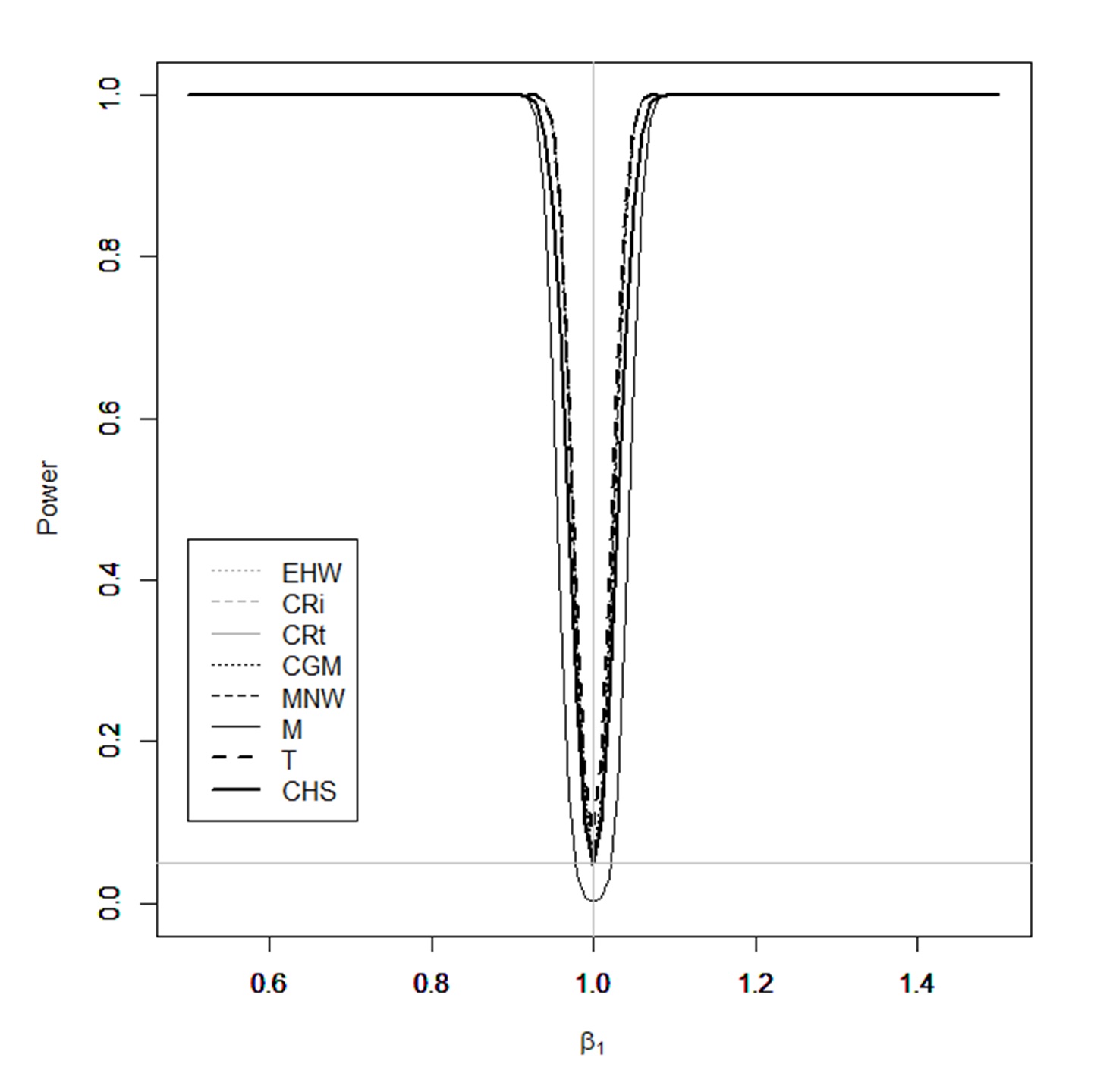}
		&
		\includegraphics[width=0.45\textwidth]{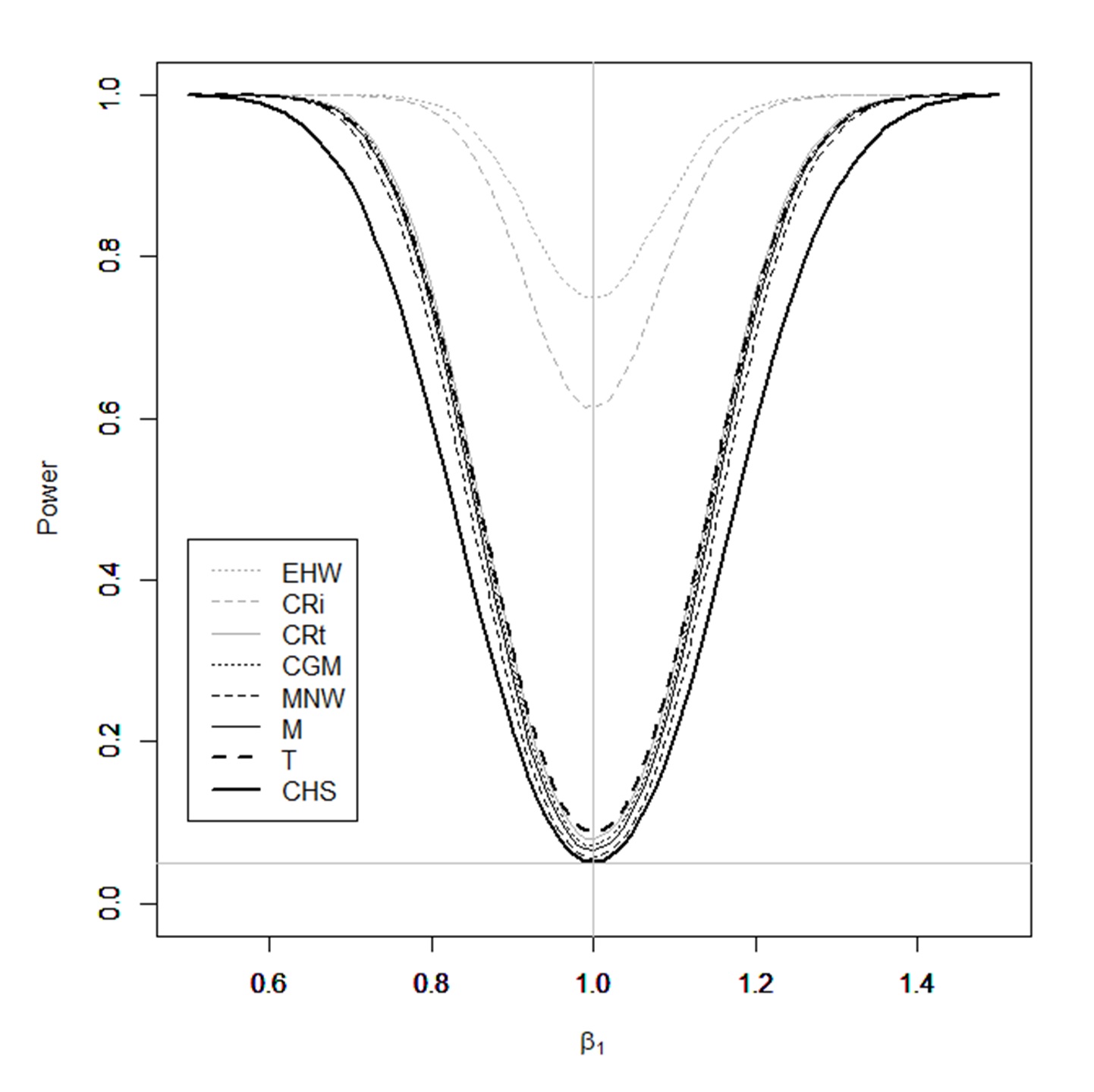}
		\\
		(C) Dependence Design with $\rho=0.50$
		&
		(D) Dependence Design with $\rho=0.75$\\
		\includegraphics[width=0.45\textwidth]{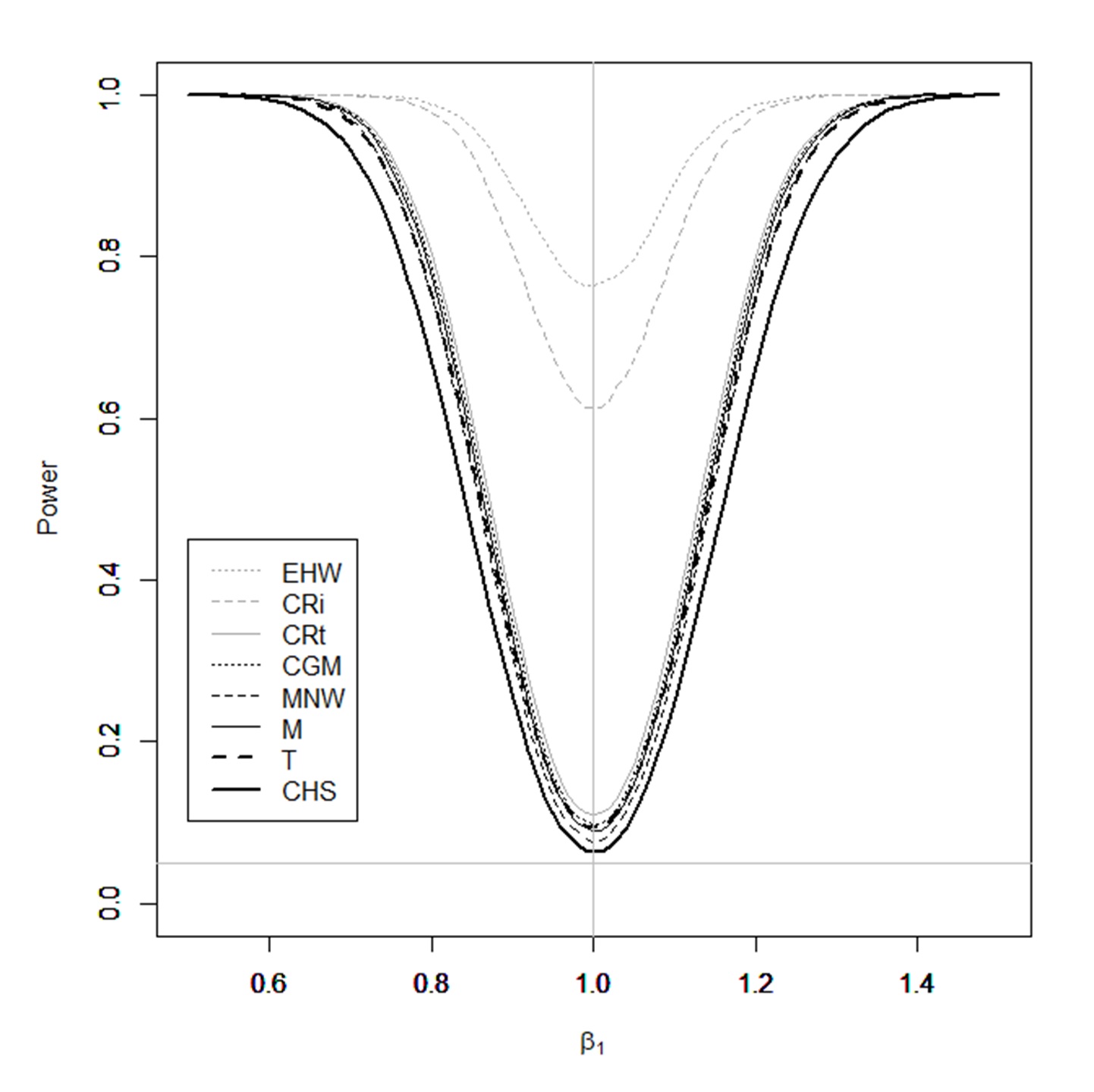}
		&
		\includegraphics[width=0.45\textwidth]{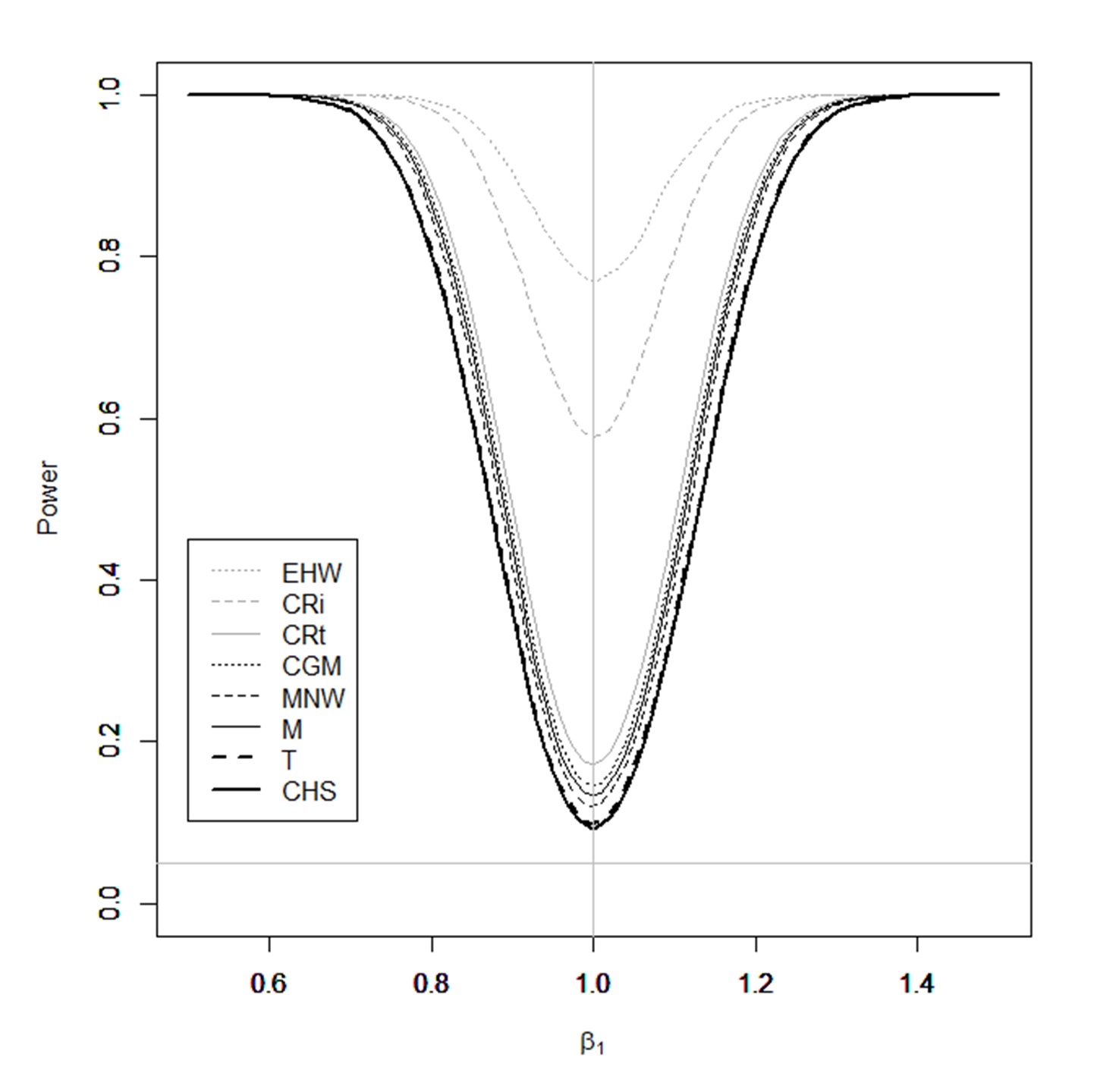}
	\end{tabular}
	\caption{Power curves for EHW, CR$i$, CR$t$, CGM, MNW, M, T, and CHS based on 10,000 Monte Carlo iterations. The true parameter value is $\beta_1=1$ and the nominal size is 5\%. Panel (A) illustrates power curves under the i.i.d. design. Panels (B), (C), and (D) illustrate power curves under the dependence designs with $\rho=0.25$, 0.50, and 0.75, respectively. The sample size is set to $N=T=75$ throughout.}${}$\\${}$
	\label{fig:power}
\end{figure}

The sizes are complementary to the coverage probabilities reported in Section \ref{sec:simulations}.
As the hypothesized value of $\beta_1$ deviates away from the true value $\beta_1=1$, the power increases for each method.
Some methods show higher power than the others, but only at the expense of size distortions.

\subsection{Simulations for the Two-Way Fixed-Effect Estimator}\label{sec:simulations_fixed_effect}

The baseline simulation studies presented in Section \ref{sec:simulations} focuses on the OLS.
In this section, we present additional simulation studies focusing on the two-way fixed-effect estimator for fixed-effect models.

Motivated by the example model \eqref{eq:FE_example} in Section \ref{sec:fixed_effect_necessary}, consider the following data generating process.
\begin{align*}
Y_{it} &= \beta_0 + \beta_1 X_{it} + U_{it},
\end{align*}
where the right-hand side variables $(X_{it},U_{it})'$ are generated through the panel dependence structure
\begin{align*}
X_{it} &= w_1 \alpha_{i1} \gamma_{t2} + w_2 \alpha_{i2} \gamma_{t1} + w_3 \varepsilon_{it0}
\qquad\text{and}\\
U_{it} &= w_4 \alpha_{i0} + w_5 \gamma_{t0} + w_6 \alpha_{i1} \gamma_{t3} + w_7 \alpha_{i3} \gamma_{t1} + w_8 \varepsilon_{it1}.
\end{align*}
For the weight parameters, we use $(w_1,w_2,w_3,w_4,w_5,w_6,w_7,w_8) = (0,0,1,0,0,0,0,1)$ to generate i.i.d. data and also use $(w_1,w_2,w_3,w_4,w_5,w_6,w_7,w_8) = (0.25,0.25,1.00,0.25,0.25,0.25,0.25,1.00)$ to generate dependent data.
The latent components $(\alpha_{i0},\alpha_{i1},\alpha_{i2},\alpha_{i3},\varepsilon_{it0},\varepsilon_{it1})$ are all mutually independent $N(0,1)$.

Similarly to the baseline simulation design presented in Section \ref{sec:simulations}, the latent common time effects $(\gamma_{t0},\gamma_{t1},\gamma_{t2},\gamma_{t3})$ are dynamically generated according to the AR(1) design:
\begin{align*}
&
\gamma_{t0} = \rho\gamma_{(t-1)0} + \tilde\gamma_{t0}
\text{ where $\tilde\gamma_{t0}$ are independent draws from $N(0,1-\rho^2)$; }
\\
&
\gamma_{t1} = \rho\gamma_{(t-1)1} + \tilde\gamma_{t1}
\text{ where $\tilde\gamma_{t1}$ are independent draws from $N(0,1-\rho^2)$; }
\\
&
\gamma_{t2} = \rho\gamma_{(t-1)2} + \tilde\gamma_{t2}
\text{ where $\tilde\gamma_{t2}$ are independent draws from $N(0,1-\rho^2)$; and}
\\
&
\gamma_{t3} = \rho\gamma_{(t-1)3} + \tilde\gamma_{t3}
\text{ where $\tilde\gamma_{t3}$ are independent draws from $N(0,1-\rho^2)$.}
\end{align*}
The initial values are drawn from $N(0,1)$.
We vary the AR coefficient $\rho \in \{0.25,0.50,0.75\}$ across sets of simulations.

For each realization of observed data $\{(Y_{it},X_{it}) : 1 \le i \le N, 1 \le t \le T\}$ constructed according to the data generating process described above, we estimate $\beta_1$ by the two-way fixed-effect estimator and compute its standard error as in Section \ref{sec:fixed_effect_theory}.
As in the baseline simulation studies, we compare our estimator CHS with EHW, CR$i$, CR$t$, CGM, MNW, M and T;
see Section \ref{sec:simulations} for details.

Table \ref{tab:sim_fixed_effect} reports simulation results.
Reported values are the coverage frequencies for the slope parameter $\beta_1$ for the nominal probability of 95\% based on 10,000 Monte Carlo iterations. 
The top and bottom panels show coverage probability results under the i.i.d. design and the dependence design, respectively.
In each group of three consecutive rows, the panel sample sizes $(N,T)$ vary by rows.
Cells are shaded based on the proximity of the simulated coverage probability to the nominal probability of 0.95; the darker shades indicate more correct coverage.

\begin{table}
	\renewcommand{\arraystretch}{0.975}
	\centering
	\begin{tabular}{lcccccccccccc}
		\multicolumn{12}{c}{I.I.D. Design: Nominal Probability = 95\%}\\
		\hline\hline
		& $N$ & $T$ & $\rho$ & EHW & CR$i$ & CR$t$ & CGM & MNW & M & T & CHS\\
		\hline
		(I)    & 50 & 100 & --- & \one 0.945 & \two 0.936 & \one 0.942 & \two 0.934 & \one 0.947 & \fiv 0.999 & \fou 0.915 & \one 0.950\\
		(II)   & 75 & 75 & --- & \one 0.949 & \one 0.945 & \one 0.946 & \two 0.939 & \one 0.952 & \fiv 0.999 & \fou 0.911 & \one 0.954\\
		(III)  & 100 & 50 & --- & \one 0.946 & \one 0.943 & \two 0.939 & \two 0.937 & \one 0.948 & \fiv 0.999 & \six 0.890 & \one 0.948\\
		\hline\hline
		&&&& \hspace{1.12cm} & \hspace{1.12cm} & \hspace{1.12cm} & \hspace{1.12cm} & \hspace{1.12cm} & \hspace{1.12cm} & \hspace{1.12cm} & \hspace{1.12cm} \ \\
		\multicolumn{12}{c}{Dependence Design: Nominal Probability = 95\%}\\
		\hline\hline
		& $N$ & $T$ & $\rho$ & EHW & CR$i$ & CR$t$ & CGM & MNW & M & T & CHS\\
		\hline
		(IV)    & 50 & 100 & 0.25 & \sev 0.888 & \fou 0.919 & \fiv 0.909 & \two 0.935 & \one 0.948 & \fiv 0.997 & \thr 0.925 & \one 0.952\\
		(V)   & 75 & 75 & 0.25 & \six 0.892 & \fou 0.918 & \fou 0.918 & \one 0.940 & \one 0.949 & \fiv 0.995 & \thr 0.927 & \one 0.954\\
		(VI)  & 100 & 50 & 0.25 & \six 0.891 & \fou 0.911 & \thr 0.924 & \two 0.938 & \one 0.951 & \fiv 0.996 & \fou 0.911 & \one 0.953\\
		\hline
		(VII)    & 50 & 100 & 0.50 & \sev 0.880 & \thr 0.912 & \fiv 0.903 & \thr 0.928 & \one 0.942 & \fiv 0.995 & \thr 0.925 & \one 0.951\\
		(VIII)   & 75 & 75 & 0.50 & \sev 0.883 & \fou 0.910 & \fiv 0.909 & \two 0.932 & \one 0.942 & \fiv 0.994 & \thr 0.925 & \one 0.951\\
		(IX)  & 100 & 50 & 0.50 & \eig 0.877 & \six 0.895 & \fou 0.911 & \thr 0.924 & \two 0.939 & \fiv 0.995 & \fiv 0.906 & \one 0.949\\
		\hline
		(X)    & 50 & 100 & 0.75 & 0.857 & \six 0.894 & \eig 0.879 & \fou 0.910 & \thr 0.925 & \fiv 0.992 & \fou 0.919 & \one 0.940\\
		(XI)   & 75 & 75 & 0.75 & 0.848 & \sev 0.884 & \eig 0.877 & \fiv 0.907 & \thr 0.921 & \fou 0.989 & \fou 0.918 & \two 0.935\\
		(XII)  & 100 & 50 & 0.75 & 0.842 & \nin 0.863 & \eig 0.872 & \six 0.890 & \fiv 0.907 & \thr 0.986 & \six 0.892 & \thr 0.923\\
		\hline\hline
	\end{tabular}
	\caption{Coverage probabilities for the slope parameter $\beta_1$ for the two-way fixed-effect estimator with the nominal probability of 95\% based on 10,000 Monte Carlo iterations. The top and bottom panels show results under the i.i.d. and dependence designs, respectively. The sample size is indicated by $(N,T)$. The parameter $\rho$ indicates the AR coefficient in the dependence design. EHW stands for Eicker–Huber–White, CR$i$ stands for cluster robust within $i$, CR$t$ stands for cluster robust within $t$, CGM stands for Cameron-Gelbach-Miller, MNW stands for MacKinnon-Nielsen-Webb, M stands for Menzel, T stands for Thompson, and CHS stands for Chiang-Hansen-Sasaki.\\${}$}
	\label{tab:sim_fixed_effect}
\end{table}

Observe that we have similar qualitative patterns in these results to those presented for the baseline simulation studies presented in Section \ref{sec:simulations}.

\bibliography{biblio}
\end{document}